\documentclass[12pt]{article}

\usepackage{pdfpages}
\usepackage{color}
\usepackage{diagbox}
\usepackage{amsthm, amsmath, amsfonts,graphicx,xr,url}
\usepackage{amssymb}
\usepackage{enumitem}
\usepackage{float}
\usepackage{graphicx}
\usepackage{bm, bbm}
\usepackage{blkarray}
\usepackage{multirow}
\usepackage{changepage}
\usepackage{authblk}
\usepackage{titlesec}
\usepackage[font=small,labelfont=bf]{caption}
\usepackage[font=small,labelfont=bf]{subcaption}
\numberwithin{equation}{section}
\usepackage{hyperref}
\hypersetup{
	colorlinks,
	citecolor=blue,
	linkcolor=cornflowerblue,
	filecolor=yaleblue,      
	urlcolor=cornflowerblue
}
\usepackage{cleveref} 
\usepackage{footnotebackref}
\crefformat{footnote}{#2\footnotemark[#1]#3}
\usepackage{easyReview}
\usepackage{comment}
\allowdisplaybreaks
\usepackage{xr}
\makeatletter

\usepackage{natbib}
\usepackage{pifont}
\DeclareMathOperator*{\Var}{Var}

\DeclareMathOperator*{\Pois}{Poisson}

\DeclareMathOperator*{\Ga}{Gamma}
\DeclareMathOperator*{\tGa}{Truncated\;Gamma}
\DeclareMathOperator*{\IG}{\mathcal{IG}}

\DeclareMathOperator*{\Bessel}{Bessel}

\DeclareMathOperator*{\ind}{ind}
\DeclareMathOperator*{\supth}{th}

\DeclareMathOperator*{\mean}{mean}

\def\NN{\mathbb{N}}
\def\RR{\mathbb{R}}
\def\T{{\mathrm{\scriptscriptstyle T}}}

\def\bfV{\mathbf{V}}
\def\bfone{\mathbf{1}}

\newtheorem{theorem}{Theorem}

\newtheorem{corollary}{Corollary}

\definecolor{cornflowerblue}{rgb}{0.39, 0.58, 0.93}
\definecolor{yaleblue}{rgb}{0.06, 0.3, 0.57}

\newcommand*{\addFileDependency}[1]{
	\typeout{(#1)}
	\@addtofilelist{#1}
	\IfFileExists{#1}{}{\typeout{No file #1.}}
}\makeatother

\begin{document}
	
	\def\spacingset#1{\renewcommand{\baselinestretch}%
		{#1}\small\normalsize} \spacingset{1}	
	
	\title{\textbf{Effective Bayesian Modeling of Large Spatiotemporal Count Data Using Autoregressive Gamma Processes}}
	\author{\small Yifan Cheng \thanks{y.cheng@u.nus.edu}  }
	\author{\small Cheng Li \thanks{stalic@nus.edu.sg} }
	\affil{\footnotesize Department of Statistics and Data Science, National University of Singapore}
	\date{}
    \maketitle
	\begin{abstract}
		We put forward a new Bayesian modeling strategy for spatiotemporal count data that enables efficient posterior sampling. Most previous models for such data decompose logarithms of the response Poisson rates into fixed effects and spatial random effects, where the latter is typically assumed to follow a latent Gaussian process, the conditional autoregressive model, or the intrinsic conditional autoregressive model. Since log-Gaussian is not conjugate to Poisson, such implementations must resort to either approximation methods like INLA or Metropolis moves on latent states in MCMC algorithms for model fitting and exhibit several approximation and posterior sampling challenges. Instead of modeling logarithms of spatiotemporal frailties jointly as a Gaussian process, we construct a spatiotemporal autoregressive gamma process guaranteed stationary across the time dimension. We decompose latent Poisson variables to permit fully conjugate Gibbs sampling of spatiotemporal frailties and design a sparse spatial dependence structure to get a linear computational complexity that facilitates efficient posterior computation. Our model permits convenient Bayesian predictive machinery based on posterior samples that delivers satisfactory performance in predicting at new spatial locations and time intervals. We have performed extensive simulation experiments and real data analyses, which corroborated our model's accurate parameter estimation, model fitting, and out-of-sample prediction capabilities.
	\end{abstract}
	\noindent%
	\newpage
	\tableofcontents
	
	\newpage
	\spacingset{1.75} 
	
	\section{Introduction} \label{sec:intro}
	Spatial and spatiotemporal count data are ubiquitous in many scientific fields, such as ecology \citep{Estetal18}, epidemiology \citep{Thoetal23}, and transcriptomics \citep{Zhaetal21}. Bayesian models for such data enable flexible hierarchical structures, convenient posterior sampling, and robust uncertainty quantification. A classic statistical model for spatial and spatiotemporal count data is the latent Gaussian model (LGP, \citealt{Digetal98,ChrWaa04}), where the count data is assumed to follow a standard generalized linear model such as the log-linear model, whose log-transformed parameters are specified as a linear combination of fixed effects and spatial random effects. For point-referenced spatial data, the spatial random effects are typically assumed to follow a continuous latent Gaussian process, characterizing the spatial or spatiotemporal dependence and leading to a well-defined stochastic process over continuous domains. For areal data, the spatial random effects are typically assumed to follow the conditional autoregressive model (CAR) or the intrinsic conditional autoregressive model (ICAR, \citealt{Bes74,Mar88,Besag1991,BesKoo95}). Computationally, both LGP and CAR models are often fit by either the integrated nested Laplace approximation (INLA, \citealt{Rueetal09, Blangiardo2013}) or Markov chain Monte Carlo (MCMC) algorithms that require Metropolis moves for latent states \citep{Lee2018}. For example, a recent survey on the spatial and spatiotemporal analysis of
	COVID-19 epidemiology reveals that almost all statistical models adopted in practice belong to either the LGP or ICAR class of models and used either INLA or MCMC for model fitting \citep{Nazetal22}.
	
	While the LGP and CAR models have achieved considerable empirical success in statistical applications, there exist several potential drawbacks in current statistical models and computational strategies. First, modeling the spatial random effects as Gaussian processes implies that the rate parameter of a Poisson response variable follows a log-Gaussian process, which is not conjugate to the Poisson distribution. Therefore, to compute the posterior distribution, one has to rely on either approximation methods such as INLA or Metropolis moves on latent states in an MCMC algorithm. INLA typically depends on tessellation of the spatial domain and incurs an approximation error that does not necessarily vanish with large samples. For MCMC algorithms, the dimension of latent states gets large as the number of spatial locations becomes large, and thus, the Metropolis moves may result in slow mixing of posterior chains and low effective sample sizes. Second, the log-Gaussian process is known to lead to potentially large gradients of log-likelihoods with respect to the latent states and model parameters in the covariance functions \citep{ChrWaa04}, which makes exploration of the full posterior difficult and requires numerical truncation. Third, realizations of an LGP model often demonstrate the whitening effect \citep{Moretal24}, due to implicit restrictions on the correlation structure caused by the nugget for spatial count data \citep{DeO13}. The literature offers many alternative strategies to overcome the aforementioned modeling and computational challenges. For instance, \citet{DeO13,DeO14} studied several classes of transformed covariance functions to remove discontinuity in the covariance function caused by the nugget and reduce the whitening effect. More recently, new efficient MCMC algorithms based on the P\'olya-Gamma sampler have been proposed in \citet{DAnCan23} for the Poisson model and in \citet{Banetal21} for spatial negative binomial model. \citet{Braetal18} and \citet{BraCli25} introduced new classes of multivariate log-gamma distributions as a fully conjugate prior of the Poisson model for spatial count data, which results in a fully conjugate posterior that does not require MCMC sampling. 
	
	This work focuses on Bayesian modeling of spatiotemporal count data and proposes a new modeling strategy that enables efficient and fully conjugate Gibbs sampling of spatiotemporal frailties and model parameters. Instead of assuming a log-Gaussian process on the spatial random effects as in most of the previous literature, we take the Poisson frailties associated with all spatial locations corresponding to each time interval as a vector and generalize scalar autoregressive gamma processes for time series (\citealt{Gourieroux2006}, \citealt{Creal2017}) to model the spatiotemporal frailties, thus extending prior applications of autoregressive gamma processes from the time series setting and the time-to-event network setting \citep{Zhaetal25} to the spatiotemporal setting. We make four main contributions, as follows. First, we establish a spatiotemporal autoregressive gamma process provably stationary across the time dimension. Different from LGP and CAR, our construction does not involve any Gaussian latent variables. Second, we adopt a novel decomposition of latent Poisson variables that allows fully conjugate Gibbs sampling of the frailties and scalar parameters $c,\,\kappa,\,\rho$ in \Cref{sec:model}. Third, we design a spatial dependence structure built on a sparse neighborhood graph that leads to a computational complexity linear in both the number of spatial locations and the number of time intervals in each posterior sampling iteration, thus facilitating efficient posterior computation. Fourth, we develop a complete set of Bayesian predictive machinery for new points in both spatial and temporal dimensions, based on posterior samples of spatiotemporal frailties and model parameters. 
	
	The rest of this paper is organized as follows. \Cref{sec:model} introduces our Bayesian spatiotemporal count model with prior specification and a rigorous theory on the stationarity condition for the spatiotemporal frailty process. \Cref{sec:sampling} details the posterior sampling steps for all model parameters. \Cref{sec:prediction} elaborates on procedures for predicting at future time intervals and unseen spatial locations. \Cref{sec:simulation} conducts extensive simulation experiments that justify our model's accurate parameter estimation, model fitting, and prediction capabilities on par with alternative methods. \Cref{sec:covid} analyzes world weekly COVID-19 cases and deaths data using our framework and competitive spatiotemporal count models, validating our model's effectiveness and satisfactory performance in both in-sample and out-of-sample predictions. \Cref{sec:discussion} concludes with pointers towards potential future work. Finally, the supplementary material provides deduction details of the MCMC sampler in \Cref{sec:sampling}, presents a theorem that justifies stationarity of the spatiotemporal frailty process for the third simulation group in \Cref{sec:simulation}, and displays complementary results for simulation experiments and the COVID-19 real data example in \Cref{sec:simulation,sec:covid}.
	
	\section{A Bayesian Spatiotemporal Model for Count Data} \label{sec:model}
	
	Let $m, T\in\mathbb{N}$ be the numbers of spatial locations and time intervals of equal length, respectively. For each $(i,t)\in\{1,\ldots,m\}\times\{1,\ldots,T\}$, we denote $y_t\left(\bm{s}_i\right)\in\NN$ as our observed counts at the spatial location $\bm{s}_i$ during the $t^{\supth}$ time interval. Then $\{y_t\left(\bm{s}_i\right):t=1,\ldots,T,\;i=1,\ldots,m\}$ consists of all the observed count data. We assume
	\begin{equation}\label{yPois}
		y_t\left(\bm{s}_i\right)~|~ U^t(\bm{s}_i), \bm{\beta},\bm{x}_t\left(\bm{s}_i\right)\overset{\ind}{\sim}\Pois\left(U^t(\bm{s}_i) \exp\{\bm{x}_t\left(\bm{s}_i\right)^{\T}\bm{\beta}\}\right), 
	\end{equation}
	where $\Pois(\lambda)$ denotes the Poisson distribution with rate parameter $\lambda$, $U^t(\bm{s}_i)$ is a positive spatiotemporal frailty term, $\bm{x}_t\left(\bm{s}_i\right) \in \RR^p$ is a vector of observed spatiotemporal covariates, and $\bm{\beta} \in \RR^p$ is the vector of regression coefficients. 
	
	Our Bayesian modeling strategy differs from the prior literature primarily in the model for frailties $U^{1:T}(\bm{s}_{1:m})=\left\{U^t(\bm{s}_i):~ t=1,\ldots,T,\;i=1,\ldots,m\right\}$. Instead of modeling their logarithms jointly as a Gaussian process, as in the classic latent Gaussian model \citep{ChrWaa04}, we construct a spatiotemporal autoregressive gamma process for $U^{1:T}(\bm{s}_{1:m})$, extending previous applications of such processes in the time series setting (\citealt{Gourieroux2006}, \citealt{Creal2017}) and the time-to-event network setting \citep{Zhaetal25}. Specifically, for each $t=2,\ldots,T$ and $\bm{s}_i$, we assume that each $U^{t}(\bm{s}_i)$ is associated with a latent space $Z^{t-1}(\bm{s}_i)$, such that
	\begin{align} \label{eq:ARG1}
		& U^{t}(\bm{s}_i) ~|~ Z^{t-1}(\bm{s}_i) \overset{\ind}{\sim} \Ga\left(\alpha + Z^{t-1}(\bm{s}_i), \frac{1}{c} \right), \nonumber \\
		& Z^{t-1}(\bm{s}_i) ~|~  \left\{U^{t-1}(\bm{s}_j):~ j=1,\ldots,m\right\} \overset{\ind}{\sim} \Pois \left(\frac{1}{c} \sum_{j=1}^m v_{ij} U^{t-1}(\bm{s}_j)\right) 
	\end{align}
	for some hyperparameter $\alpha>1$, parameter $c>0$, and nonnegative weights $\bfV=\{v_{ij}:i,j=1,\ldots,m\}$ such that $\sum_{j=1}^m v_{ij}>0$ for all $i$, where $\overset{\ind}{\sim}$ indicates the independence of all concerned distributions, and the $\Ga(a,b)$ distribution has density $f(x;a,b)=\frac{b^a}{\Gamma(a)}x^{a-1}\exp(-bx)$ for $x>0$. As a result of the conditional distributions in \eqref{eq:ARG1}, for a given set of $\alpha,c,\bfV$, the conditional distribution of $U^{t}(\bm{s}_i)$ given 1-lag values $U^{t-1}(\bm{s}_{1:m})$ is a noncentral gamma distribution with density
	\begin{align} \label{eq:nc.gamma1}
		&p\left(U^{t}(\bm{s}_i) ~|~ U^{t-1}(\bm{s}_{1:m})\right) = \left(\frac{U^{t}(\bm{s}_i)}{\sum_{j=1}^m v_{ij} U^{t-1}(\bm{s}_j)}\right)^{\frac{\alpha-1}{2}} \times \nonumber \\
		&\qquad \frac{1}{c}\exp\left(-\frac{U^{t}(\bm{s}_i)+\sum_{j=1}^m v_{ij} U^{t-1}(\bm{s}_j)}{c}\right) \times I_{\alpha-1} \left(2\sqrt{\frac{U^{t}(\bm{s}_i)\sum_{j=1}^m v_{ij} U^{t-1}(\bm{s}_j)}{c^2}}\right),
	\end{align}
	where $I_{\alpha-1}(\cdot)$ is the modified Bessel function of the first kind. From \eqref{eq:ARG1}, by the law of iterated expectations and the law of total variance, the conditional mean and variance of the vector $U^{t}(\bm{s}_{1:m})$ given $U^{t-1}(\bm{s}_{1:m})$ are
	\begin{align} \label{eq:cond.mean.var1}
		& \mathbb{E}\left[U^{t}(\bm{s}_{1:m}) ~|~ U^{t-1}(\bm{s}_{1:m})\right]  = \bfone_m \alpha c + \bfV U^{t-1}(\bm{s}_{1:m}) \text{ and} \nonumber \\
		& \Var\left[U^{t}(\bm{s}_{1:m}) ~|~ U^{t-1}(\bm{s}_{1:m})\right] = \bfone_m \alpha c^2 + 2c\bfV U^{t-1}(\bm{s}_{1:m}) ,
	\end{align}
	respectively, where $\bfone_m = (1,\ldots,1)^\T \in \RR^m$.
	
	The process defined by \eqref{eq:ARG1} along the temporal dimension is a vector generalization of the original autoregressive gamma process for time series, as in \citet{Gourieroux2006} and \citet{Creal2017}. We can show the following result regarding the stationarity condition for this process.
	
	\begin{theorem} \label{thm:stationary1}
		Suppose that for any fixed $m\in \NN$, $\bfV$ is an $m\times m$ matrix whose entries are all nonnegative with $\sum_{j=1}^m v_{ij}>0$ for all $i=1,\ldots,m$. Assume $\bfV$ satisfies at least one of the two following conditions.
		\begin{itemize}
			\item \textbf{Condition 1:} $\sum_{i=1}^m v_{ij}\leq 1$ for all $j=1,\ldots,m$;  
			\item \textbf{Condition 2:} $\sum_{j=1}^m v_{ij}\leq 1$ for all $i=1,\ldots,m$. 
		\end{itemize}
		Then the process $\{U^t(\bm{s}_{1:m}):t=1,2,\ldots,T\}$ defined in \eqref{eq:ARG1} is a stationary process with an invariant distribution as $T\to\infty$.
	\end{theorem}
	
	\begin{proof}[Proof of \Cref{thm:stationary1}]
		We show the stationarity condition using the theory on Laplace transform of compound autoregressive processes in \citet{Daretal06}. We first verify that for a fixed $m$ and given spatial locations $\{\bm{s}_1,\ldots,\bm{s}_m\}$, the vector process $\{U^t(\bm{s}_{1:m}):t=1,2,\ldots\}$ is a temporal compound autoregressive process of order 1. For any vector $\bm{x}=(x_1,\ldots,x_m)^\T \in \RR^{m}$ such that $x_i\geq 0$ for all $i=1,\ldots,m$, we have 
		\begin{align} \label{eq:CAR1.1}
			&~\quad \mathbb{E}\left[\exp\left\{-\bm{x}^\T U^t(\bm{s}_{1:m})\right\}  ~|~ U^{t-1}(\bm{s}_{1:m}) \right] \nonumber \\
			& = \mathbb{E}\left(\mathbb{E}\left[\exp\left\{-\bm{x}^\T U^t(\bm{s}_{1:m})\right\} ~|~ Z^{t-1}(\bm{s}_{1:m})\right] ~|~ U^{t-1}(\bm{s}_{1:m})  \right) \nonumber \\
			&= \mathbb{E}\left( \prod_{i=1}^m \mathbb{E}\left[\exp\left\{-x_i U^t(\bm{s}_i)\right\} ~|~ Z^{t-1}(\bm{s}_i)\right] ~\Big|~ U^{t-1}(\bm{s}_{1:m})  \right) \nonumber \\
			&= \mathbb{E}\left( \prod_{i=1}^m \frac{1}{(cx_i+1)^{\alpha+Z^{t-1}(\bm{s}_i)}} ~\Big|~ U^{t-1}(\bm{s}_{1:m})  \right) \nonumber \\
			&= \prod_{i=1}^m\mathbb{E}\left(\frac{1}{(cx_i+1)^{\alpha+Z^{t-1}(\bm{s}_i)}} ~\Big|~ U^{t-1}(\bm{s}_{1:m})  \right) \nonumber \\
			&= \prod_{i=1}^m \sum_{k=0}^{\infty} \frac{\left(\frac{1}{c} \sum_{j=1}^m v_{ij} U^{t-1}(\bm{s}_j)\right)^k}{k!} \exp\left\{-\frac{1}{c} \sum_{j=1}^m v_{ij} U^{t-1}(\bm{s}_j)\right\} \cdot \frac{1}{(cx_i+1)^{\alpha+k}}   \nonumber \\
			&= \prod_{i=1}^m \frac{1}{(cx_i+1)^{\alpha}}\times \exp\nonumber \left\{-\frac{1}{c} \sum_{j=1}^m v_{ij} U^{t-1}(\bm{s}_j)\left[1 - \frac{1}{cx_i+1}\right]\right\}\\
			&= \exp\left\{-\sum_{j=1}^m \left(\sum_{i=1}^m \frac{x_iv_{ij}}{cx_i+1} \right) U^{t-1}(\bm{s}_j) -\alpha\sum_{i=1}^m \log (cx_i+1)\right\}.
		\end{align}
		If we define vector-valued function $a(\bm{x}) = \left(\sum_{i=1}^m \frac{x_iv_{i1}}{cx_i+1},\ldots,\sum_{i=1}^m \frac{x_iv_{im}}{cx_i+1}\right)^\T \in \RR^m$ and scalar-valued function $b(\bm{x}) = -\alpha\sum_{i=1}^m \log (cx_i+1) \in \RR$, then \eqref{eq:CAR1.1} implies 
		\begin{align*}
			\mathbb{E}\left[\exp\left\{-\bm{x}^\T U^t(\bm{s}_{1:m})\right\}  ~|~ U^{t-1}(\bm{s}_{1:m}) \right] &= \exp\left\{- a(\bm{x})^\T U^{t-1}(\bm{s}_{1:m}) + b(\bm{x})\right\}.
		\end{align*}
		Since $v_{ij}\geq 0$ for all $(i,j)$ and $\sum_{j=1}^m v_{ij}>0$ for all $i$, $a(\bm{x})\neq \bm{0}_{m\times 1}$ for all $\bm{x}\in\mathbb{R}^m$ with $x_i\geq 0$ for all $i$ and $\sum_{i=1}^m x_i>0$. Hence, based on Definition 1 in \citet{Daretal06}, $\{U^t(\bm{s}_{1:m}):t=1,2,\ldots\}$ is a temporal compound autoregressive process of order 1. 
		
		By Proposition 6 in \citet{Daretal06}, the process $\{U^t(\bm{s}_{1:m}):t=1,2,\ldots\}$ converges to a stationary distribution as $T\to\infty$ if and only if $\lim_{h\to\infty} a^{\circ h}(\bm{x}) =\bm{0}_{m\times 1}$ for all $\bm{x}\in\RR^m$ with $x_i\geq 0$ for all $i=1,\ldots,m$ and $\sum_{i=1}^m x_i>0$, where $a^{\circ h}(\bm{x})$ represents the function $a(\bm{x})$ compounded $h$ times with itself. In this case, if we define the $m$-dimensional function $K(\bm{x})= (K_1(\bm{x}),\ldots,K_m(\bm{x}))^\T =\left(\frac{x_1}{cx_1+1},\ldots, \frac{x_m}{cx_m+1}\right) ^\T $ for any $\bm{x}\in \RR^m$  with $x_i\geq 0$ for all $i=1,\ldots,m$ and $\sum_{i=1}^m x_i>0$, then for any $h\in\mathbb{N}$, $h\geq 1$, it is straightforward to see that 
		\begin{align*}
			a^{\circ h}(\bm{x}) = \bfV^\T K\left(a^{\circ (h-1)}(\bm{x})\right).
		\end{align*}
		In particular, for any $j=1,\ldots,m$, the $j$th entry of $a^{\circ h}(\bm{x})$ is
		\begin{align*}
			a^{\circ h}_j(\bm{x}) &= \sum_{i=1}^m v_{ij} K_i\left(a^{\circ (h-1)}(\bm{x})\right).
		\end{align*}
		
		\noindent \textbf{Condition 1:} Suppose $\sum_{i=1}^m v_{ij}\leq 1$ for all $j=1,\ldots,m$. For each $h\in \NN$, let $i_{\max}^{(h)} = \arg\max_{1\leq i \leq m} K_i\left(a^{\circ h}(\bm{x})\right)$. Since $v_{ij}$'s are all nonnegative, for any $j=1,\ldots,m$,
		\begin{align} \label{eq:a.upper1}
			0 &\leq  a^{\circ h}_j(\bm{x}) \leq  \left(\sum_{i=1}^m v_{ij} \right) K_{i_{\max}^{(h-1)}}\left(a^{\circ (h-1)}(\bm{x})\right) \nonumber \\
			&\overset{(i)}{\leq} 1\cdot \frac{a_{i_{\max}^{(h-1)}}^{\circ (h-1)}(\bm{x})}{ca_{i_{\max}^{(h-1)}}^{\circ (h-1)}(\bm{x})+1} \overset{(ii)}{=} \frac{\max_{1\leq k\leq m} a^{\circ (h-1)}_k(\bm{x})}{c\max_{1\leq k\leq m} a^{\circ (h-1)}_k(\bm{x})+1},
		\end{align}
		where $(i)$ follows from $0\leq \sum_{i=1}^m v_{ij}\leq 1$ for all $j=1,\ldots,m$ and $(ii)$ from the fact that the function $x/(cx+1)$ increases in $x$ for $x\geq 0$. The relation \eqref{eq:a.upper1} implies that $a_{\max}^h(\bm{x}) := \max_{1\leq j\leq m} a^{\circ h}_j(\bm{x})$ is a nonnegative sequence that strictly decreases in $h$ when strictly positive for any fixed $\bm{x}$ and satisfies
		\begin{align} \label{eq:a.decrease11}
			& 0\leq a_{\max}^h(\bm{x}) \leq \frac{a_{\max}^{h-1}(\bm{x})}{ca_{\max}^{h-1}(\bm{x})+1} \leq a_{\max}^{h-1}(\bm{x}).
		\end{align}
		By the increasing property of the function $f(x)=x/(cx+1)$ for $x\geq 0$, we can apply the second inequality in \eqref{eq:a.decrease11} $h-1$ times iteratively to obtain 
		\begin{align} \label{eq:a.decrease12}
			& 0\leq a_{\max}^h(\bm{x})  \leq f\left(a_{\max}^{h-1}(\bm{x})\right)\leq f^{\circ(h-1)}\left(a_{\max}^{1}(\bm{x})\right) = \frac{a_{\max}^{1}(\bm{x})}{ca_{\max}^{1}(\bm{x})\cdot (h-1) +1}.
		\end{align}
		Therefore, $\lim_{h\to\infty} a_{\max}^h(\bm{x}) =0$. This further implies that $\lim_{h\to\infty} a^{\circ h}_j(\bm{x}) =0$ for all $j=1,2,\ldots,m$, i.e., $\lim_{h\to\infty} a^{\circ h}(\bm{x}) =\bm{0}_{m\times 1}$, as $a^{\circ h}_j(\bm{x})\geq 0$ for all $j\in\{1,\ldots,m\}$ and $h\in\mathbb{N}$.
		
		\noindent \textbf{Condition 2:} Suppose $\sum_{j=1}^m v_{ij}\leq 1$ for all $i=1,\ldots,m$. Since $v_{ij}$'s are all nonnegative, we have that for any $j=1,\ldots,m$,
		\begin{align} \label{eq:a.upper2}
			0 &\leq  \frac{1}{m}\sum_{j=1}^m a^{\circ h}_j(\bm{x}) = \frac{1}{m} \sum_{j=1}^m\sum_{i=1}^m v_{ij} K_{i}\left(a^{\circ (h-1)}(\bm{x})\right) = \frac{1}{m}\sum_{i=1}^m \left(\sum_{j=1}^mv_{ij}\right) K_{i}\left(a^{\circ (h-1)}(\bm{x})\right)  \nonumber \\
			&\overset{(i)}{\leq} \frac{1}{m}\sum_{i=1}^m 1\cdot \frac{a_{i}^{\circ (h-1)}(\bm{x})}{ca_{i}^{\circ (h-1)}(\bm{x})+1} = \frac{1}{m}\sum_{j=1}^m  \frac{a_{j}^{\circ (h-1)}(\bm{x})}{ca_{j}^{\circ (h-1)}(\bm{x})+1} \overset{(ii)}{\leq} \frac{\frac{1}{m}\sum_{j=1}^m a^{\circ (h-1)}_j(\bm{x})}{c\cdot \frac{1}{m}\sum_{j=1}^m a^{\circ (h-1)}_j(\bm{x})  +1},
		\end{align}
		where $(i)$ follows from $0< \sum_{j=1}^m v_{ij}\leq 1$ for all $i=1,\ldots,m$ and $(ii)$ from the concavity of the function $f(x)=x/(cx+1)$ for $x\geq 0$, as $c>0$ and thus $f''(x) = -\frac{2c}{(cx+1)^3}<0$ for all $x\geq 0$. The relation \eqref{eq:a.upper2} implies that $a_{\mean}^h(\bm{x}) := \frac{1}{m}\sum_{j=1}^m a^{\circ h}_j(\bm{x})$ is a nonnegative sequence that strictly decreases in $h$ when strictly positive for any fixed $\bm{x}$ and satisfies
		\begin{align} \label{eq:a.decrease21}
			& 0\leq a_{\mean}^h(\bm{x}) \leq \frac{a_{\mean}^{h-1}(\bm{x})}{ca_{\mean}^{h-1}(\bm{x})+1} \leq a_{\mean}^{h-1}(\bm{x}).
		\end{align}
		By the increasing property of the function $f(x)=x/(cx+1)$ for $x\geq 0$, we can apply the second inequality in \eqref{eq:a.decrease21} $h-1$ times iteratively to obtain 
		\begin{align} \label{eq:a.decrease22}
			& 0\leq a_{\mean}^h(\bm{x}) \leq f\left(a_{\mean}^{h-1}(\bm{x})\right)\leq f^{\circ(h-1)}\left(a_{\mean}^{1}(\bm{x})\right) = \frac{a_{\mean}^{1}(\bm{x})}{ca_{\mean}^{1}(\bm{x})\cdot (h-1) +1}.
		\end{align}
		Therefore, $\lim_{h\to\infty} a_{\mean}^h(\bm{x}) =0$, which is equivalent to $\lim_{h\to\infty} \sum_{j=1}^m a^{\circ h}_j(\bm{x}) =0$. Since $a^{\circ h}_j(\bm{x})\geq 0$ for all $j\in\{1,\ldots,m\}$ and $h\in\mathbb{N}$, we have $\lim_{h\to\infty} a^{\circ h}_j(\bm{x}) =0$ for all $j=1,2,\ldots,m$, i.e., $\lim_{h\to\infty} a^{\circ h}(\bm{x}) =\bm{0}_{m\times 1}$.
	\end{proof}
	
	Under the joint framework specified by \eqref{yPois} and \eqref{eq:ARG1}, the conditional posterior distribution of each $U^t(\bm{s}_i)$ ($t=2,\ldots,T$) does not immediately follow a closed-form parametric distribution, even with a fixed set of $\alpha,c,\bfV$. The fundamental issue lies in the weighted sum of $U^{t-1}(\bm{s}_{1:m})$ in the second equation of \eqref{eq:ARG1}. Using the additive property of independent Poisson random variables, we can further decompose each $Z^{t-1}(\bm{s}_i)$ into a sum of latent Poisson variables with rates $v_{ij}U^{t-1}(\bm{s}_j)/c$, $j=1,\ldots,m$ so that $v_{ij}>0$, such that the conditional posterior of $U^t(\bm{s}_i)$ can be made conjugate as a gamma distribution. If all entries in the $m\times m$ weighting matrix $\bfV$ are strictly positive, however, this strategy will introduce $m^2(T-1)$ latent variables, which incurs exceptionally high posterior sampling cost for spatiotemporal count data with a large number of spatial locations $m$.
	
	To reduce this computational cost quadratic in $m$, we consider a sparse weighting matrix $\bfV$ for practical implementation. In the following, we construct a $\bfV$ that satisfies Condition 2 in Theorem \ref{thm:stationary1} by making each row of $\bfV$ sparse. That is, each $Z^{t-1}(\bm{s}_i)$ only depends on a small number of neighboring spatial locations in the $t-1^{\text{th}}$ time interval. This modeling strategy is fundamentally similar to the Vecchia approximation in the Gaussian process literature (\citealt{Vecchia1988}, \citealt{Datta2016}, \citealt{KatGui21}, \citealt{Zhuetal24}), where the conditional distribution of the response variable at each spatial location is made only dependent on a small set of neighboring locations. 
	
	Specifically, for each $\bm{s}_i$ ($i=1,\ldots,m$), we let the neighboring set of $\bm{s}_i$ be $N(\bm{s}_i)$, which usually consists of the $|N(\bm{s}_i)|$ ``nearest'' neighbors of $\bm{s}_i$ in $\{\bm{s}_1,\dots,\bm{s}_{m}\}\setminus\{\bm{s}_i\}$. We denote $h_s\in\NN,\,h_s\ll m$ as the maximum number of neighbors for each spatial location; thus, $1\leq|N(\bm{s}_i)|\leq h_s$ for all $i=1,\ldots,m$. It is permitted that some spatial locations' numbers of nearest neighbors are strictly smaller than $h_s$. Let us fix any arbitrary $i\in\{1,2,\ldots,m\}$ and consider the $i^{\text{th}}$ row in the $m\times m$ weighting matrix $\bfV$ in \eqref{eq:ARG1}. We assume the following parametric form of $\bfV$:
	\begin{align} \label{eq:V.rho.kappa}
		v_{ii} &= \rho, \nonumber \\
		v_{ij} &= \kappa \cdot w_{ik_i(\bm{s}_j)}, \quad \text{for all } j \text{ such that } \bm{s}_j\in N(\bm{s}_i), \nonumber \\
		v_{ij} &=0, \quad \text{for all } j\neq i\text{ such that } \bm{s}_j\notin N(\bm{s}_i),
	\end{align}
	where $k_i(\bm{s}_j)\in\{1,\ldots,|N(\bm{s}_i)|\}$ satisfies $\bm{s}_j=N(\bm{s}_i)[k_i(\bm{s}_j)]$, i.e., $\bm{s}_j$ is the $k_i(\bm{s}_j)^{\text{th}}$ spatial location in $N(\bm{s}_i)$. $\left\{w_{ik_i(\bm{s}_j)}\right\}_{j:\:\bm{s}_j\in N(\bm{s}_i)} = \left\{w_{ik}\right\}_{1\leq k\leq |N(\bm{s}_i)|}$ consists of prespecified constants in $(0,1]$ with $\sum_{k=1}^{|N(\bm{s}_i)|} w_{ik}=1$. In real applications, for $\bm{s}_j \in N(\bm{s}_i)$, $w_{ik_i(\bm{s}_j)}$ can either be inversely related to the distance between $\bm{s}_i$ and $\bm{s}_j$ or set to $w_{ik_i(\bm{s}_j)} = 1/|N(\bm{s}_i)|$ as uniform over all neighboring locations of $\bm{s}_i$. Following Theorem \ref{thm:stationary1}, we have the following corollary regarding the parameters $\rho$ and $\kappa$ in the model \eqref{eq:V.rho.kappa}.
	
	\begin{corollary} \label{cor:rho.kappa}
		Assume the weighting matrix $\bfV=\{v_{ij}\geq 0:i,j=1,\ldots,m\}$ is specified in \eqref{eq:V.rho.kappa} with known constants $w_{ik_i(\bm{s}_j)}\in(0,1]$ for all $\bm{s}_j\in N(\bm{s}_i)$ and $i=1,\ldots,m$ that satisfy $\sum_{j:\:\bm{s}_j \in  N(\bm{s}_i)} w_{ik_i(\bm{s}_j)}=\sum_{k=1}^{|N(\bm{s}_i)|}w_{ik}=1$ for all $i=1,\ldots,m$. Suppose that the parameters $(\rho,\kappa)$ satisfy $\rho \geq 0,\: \kappa \geq 0,\: 0< \rho+\kappa\leq 1$. Then $\{U^t(\bm{s}_{1:m}):t=1,2,\ldots,T\}$ defined in \eqref{eq:ARG1} is a stationary process with an invariant distribution as $T\to\infty$.
	\end{corollary}
	
	\begin{proof}[Proof of Corollary \ref{cor:rho.kappa}]
		The construction \eqref{eq:V.rho.kappa} ensures that all entries in $\bfV$ are nonnegative. For any $i=1,2,\ldots,m$, 
		\begin{align*}
			& \sum_{j=1}^m v_{ij} = \rho + \sum_{j:\:\bm{s}_j \in  N(\bm{s}_i)} \kappa w_{ik_i(\bm{s}_j)} + 0 = \rho + \kappa \cdot \sum_{k=1}^{|N(\bm{s}_i)|}w_{ik} =\rho + \kappa \in (0,1].
		\end{align*}
		Hence, Condition 2 of Theorem \ref{thm:stationary1} is satisfied, and the conclusion follows.
	\end{proof}
	
	We introduce $(T-1)\sum_{i=1}^m|N(\bm{s}_i)|\leq mh_s(T-1)\ll m^2(T-1)$ latent variables $\{Z_{ij}^t:\,t=1,\ldots,T-1,\,i=1,\ldots,m,\,j=0,1,\ldots,|N(\bm{s}_i)|\}$ that lead to fully conjugate Gibbs sampling steps of $U^t(\bm{s}_i)$'s, $c,\kappa,\rho$ for the model specified by \eqref{yPois}, \eqref{eq:ARG1}, and \eqref{eq:V.rho.kappa}. For each spatial location $\bm{s}_i$, we let $U^1(\bm{s}_i) ~|~ c \overset{\ind}{\sim} \Ga\left(\alpha, \frac{1}{c} \right)$ and for each $t=2,\ldots,T$, 
	\begin{equation}\label{UmodelEq}
		U^t(\bm{s}_i) ~|~ Z_{i0}^{t-1},Z_{i1}^{t-1}, Z_{i2}^{t-1},\ldots,Z_{i|N(\bm{s}_i)|}^{t-1},c\overset{\ind}{\sim} \Ga\left(\alpha+\sum_{j=0}^{|N(\bm{s}_i)|}Z_{ij}^{t-1}, \frac{1}{c} \right),
	\end{equation}
	where $Z^{t-1}(\bm{s}_i)$ in \eqref{eq:ARG1} equals $\sum_{j=0}^{|N(\bm{s}_i)|}Z_{ij}^{t-1}$ in \eqref{UmodelEq} and
	\begin{align}\label{ZsepmodelEq}
		& Z_{ij}^{t-1} ~|~ U_{N(\bm{s}_i)[j]}^{t-1},\kappa,c\overset{\ind}{\sim} \Pois\left(\frac{\kappa}{c}\cdot w_{ij}U_{N(\bm{s}_i)[j]}^{t-1}\right),\;j=1,2,\ldots,|N(\bm{s}_i)|\text{ independent of}\nonumber\\
		& Z_{i0}^{t-1} ~|~ U^{t-1}(\bm{s}_i),\rho,c\overset{\ind}{\sim} \Pois\left(\frac{\rho}{c}\cdot U^{t-1}(\bm{s}_i)\right).
	\end{align}
	Hence, for each $(t,i)$ with $t\in\{2,\ldots,T\}$, 
	\begin{equation}\label{ZmodelEq}
		Z^{t-1}(\bm{s}_i)=\sum_{j=0}^{|N(\bm{s}_i)|}Z_{ij}^{t-1}~\vert~c,U^{t-1}(\bm{s}_i),\rho,U_{N(\bm{s}_i)}^{t-1},\kappa\sim\Pois\left(\lambda_{it}\right),
	\end{equation}
	where $\lambda_{it}=\frac{\kappa}{c}\cdot\sum_{j=1}^{|N(\bm{s}_i)|}w_{ij}U^{t-1}_{N(\bm{s}_i)[j]}+\frac{\rho}{c}\cdot U^{t-1}(\bm{s}_i)$ and $U_{N(\bm{s}_i)}^{t-1} = \left\lbrace U_{N(\bm{s}_i)[j]}^{t-1},\,j=1,2,\ldots,|N(\bm{s}_i)|\right\rbrace$. $U^t(\bm{s}_i)~|~c,U^{t-1}(\bm{s}_i),\rho,U_{N(\bm{s}_i)}^{t-1},\kappa\sim$ Non-Central $\Ga\left(\alpha,\frac{1}{c},\lambda_{it}\right)$, with mean $(\alpha +\lambda_{it}) c$ and variance $(\alpha +2\lambda_{it}) c^2$. See \Cref{countFlowchart} for the structure and flow of our model.
	
	\begin{figure}[htb]
		\centering
		\includegraphics[width=\textwidth]{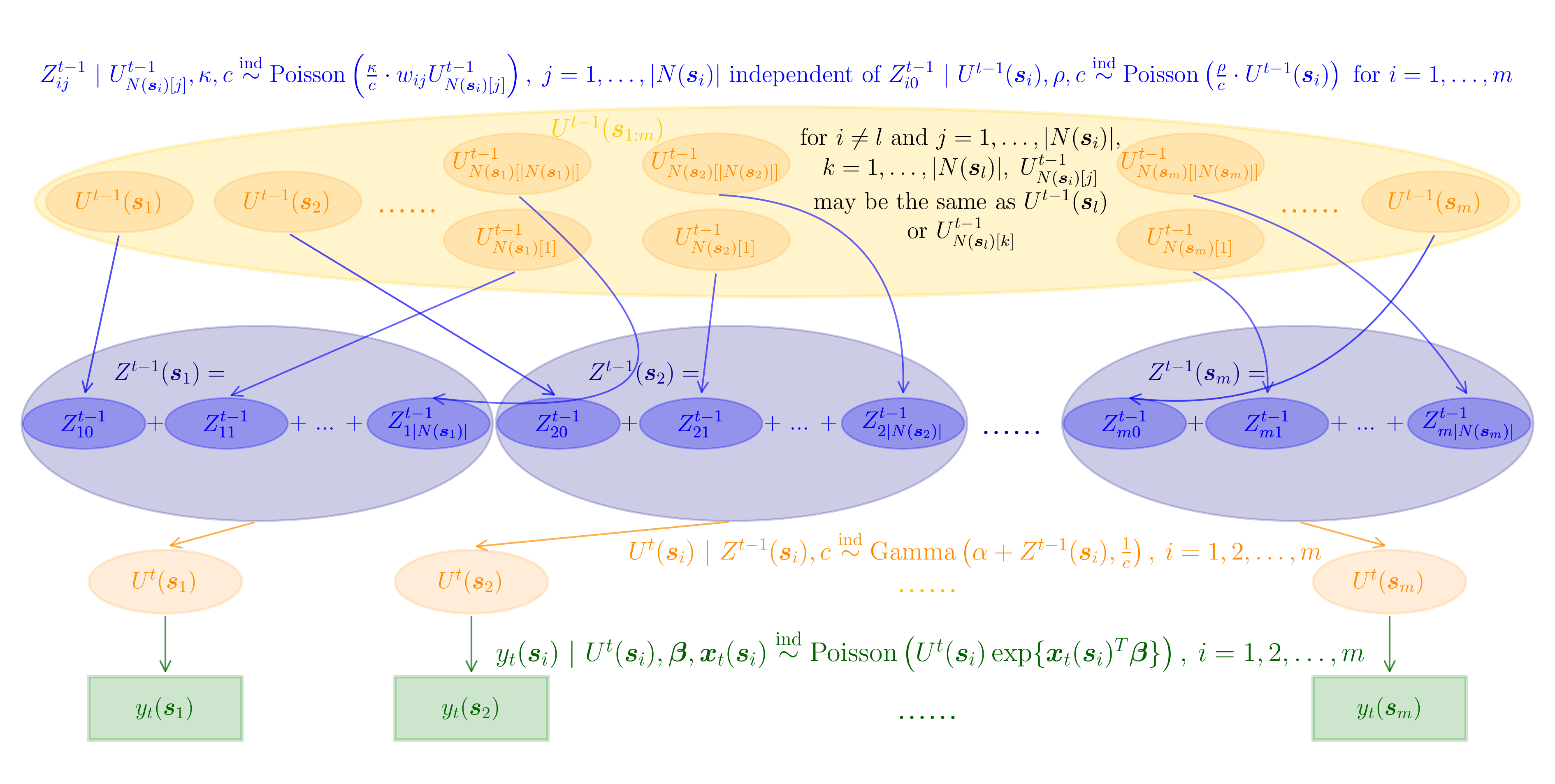}     
		\caption{A schematic diagram of our model specified by \eqref{yPois}, \eqref{UmodelEq}, and \eqref{ZsepmodelEq}.}
		\label{countFlowchart}
	\end{figure}
	
	The hyperparameter $\alpha$ satisfies the Feller condition \citep{Creal2017} $\alpha>1$; $w_{ij}$'s are positive constants satisfying $\sum_{k=1}^{|N(\bm{s}_i)|}w_{ik}=1$ for all $i=1,\ldots,m$; positive parameters $\rho,\kappa$ satisfy $\rho+\kappa \leq 1$. Therefore, \Cref{cor:rho.kappa} applies and thus the process $\{U^t(\bm{s}_{1:m}):t=1,2,\ldots,T\}$ is guaranteed stationary. $\rho$ is imposed a conjugate truncated gamma prior $\pi(\rho)\sim\tGa(a_{\rho},b_{\rho},0,1)$, and $\kappa$ is imposed a conjugate truncated gamma prior $\pi(\kappa)\sim\tGa(a_{\kappa},b_{\kappa},0,1-\rho)$. 
	The $\tGa(a,b,\text{lo},\text{hi})$ distribution is a $\Ga(a,b)$ distribution restricted to a specific interval $(\text{lo},\text{hi})$, where $0\leq\text{lo}<\text{hi}\leq\infty$. $c> 0$ is assigned a conjugate inverse gamma prior $\pi(c)\sim\IG(\alpha_c,\theta_c)$.  In presence of spatiotemporal covariates $\bm{x}_t(\bm{s}_i)$'s, the coefficient vector $\bm{\beta}\in\mathbb{R}^p$ is imposed a standard $p$-variate normal prior $\pi(\bm{\beta})\sim N_p\left(\bm{\mu}_{0\bm{\beta}}, {\Sigma}_{0\bm{\beta}}\right)$.

	\section{Posterior Sampling Algorithm} \label{sec:sampling}
	
	We detail below posterior sampling steps for the model \eqref{yPois} with spatiotemporal frailties $U^{1:T}(\bm{s}_{1:m})=\left\{U^t(\bm{s}_i):~ t=1,\ldots,T,\;i=1,\ldots,m\right\}$ modeled as in \eqref{UmodelEq} and \eqref{ZsepmodelEq}. Our parameter set is $\Theta=\{\bm{\beta}, c, U^{1:T}(\bm{s}_{i}), \kappa, \rho,Z_{ij}^{1:(T-1)},\,i=1,\ldots,m,\,j=0,1,\ldots,|N(\bm{s}_i)|\}$ when $m>1$ and $T>1$. $\Theta$ simplifies to $\{ \bm{\beta}, c, U^{1:T}(\bm{s}_{1}), \rho, Z_{10}^{1:(T-1)}\}$ when $m=1$ and $T>1$, and to $\{ \bm{\beta}, c, U^{1}(\bm{s}_{1:m})\}$ when $T=1$. If we do not consider $U^t(\bm{s}_i)$ directly affected by $U^{t-1}(\bm{s}_i)$, then $\rho$ and $\{Z^{t}_{i0}: t=1,\ldots,T-1,\,i=1,\ldots,m\}$ are not required anymore, i.e., $\rho=0$ in \eqref{eq:V.rho.kappa}. All parameters have full conditional densities corresponding to standard distributions except for the coefficient vector $\bm{\beta}\in\mathbb{R}^p$, which requires a Metropolis step in posterior sampling. See \Cref{appenA} for full deduction details of this algorithm.
	\begin{enumerate}[leftmargin=5mm]
		\item Sampling from the full conditional distribution of $c$:\\
		If $m>1$ and $T>1$, then
		\begin{align}\label{postC}
			f(c|\cdot) &\propto \pi(c) \times \prod_{i=1}^mf\left(U^1(\bm{s}_i)~|~c\right) \times\prod_{t=2}^T\prod_{i=1}^m  f\left(U^t(\bm{s}_i)~|~Z^{t-1}_{i0},Z^{t-1}_{i1},\ldots,Z^{t-1}_{i|N(\bm{s}_i)|},c\right) \times\nonumber\\
			& \quad \prod_{t=2}^T\prod_{i=1}^m f\left(Z_{i0}^{t-1}~|~U^{t-1}(\bm{s}_i),\rho,c\right) \times\prod_{t=2}^T\prod_{i=1}^m \prod_{j=1}^{|N(\bm{s}_i)|}f\left(Z^{t-1}_{ij}~|~U^{t-1}_{N(\bm{s}_i)[j]},\kappa,c\right) \nonumber\\
			&\sim \IG\left(\alpha_c+ Tm\alpha+2\sum_{t=1}^{T-1}\sum_{i=1}^m \sum_{j=0}^{|N(\bm{s}_i)|} Z^t_{ij},\: \theta_c^{\text{post}}\right),
		\end{align}
		where $\theta_c^{\text{post}}=\theta_c+\sum\limits_{t=1}^T\sum\limits_{i=1}^m U^t(\bm{s}_i) + \rho\cdot\sum\limits_{t=1}^{T-1}\sum\limits_{i=1}^m U^t(\bm{s}_i) + \kappa\cdot\sum\limits_{i=1}^m  \sum\limits_{j=1}^{|N(\bm{s}_i)|}w_{ij}\sum\limits_{t=1}^{T-1}U^t_{N(\bm{s}_i)[j]}$.
		
		If $T=1$, then
		\begin{align}
			f(c|\cdot) & \propto \pi(c) \times \prod_{i=1}^mf\left(U^1(\bm{s}_i)~|~c\right)  
			\sim  \IG\left(m\alpha+\alpha_c,\: \sum_{i=1}^m U^1(\bm{s}_i) +\theta_c\right). \nonumber
		\end{align}
		If $m=1$ and $T>1$, then
		\begin{align*}
			f(c|\cdot) &\propto \pi(c) \times f\left(U^1(\bm{s}_1)~|~c\right) \times\prod_{t=2}^T  f\left(U^t(\bm{s}_1)~|~Z^{t-1}_{10},c\right) \times \prod_{t=1}^{T-1} f\left(Z_{10}^{t}~|~U^{t}(\bm{s}_1),\rho,c\right)  \\
			&\sim  \IG\left(T\alpha+2\sum_{t=1}^{T-1}Z_{10}^t+\alpha_c,\: \sum_{t=1}^T U^t(\bm{s}_1)+\rho\cdot\sum_{t=1}^{T-1}U^t(\bm{s}_1)+\theta_c\right).
		\end{align*}
		\item Sampling from the full conditional distributions of $U^{1:T}(\bm{s}_{1:m})$:\par Fix any arbitrary $i\in\{1,2,\ldots,m\}$. For each $l\in\{1,2,\ldots,m\}$ such that $\bm{s}_i\in N(\bm{s}_l)$, we denote $k_l(\bm{s}_i)$ as the positive integer less than or equal to $h_s$ such that $\bm{s}_i=N(\bm{s}_l)[k_l(\bm{s}_i)]$. If $m>1$ and $T>1$, then for any $t\in\{2,3,\ldots,T-1\}$,
		\begin{align}\label{postU}
			&f\left(U^t(\bm{s}_i)|\cdot\right)  \propto f\left(y_t(\bm{s}_i)~|~U^t(\bm{s}_i),\bm{\beta}\right) \times f(U^t(\bm{s}_i)~|~Z^{t-1}_{i0},Z^{t-1}_{i1},\ldots,Z^{t-1}_{i|N(\bm{s}_i)|},c)\times\nonumber\\
			&\qquad f\big(Z^t_{i0}~|~U^t(\bm{s}_i),\rho,c\big)\times\prod_{l:\bm{s}_i\in N(\bm{s}_l)}^{1\leq l\leq m} f\big(Z^t_{lk_l(\bm{s}_i)}~|~U^t(\bm{s}_i),\kappa,c\big) \nonumber\\
			&\quad \sim \Ga\left(y_t(\bm{s}_i)+\alpha+\sum_{j=0}^{|N(\bm{s}_i)|}Z^{t-1}_{ij}+ Z^t_{i0}+\sum_{l:\bm{s}_i\in N(\bm{s}_l)}^{1\leq l\leq m}Z^t_{lk_l(\bm{s}_i)},\: \text{rate}_{U^t(\bm{s}_i)}\right), 
		\end{align}
		where $\text{rate}_{U^t(\bm{s}_i)}= e^{\bm{x}_t(\bm{s}_i)^{\T}\bm{\beta}}+\frac{1}{c}\left(1+\rho+\kappa\sum_{l:\bm{s}_i\in N(\bm{s}_l)}^{1\leq l\leq m}w_{lk_l(\bm{s}_i)}\right)$.       
		\begin{align*}             
			&f\left(U^t(\bm{s}_i)|\cdot\right) \sim \Ga\left(y_t(\bm{s}_i)+\alpha+ Z^t_{i0}+\sum_{l:\bm{s}_i\in N(\bm{s}_l)}^{1\leq l\leq m}Z^t_{lk_l(\bm{s}_i)},\: \text{rate}_{U^t(\bm{s}_i)}\right)\text{  for } t=1,\text{ and}\\           
			&f\left(U^t(\bm{s}_i)|\cdot\right) \sim \Ga\left(y_t(\bm{s}_i)+\alpha+\sum_{j=0}^{|N(\bm{s}_i)|}Z^{t-1}_{ij},\: e^{\bm{x}_t(\bm{s}_i)^{\T}\bm{\beta}}+\frac{1}{c}\right)\text{ for }t=T.\\
			&\text{If }T=1,\text{ then }
			f\left(U^1(\bm{s}_i)|\cdot\right) \sim \Ga\left(y_1(\bm{s}_i)+\alpha,\: e^{\bm{x}_1(\bm{s}_i)^{\T}\bm{\beta}}+\frac{1}{c}\right).\\
			&\text{If }m=1\text{ and }T>1,\text{ then for any }t\in\{2,\ldots,T-1\},\\
			&f\left(U^t(\bm{s}_1)|\cdot\right) \sim \Ga\left(y_t(\bm{s}_1)+\alpha+Z_{10}^{t-1}+Z_{10}^t,\: e^{\bm{x}_t(\bm{s}_1)^{\T}\bm{\beta}}+\frac{1}{c}(1+\rho)\right).\\
			&f\left(U^t(\bm{s}_1)|\cdot\right) \sim \Ga\left(y_t(\bm{s}_1)+\alpha+Z_{10}^{t},\: e^{\bm{x}_t(\bm{s}_1)^{\T}\bm{\beta}}+\frac{1}{c}(1+\rho)\right)\text{ for }t=1,\text{ and }\\
			&f\left(U^t(\bm{s}_1)|\cdot\right) \sim \Ga\left(y_t(\bm{s}_1)+\alpha+Z_{10}^{t-1},\: e^{\bm{x}_t(\bm{s}_1)^{\T}\bm{\beta}}+\frac{1}{c}\right)\text{ for }t=T.		
		\end{align*}
		\item Sampling from the full conditional distribution of $\rho$ when $T>1$:
		\begin{align}\label{postRho}
			f\left(\rho|\cdot\right) &\propto \pi(\rho) \times \prod_{t=1}^{T-1}\prod_{i=1}^m f\left(Z_{i0}^t~|~U^t(\bm{s}_i),\rho,c\right)  \nonumber\\
			& \sim \tGa\left(a_{\rho}+\sum_{t=1}^{T-1}\sum_{i=1}^m Z^t_{i0}, \: b_{\rho}+\frac{1}{c}\sum_{t=1}^{T-1}\sum_{i=1}^mU^{t}(\bm{s}_i),\:0,\:1\right).
		\end{align}
		\item Sampling from the full conditional distribution of $\kappa$ when $m>1$ and $T>1$:
		\begin{align}\label{postKappa}
			f\left(\kappa|\cdot\right) &\propto \pi(\kappa) \times \prod_{t=1}^{T-1}\prod_{i=1}^m \prod_{j=1}^{|N(\bm{s}_i)|} f\left(Z_{ij}^t~|~U^t_{N(\bm{s}_i)[j]},\kappa,c\right) \nonumber\\
			& \sim \tGa\left(a_{\kappa}+\sum_{t=1}^{T-1}\sum_{i=1}^m \sum_{j=1}^{|N(\bm{s}_i)|}Z^t_{ij}, \: \text{rate}^{\text{post}}_{\kappa},\:0,\:1-\rho\right), 
		\end{align}
		where $\text{rate}^{\text{post}}_{\kappa}=b_{\kappa}+\frac{1}{c}\sum_{i=1}^m\sum_{j=1}^{|N(\bm{s}_i)|}w_{ij}\sum_{t=1}^{T-1}U^t_{N(\bm{s}_i)[j]}$.
		\item Sampling from the full conditional distributions of $Z^t_{ij}$'s when $T>1$:\par
		Fix any arbitrary $t\in\{1,2,\ldots,T-1\}$, $i\in\{1,2,\ldots,m\}$. Then 
		\begin{align}\label{postZi0}
			f\left(Z^t_{i0}|\cdot\right)&\propto f\left(U^{t+1}(\bm{s}_i)~|~Z^t_{i0},Z^t_{i1},\ldots,Z^t_{i|N(\bm{s}_i)|},c\right) \times  f\left(Z^t_{i0}~|~U^t(\bm{s}_i),\rho,c\right)\nonumber\\
			&\sim \Bessel\left(\alpha+\sum_{k=1}^{|N(\bm{s}_i)|}Z^t_{ik}-1, \frac{2}{c}\sqrt{\rho \cdot U^{t+1}(\bm{s}_i)U^t(\bm{s}_i)}\right). 
		\end{align}
		When $m>1$, we further fix any arbitrary $j\in\{1,2,\ldots,|N(\bm{s}_i)|\}$. Then 
		\begin{align}\label{postZi}
			f\left(Z^t_{ij}|\cdot\right)&\propto f\left(U^{t+1}(\bm{s}_i)~|~Z^t_{i0},Z^t_{i1},\ldots,Z^t_{i|N(\bm{s}_i)|},c\right) \times  f\left(Z^t_{ij}~|~U^t_{N(\bm{s}_i)[j]},\kappa,c\right)\nonumber\\
			&\sim \Bessel\left(\alpha+\sum_{k=0,k\neq j}^{|N(\bm{s}_i)|}Z^t_{ik}-1,\: \frac{2}{c}\sqrt{\kappa w_{ij}U^{t+1}(\bm{s}_i)U^t_{N(\bm{s}_i)[j]}}\right),
		\end{align}
		where $X\sim  \Bessel(\nu,a)$ has the probability mass function $\mathbb{P}(X=n\mid a,\nu)=\frac{(a/2)^{2n+\nu}}{I_{\nu}(a)\Gamma(n+\nu+1)n!}$ for $n\in \NN$, $a>0,\nu>-1$, and $I_{\nu}(\cdot)$ is the modified Bessel function of the first kind. We generate Bessel random variables via the efficient rejection sampling algorithms in Section 3 of \citet{Devroye2002}.
		\item Sampling from the full conditional distribution of $\bm{\beta}_{p\times 1}$ via a Metropolis step:
		\begin{align}\label{postBeta}
			&f\left(\bm{\beta}|\cdot\right) \propto \pi(\bm{\beta}) \times \prod_{i=1}^m\prod_{t=1}^T f\left(y_t(\bm{s}_i)~|~U^t(\bm{s}_i),\bm{\beta}\right) \\
			&\;\propto \exp\left\lbrace -\frac{1}{2}\left(\bm{\beta}-\bm{\mu}_{0\bm{\beta}}\right)^{\T}\Sigma_{0\bm{\beta}}^{-1}\left(\bm{\beta}-\bm{\mu}_{0\bm{\beta}}\right) + \sum_{i=1}^m\sum_{t=1}^T \left[y_t(\bm{s}_i)\cdot \bm{x}_t(\bm{s}_i)^{\T}\bm{\beta} -  U^t(\bm{s}_i) e^{\bm{x}_t(\bm{s}_i)^{\T}\bm{\beta}} \right]\right\rbrace  \nonumber 
		\end{align}
		is not a standard parametric family, and we resort to a Metropolis step for sampling $\bm{\beta}$.
		\par We consider a symmetric Markov kernel $Q:\mathbb{R}^p \times 
		\mathbb{R}^p\rightarrow \mathbb{R}^+$ with $
		Q\left(\bm{\beta}_a, \bm{\beta}_b\right)=Q(\bm{\beta}_b,\bm{\beta}_a)$ for all $\bm{\beta}_a,\bm{\beta}_b\in\mathbb{R}^p$. $
		Q\left(\bm{\beta}_a, \bm{\beta}_b\right)$ equals the density of a $N_{p}(\bm{\beta}_a, V)$ random variable evaluated at $\bm{\beta}_b$ with probability $p$ and the density of a $N_{p}(\bm{\beta}_a, 100V)$ random variable evaluated at $\bm{\beta}_b$ with probability $1-p$, 
		where $V$ is an estimate 
		of the inverse Hessian of $\ln\left[f(\bm{\beta}|\cdot)\right]$ and $p\in(0,1)$ should be quite large to give less probability to big moves in the parameter space. 
		At each MCMC iteration $w\in\mathbb{N}$, we let the current parameter estimate for $\bm{\beta}$ be $\bm{\beta}^{(w)}$. Since 
		\begin{align}\label{sampleBetaHessian}
			\frac{\partial^2}{\partial \bm{\beta}\partial \bm{\beta}^{\T}}\ln f(\bm{\beta}|\cdot) = -\Sigma_{0\bm{\beta}}^{-1} -  \sum_{i=1}^m\sum_{t=1}^T U^t(\bm{s}_i) e^{\bm{x}_t(\bm{s}_i)^{\T}\bm{\beta}}\bm{x}_t(\bm{s}_i)\bm{x}_t(\bm{s}_i)^{\T},
		\end{align}
		we set $V^{(w)} = \left[\Sigma_{0\bm{\beta}}^{-1} +  \sum_{i=1}^m\sum_{t=1}^T U^t(\bm{s}_i) e^{\bm{x}_t(\bm{s}_i)^{\T}\bm{\beta}^{(w)}}\bm{x}_t(\bm{s}_i)\bm{x}_t(\bm{s}_i)^{\T}\right]^{-1}$. We propose a new parameter value $\bm{\beta}^*$ from the density $q(\bm{\beta}|\bm{\beta}^{(w)})=Q\left(\bm{\beta}^{(w)},\bm{\beta}\right)$ and then set 
		\begin{equation*}
			\bm{\beta}^{(w+1)}=
			\left\{ \begin{array}{ll}
				\bm{\beta}^*,&\text{ with probability }\alpha\left(\bm{\beta}^{(w)},\bm{\beta}^*\right)\\ 
				\bm{\beta}^{(w)},&\text{ with probability }1-\alpha\left(\bm{\beta}^{(w)},\bm{\beta}^*\right)
			\end{array} \right. ,
		\end{equation*}
		where $\alpha\left(\bm{\beta}^{(w)},\bm{\beta}^*\right)=\min\left\lbrace\frac{f(\bm{\beta}^*|\cdot)}{f(\bm{\beta}^{(w)}|\cdot)},1\right\rbrace$ is the acceptance probability for this multivariate random walk Metropolis algorithm. 
	\end{enumerate}

	\section{Spatial and Temporal Predictions} \label{sec:prediction}
	Bayesian predictions at any arbitrary future time intervals and unseen spatial locations under our modeling framework are straightforward. We can decompose the integral representation of the posterior predictive distribution (PPD) into several known densities and then obtain the PPD via composition sampling. 
	Suppose we have an arbitrary number $q\in\mathbb{N},q\geq 1$ of future time periods $T+1,\dots,T+q$ of the same length and $r\in\mathbb{N},r\geq 1$ new spatial locations $\bm{s}_{(m+1):(m+r)}$. Then for Bayesian predictions given the corresponding covariates matrices $\left[ X_{(T+1):(T+q)}\big(\bm{s}_{1:m}\big)\right] _{qm\times p}$ and $X_{1:(T+q)}\big(\bm{s}_{(m+1):(m+r)}\big)_{(T+q)r\times p}$, if any, we can write the PPD as
	\begin{align}\label{PPDspatemp}
		&f\left(\bm{y}_{(T+1):(T+q)}\big(\bm{s}_{1:m}\big), \bm{y}_{1:(T+q)}\big(\bm{s}_{(m+1):(m+r)}\big)~|~\bm{y}_{1:T}\big(\bm{s}_{1:m}\big),X^{\text{new}}\big)\right)\nonumber\\
		=&\int_{\Theta}f\left(\bm{y}_{(T+1):(T+q)}\big(\bm{s}_{1:m}\big), \bm{y}_{1:(T+q)}\big(\bm{s}_{(m+1):(m+r)}\big)~|~\Theta,\bm{y}_{1:T}\big(\bm{s}_{1:m}\big),X^{\text{new}}\right)\nonumber\\
		&\times\pi\left(\Theta~|~\bm{y}_{1:T}\big(\bm{s}_{1:m}\big)\right)d\Theta,
	\end{align}
	where $\Theta= 	
	\left(\bm{\beta},U^{(T+1):(T+q)}\left(\bm{s}_{1:m}\right), U^{1:(T+q)}\big(\bm{s}_{(m+1):(m+r)}\big),U^{1:T}(\bm{s}_{1:m}),\kappa,\rho,c\right)$ and 
	$X^{\text{new}} = \left\lbrace X_{(T+1):(T+q)}\big(\bm{s}_{1:m}\big), X_{1:(T+q)}\big(\bm{s}_{(m+1):(m+r)}\big)\right\rbrace$. 	
	We can then partition the integral in \eqref{PPDspatemp} into
	\begin{align} 
		&\int_{\Theta}\underbrace{f\left(\bm{y}_{(T+1):(T+q)}\big(\bm{s}_{1:m}\big), \bm{y}_{1:(T+q)}\big(\bm{s}_{(m+1):(m+r)}\big)~|~\bm{\beta},U^{(T+1):(T+q)}\left(\bm{s}_{1:m}\right), U^{1:(T+q)}\big(\bm{s}_{(m+1):(m+r)}\big),X^{\text{new}}\right)}_{T_1}\nonumber\\	
		&\quad\times\underbrace{f\left(U^{1:(T+q)}\left(\bm{s}_{(m+1):(m+r)}\right)~|~U^{1:(T+q-1)}\left(\bm{s}_{1:m}\right),\kappa,\rho,c\right)}_{T_2}\nonumber\\&\quad\times\underbrace{f\left(U^{(T+1):(T+q)}\left(\bm{s}_{1:m}\right)~|~U^{T}\left(\bm{s}_{1:m}\right),\kappa,\rho,c\right)}_{T_3}\times\underbrace{\pi\left(\bm{\beta},U^{1:T}\left(\bm{s}_{1:m}\right),\kappa,\rho,c~|~\bm{y}_{1:T}(\bm{s}_{1:m})\right)}_{T_4}d\Theta, \nonumber
	\end{align}
	where $T_1$ is the likelihood, $T_4$ is the parameters' posterior distribution obtained from the original model fit's MCMC sampler, $\hat{U}^{1:(T+q)}\left(\bm{s}_{(m+1):(m+r)}\right)$ conditioning on $U^{1:T}\left(\bm{s}_{1:m}\right),\kappa,\rho,c$, $\hat{U}^{(T+1):(T+q-1)}\left(\bm{s}_{1:m}\right)$ can be predicted based on \eqref{predZU} for $T_2$, and $T_3$ is
	\begin{align}
		&f\left(U^{(T+1):(T+q)}\left(\bm{s}_{1:m}\right)~|~U^{T}\left(\bm{s}_{1:m}\right),\kappa,\rho,c\right) = \prod_{t=T+1}^{T+q}f\left(U^{t}\left(\bm{s}_{1:m}\right)~|~U^{t-1}\left(\bm{s}_{1:m}\right),\kappa,\rho,c\right)\nonumber\\
		=&\prod_{t=T+1}^{T+q}\prod_{i=1}^mf\left(U^{t}\left(\bm{s}_{i}\right)~|~U^{t-1}_{N(\bm{s}_i)},\kappa,U^{t-1}\left(\bm{s}_{i}\right),\rho,c\right)\nonumber\\
		=&\prod_{t=T+1}^{T+q}\prod_{i=1}^m\text{Non-Central} \Ga\left(\alpha,\frac{1}{c},\:\lambda_{it}\right),
	\end{align}
	where $\lambda_{it} = \frac{\kappa}{c}\cdot\sum_{j=1}^{|N(\bm{s}_i)|}w_{ij}U^{t-1}_{N(\bm{s}_i)[j]}+\frac{\rho}{c}\cdot U^{t-1}(\bm{s}_i)$. Assume $m>h_s$. For each new location $\bm{s}$ outside of the training set ${\mathcal{S}}=\{\bm{s}_1,\bm{s}_2,\ldots,\bm{s}_m\}$, we let $N(\bm{s})$ consist of the $h_s$ nearest neighbors of $\bm{s}$ in ${\mathcal{S}}$ with the attached weights $\{w_j(\bm{s}):j=1,\ldots,h_s\}$ that satisfy $w_{j}(\bm{s})\in(0,1]$ for all $j=1,\ldots,h_s$ and $\sum_{k=1}^{h_s}w_{k}(\bm{s})=1$. 
	Independently for all new locations, we let $U^1(\bm{s}) ~|~c\sim\Ga\left(\alpha,\frac{1}{c}\right)$, and for $t=2,3,\ldots,T+q$,
	\begin{align}\label{predZU}
		& U^t(\bm{s}) ~|~ Z^{t-1}(\bm{s}),c \sim \Ga\left(\alpha+Z^{t-1}(\bm{s}), \frac{1}{c} \right),\nonumber\\
		&Z^{t-1}(\bm{s})~|~U^{t-1}_{N(\bm{s})},\kappa,U^{t-1}(\bm{s}),\rho,c\sim\Pois\left(\frac{\kappa}{c}\cdot\sum_{j=1}^{h_s}w_{j}(\bm{s})U^{t-1}_{N(\bm{s})[j]} + \frac{\rho}{c}\cdot U^{t-1}(\bm{s})\right).
	\end{align}
	Then each $U^t(\bm{s})~|~U^{t-1}_{N(\bm{s})},\kappa,U^{t-1}(\bm{s}),\rho,c\sim$ Non-Central $\Ga\left(\alpha,\frac{1}{c},\lambda_t(\bm{s})\right)$, where $\lambda_t(\bm{s}) = \frac{\kappa}{c}\sum_{j=1}^{h_s}w_{j}(\bm{s})U^{t-1}_{N(\bm{s})[j]} + \frac{\rho}{c} U^{t-1}(\bm{s})$, whose mean and variance are $\alpha c+\kappa\sum\limits_{j=1}^{h_s}w_{j}(\bm{s})U^{t-1}_{N(\bm{s})[j]} + \rho  U^{t-1}(\bm{s})$ and $\alpha c^2+2c\left[\kappa\sum\limits_{j=1}^{h_s}w_{j}(\bm{s})U^{t-1}_{N(\bm{s})[j]}+ \rho  U^{t-1}(\bm{s})\right]$, respectively. 
	
	\section{Simulation Experiments} \label{sec:simulation}
	This section first verifies the satisfactory parameter estimation and model fitting performance of our proposed model in \eqref{yPois} and \eqref{eq:ARG1} through extensive simulation experiments. This section then compares computational efficiency, modeling and prediction performance of our framework with alternative Bayesian spatiotemporal count methods on simulated data from our MCMC model and from MCMC models in \texttt{CARBayesST} \citep{Lee2018}. 
	
	We first performed three groups of simulation experiments that all simulated data from \eqref{yPois} without external covariates $\bm{x}_t(\bm{s}_i)$'s and differ in variations of \eqref{eq:ARG1} modeling spatiotemporal frailties $U^{1:T}(\bm{s}_{1:m})=\left\{U^t(\bm{s}_i):~ t=1,\ldots,T,\;i=1,\ldots,m\right\}$ with the fixed hyperparameter $\alpha=1.0001$. We specified $T=100$ time intervals of equal length and $m=11^2=121$ spatial locations $\bm{s}_i=(i_1,i_2)$ for $i=1,\ldots,m$ on an equispaced 2-dimensional grid, where $(i_1,i_2)\in\{1,\ldots,11\}\times\{1,\ldots,11\}$. Within each group, various sets of true parameter values for $\rho$ (only for the first group), $\kappa$, $c$, and a few hyperparameter combinations were adopted. We ran each MCMC chain for $3\times 10^4$ burn-in iterations and $2\times 10^4$ post-burn-in iterations, which were thinned to $10^4$ posterior samples for analysis.
	
	The first group adopted our standard spatiotemporal frailty model \eqref{UmodelEq} and \eqref{ZsepmodelEq} for data simulation and model fitting. For each spatial location $\bm{s}_i$, $N(\bm{s}_i)$ consists of the $h_s=12$ nearest neighbors of $\bm{s}_i$ in $\{\bm{s}_1,\dots,\bm{s}_{m}\}\setminus\{\bm{s}_i\}$ and $w_{ij} = 1/h_s$, $j=1,\ldots,h_s$ in \eqref{ZsepmodelEq}. For entries of the $i$th row in the matrix $\bfV=\{v_{ij}:i,j=1,\ldots,m\}$ in \eqref{eq:ARG1}, $v_{ii}=\rho>0$, $v_{ij}=\kappa/h_s>0$ for all $j$ such that $\bm{s}_j\in N(\bm{s}_i)$, and $v_{ij}=0$ for all $j\neq i$ such that $\bm{s}_j\notin N(\bm{s}_i)$.
	
	The second group used our spatiotemporal frailty model \eqref{eq:ARG1} with entries in the matrix $\bfV=\{v_{ij}:i,j=1,\ldots,m\}$ in \eqref{eq:ARG1} given as follows for data simulation and model fitting. For each spatial location $\bm{s}_i$, $N(\bm{s}_i)$ consists of $\bm{s}_i$ itself and the $h_s-1=12-1=11$ nearest neighbors of $\bm{s}_i$ in $\{\bm{s}_1,\dots,\bm{s}_{m}\}\setminus\{\bm{s}_i\}$. $v_{ij}=\kappa/h_s>0$ for all $j$ such that $j=i$ or $\bm{s}_j\in N(\bm{s}_i)$ to weight $\bm{s}_i$ and each of its neighbors equally, and $v_{ij}=0$ for all $j$ such that $j\neq i$ and $\bm{s}_j\notin N(\bm{s}_i)$. Hence, $\sum_{j=1}^m v_{ij} = h_s\times \kappa/h_s + 0 = \kappa \in(0,1)$. As $v_{ij}\geq 0$ for all $(i,j)$ and $0<\sum_{j=1}^m v_{ij}< 1$ for all $i=1,\ldots,m$, Condition 2 in \Cref{thm:stationary1} applies and the process $\{U^t(\bm{s}_{1:m}):t=1,2,\ldots,T\}$ is guaranteed temporally stationary. 
	
	The third group assumed that the $m$ spatial locations are labeled as $\bm{s}_1,\ldots,\bm{s}_m$, starting from the central location $\bm{s}_1$ and going clockwise in the 2-dimensional grid, and employed a slightly varied version of our spatiotemporal frailty model \eqref{eq:ARG1} with no autocorrelation detailed as follows for data simulation and model fitting. We let $N(\bm{s}_1)=\emptyset$ and define for each $i\in\{2,3,\dots,m\}$ location $\bm{s}_i$'s neighboring set $N(\bm{s}_i)$ as the $\min\{h_s,i-1\}$ nearest neighbors of $\bm{s}_i$ in $\{\bm{s}_1,\dots,\bm{s}_{i-1}\}$. We assume \eqref{eq:ARG1} for all $\bm{s}_i$ with $i=2,\ldots,m$, $U^t(\bm{s}_1)\sim \Ga(\alpha,1/c)$, and set all $m$ entries in the first row of $\bfV=\{v_{ij}:i,j=1,\ldots,m\}$ as 0. For each $i=2,\ldots,m$, $v_{ij}=\kappa/|N(\bm{s}_i)|$ for all $j<i$ such that $\bm{s}_j\in N(\bm{s}_i)$, where $|N(\bm{s}_i)|=\min\{h_s,\:i-1\}=\min\{12,\:i-1\}$, and $v_{ij}=0$ for all $j$ such that $\bm{s}_j\notin N(\bm{s}_i)$. Hence, $\sum_{j=1}^m v_{ij} = |N(\bm{s}_i)|\times \kappa/|N(\bm{s}_i)| + 0 = \kappa \in(0,1)$. As $v_{ij}\geq 0$ for all $(i,j)$ and $0<\sum_{j=1}^m v_{ij}< 1$ for all $i=2,\ldots,m$, \Cref{thm:stationary2} in \Cref{appenB}, a theorem similar to \Cref{thm:stationary1}, applies and guarantees the temporal stationarity of the process $\{U^t(\bm{s}_{1:m}):t=1,2,\ldots,T\}$. 
	
	We calculate Mean Absolute Error (MAE)$ = \frac{1}{mT}\sum_{t=1}^T\sum_{i=1}^m \big|y_t(\bm{s}_i)-\hat{y}_t(\bm{s}_i)\big|$, where $\hat{y}_t(\bm{s}_i)$ is the posterior mean of the estimated expectation of the response in time interval $t$ and spatial location $\bm{s}_i$ for each $(t,i)$ (\Cref{postMeanComplete,MAE} and \Cref{plot:absErrorBoxplot_hyparaDefault,plot:absErrorBoxplot_hyparaMiddle}), and present summary statistics for posterior samples of parameters $\rho$ (for the first group only), $\kappa$, $c$, $U_t(\bm{s}_i)$'s, and the actual observed counts $y_t(\bm{s}_i)$'s corresponding to true parameter values $(\rho,\kappa,c)\in \{0.4\}\times \{0.4\}\times \{5,10,15,20,50,100,500,1000,2000,5000\}$ and hyperparameter combinations \texttt{hypara1} and \texttt{hypara2} for the first simulation group and true parameter values $(\kappa,c)\in\{0.35,0.7\}\times \{5,10,15,20,50,100,500,1000,2000,5000\}$ and hyperparameter combinations \texttt{hypara3} and \texttt{hypara4} for the second and the third simulation groups.
	\begin{enumerate}
		\setlength{\itemsep}{0pt}
		\setlength{\parskip}{0pt}
		\item \texttt{hypara1}: $(\alpha_c, \theta_c) = (2,10)$ for $c$, $(a_{\kappa}, b_{\kappa}) = (0.55,1)$ for $\kappa$, $(a_{\rho}, b_{\rho}) = (0.4,1)$ for $\rho$;
		\item \texttt{hypara2}: $(\alpha_c, \theta_c) = (2,50)$ for $c$, $(a_{\kappa}, b_{\kappa}) = (0.4,1)$ for $\kappa$, $(a_{\rho}, b_{\rho}) = (0.55,1)$ for $\rho$.
		\item \texttt{hypara3}: $(\alpha_c, \theta_c) = (2,10)$ for $c$ and $(a_{\kappa}, b_{\kappa}) = (0.9,1)$ for $\kappa$;
		\item \texttt{hypara4}: $(\alpha_c, \theta_c) = (2,50)$ for $c$ and $(a_{\kappa}, b_{\kappa}) = (0.4,1)$ for $\kappa$.
	\end{enumerate}
	The obtained results presented in \Cref{MAE,postMeanComplete,postMeanKappa,postMeanC,c5complete,c10complete,c15complete,c20complete,c50complete,c100complete,c500complete,c1000complete,c2000complete,c5000complete,undirectedkappa0.35c5,undirectedkappa0.35c10,undirectedkappa0.35c15,undirectedkappa0.35c20,undirectedkappa0.35c50,undirectedkappa0.35c100,undirectedkappa0.35c500,undirectedkappa0.35c1000,undirectedkappa0.35c2000,undirectedkappa0.35c5000,undirectedkappa0.7c5,undirectedkappa0.7c10,undirectedkappa0.7c15,undirectedkappa0.7c20,undirectedkappa0.7c50,undirectedkappa0.7c100,undirectedkappa0.7c500,undirectedkappa0.7c1000,undirectedkappa0.7c2000,undirectedkappa0.7c5000,directedkappa0.35c5,directedkappa0.35c10,directedkappa0.35c15,directedkappa0.35c20,directedkappa0.35c50,directedkappa0.35c100,directedkappa0.35c500,directedkappa0.35c1000,directedkappa0.35c2000,directedkappa0.35c5000,directedkappa0.7c5,directedkappa0.7c10,directedkappa0.7c15,directedkappa0.7c20,directedkappa0.7c50,directedkappa0.7c100,directedkappa0.7c500,directedkappa0.7c1000,directedkappa0.7c2000,directedkappa0.7c5000} and \Cref{plot:absErrorBoxplot_hyparaDefault,plot:absErrorBoxplot_hyparaMiddle} corroborate our model's accurate parameter estimation and in-sample prediction capabilities for spatiotemporal count data spanning a wide range.
	
	\begin{table}[htb]
		{
			\begin{adjustwidth}{-0cm}{-0cm}
				\begin{center}
					\scalebox{0.9}{\begin{tabular}{|*{9}{c|}}
							\hline
							\multirow{2}{*}{\textbf{True $c$}}& \multicolumn{2}{|c|}{Posterior Mean of $c$} & \multicolumn{2}{|c|}{Posterior Mean of $\kappa$} & \multicolumn{2}{|c|}{Posterior Mean of $\rho$} & \multicolumn{2}{|c|}{MAE}\\ 
							\cline{2-9}
							\textbf{Value}  & \texttt{hypara1} & \texttt{hypara2} & \texttt{hypara1} & \texttt{hypara2} & \texttt{hypara1} & \texttt{hypara2} & \texttt{hypara1} & \texttt{hypara2}\\
							\hline
							$c=5$ & 4.987 & 4.997 & 0.4023 & 0.4030 & 0.4040 & 0.4031 & 1.272418  & 1.269910\\
							\hline
							$c=10$ & 9.936 & 9.949 & 0.4031 & 0.4025 & 0.4035 & 0.4040 & 1.33638 & 1.334743 \\
							\hline
							$c=15$ & 14.91 & 14.92 & 0.4088 & 0.4074 & 0.3979 &  0.3991 & 1.350147 & 1.348748 \\
							\hline
							$c=20$ & 19.89 & 19.87 & 0.4064 & 0.4072 & 0.4001 & 0.3995 & 1.363482 & 1.362761 \\
							\hline
							$c=50$ & 49.86 & 49.84 & 0.4061 & 0.4069 & 0.3998 & 0.3991 & 1.380839 & 1.379167 \\
							\hline
							$c=100$ & 99.57 & 99.53 & 0.4061 & 0.4068 & 0.4002 & 0.3995 & 1.387671 & 1.389354  \\
							\hline
							$c=500$ & 497.5 & 497.7 & 0.4065 & 0.4062 & 0.3999 & 0.4001 & 1.429515 & 1.433766  \\
							\hline
							$c=1000$ &  996.9 & 994.9  & 0.4067 & 0.4073 & 0.3993 & 0.3991 & 1.483429 & 1.477110 \\
							\hline
							$c=2000$ & 1990 & 1994 & 0.4061 & 0.4072 & 0.4001 & 0.3989 & 1.570770 & 1.590506 \\
							\hline
							$c=5000$ & 4983 & 4973 & 0.4072 & 0.4064 & 0.3989 & 0.3999 & 1.816319 &  1.821356\\
							\hline
					\end{tabular}}
				\end{center}
			\end{adjustwidth}
			\caption{MAE and Mean value of the kept $10^4$ posterior samples of $c$, $\kappa$, $\rho$ for each of the $2\times 10=20$ settings for the first simulation group with true parameter values $\kappa=\rho=0.4$.}
			\label{postMeanComplete}
		}
	\end{table}
	
	We then implemented our model with hyperparameter combination \texttt{hypara1} given earlier and six MCMC models--\texttt{ST.CARlinear}, \texttt{ST.CARanova}, \texttt{ST.CARsepspatial}, \texttt{ST.CARar}, \texttt{ST.CARadaptive}, and \texttt{ST.CARlocalised}--from the R package \texttt{CARBayesST} \citep{Lee2018} and used \href{https://www.r-inla.org/}{the \texttt{R-INLA} package} to fit the last two models for spatiotemporal areal data (Equations 7 and 8 in Section 3.3 of \citealt{Blangiardo2013}), denoted as \texttt{INLA.ST1} and \texttt{INLA.STint}, on data simulated from our model with the same mechanism as the first of the three simulation groups mentioned earlier in this section, where we specified $h_s=4$, $c=50$, $\kappa=0.4$, $\rho=0.1$, and on data generated from the aforementioned \texttt{CARBayesST} models with the same mechanisms as in the package manual examples \citep{Lee2018}, where \texttt{ST.CARlocalised} and \texttt{ST.CARadaptive} share a common mechanism. We set $T=50$ with two choices of the number of spatial locations--$m=30^2=900$ and $m=40^2=1600$. 
	
	When fitting \texttt{CARBayesST} models, we specified \texttt{AR=1} for \texttt{ST.CARar} and \texttt{interaction = TRUE} for \texttt{ST.CARanova}. On data generated from our model, we specified the $m\times m$ adjacency matrix $W=(w_{lj})$ as binary, where $w_{lj}=1$ if and only if $\bm{s}_l\in N(\bm{s}_j)$ or $\bm{s}_j\in N(\bm{s}_l)$, and \texttt{G=6} for \texttt{ST.CARlocalised}, where \texttt{G} is the maximum number of distinct clusters permitted, indicated as 3 on data simulated from \texttt{ST.CARlocalised} itself. When fitting our model, we specified $|N(\bm{s}_i)|=h_s=4$ and $w_{ij} = 1/h_s$, $j=1,\ldots,h_s$ in \eqref{ZsepmodelEq} for all $i=1,\ldots,m$. We ran each MCMC chain for $10^4$ burn-in iterations and $10^4$ post-burn-in iterations, thinned to 5000 posterior samples for analysis. For both \texttt{INLA.ST1} and \texttt{INLA.STint}, the \texttt{graph} input for the spatial structured and unstructured components under a Besag-York-Mollie (BYM) specification \citep{Besag1991} is produced from the $W$ adjacency matrix for \texttt{CARBayesST} models. A logGamma$(1,0.1)$ prior is imposed on the four precisions in \texttt{INLA.ST1} and the five precisions in \texttt{INLA.STint}.   
	
	\begin{table}[h!]
		{
			\begin{adjustwidth}{-0cm}{-0cm}
				\begin{center}
					\scalebox{0.86}{\begin{tabular}{|*{8}{c|}} 
							\hline
							\backslashbox{\textbf{Model}}{\textbf{Measure}} & \texttt{runtime} & $\overline{\text{ESS}}$ & MAE & DIC & p.d. & WAIC & p.w. \\
							\hline
							our model & 5976 & $\approx5000$ & 0.8700 & 629253.4 & 76748.78 & 631860.1 & 40400.18\\
							\hline                     
							\texttt{ST.CARar} & 810.5 & 1831.7 & 0.6561196 &  634740.50 &     77915.41  &   615289.45   &   42213.60\\
							\hline                      
							\texttt{ST.CARanova} & 632.6 & 1815.29 & 0.1972 & 631168.50 & 79178.71 & 607686.27 & 40153.84 \\
							\hline
							\texttt{ST.CARsepspatial} & 781.5 & 1834.9 & 0.7036909 & 634851.46 & 77972.09   &  615489.38   &   42307.94\\
							\hline
							\texttt{ST.CARadaptive} & 26339.3 & 1830.7 & 0.7262935 & 634386.92 &     77719.66  &   614979.23   &   42107.87  \\
							\hline
							\texttt{ST.CARlocalised} & 5421.9 & 1825.8 & 1.393124 &  589610.92    &  27372.65  &   625742.84   &   45768.54 \\
							\hline
							\texttt{INLA.STint} & 0.49 & N.A. & 0.7418 & 637908.4 & 78591.9 & 627063.6 & 47732.64\\
							\hline 
					\end{tabular}}
				\end{center}
			\end{adjustwidth}
			\caption{Data simulated from our model with $T=50$ and $m=1600$: overall computation time (\texttt{runtime}) in seconds, mean response ESS ($\overline{\text{ESS}}$) at the 5000 kept post-burn-in MCMC iterations, and diagnostics metrics from our model and adequate \texttt{CARBayesST}, \texttt{INLA} models. \texttt{ST.CARlinear} and \texttt{INLA.ST1} perform very poorly and are thus excluded from the table.}
			\label{m1600T50poisGammaTimeDiags}
		}
	\end{table}
	
	\begin{table}[h!]
		{
			\begin{adjustwidth}{-0cm}{-0cm}
				\begin{center}
					\scalebox{0.81}{\begin{tabular}{|*{8}{c|}} 
							\hline
							\backslashbox{\textbf{Model}}{\textbf{Measure}} & \texttt{runtime} & $\overline{\text{ESS}}$ & MAE & DIC & p.d. & WAIC & p.w. \\
							\hline
							\texttt{ST.CARlinear} & 376.9 & 1758.5 & 0.7615936 & 203867.457   &   1245.449  &  203877.468   &   1231.952\\      
							\hline
							our model & 7164 & 1780.7 & 0.5877898 & 206284.7 & 16745.43 & 227533.2 & 19548.92\\
							\hline    
							\texttt{INLA.ST1} & 0.012 & N.A. & 0.7784977 & 206559.4 & 1063.286 & 206561.6 & 1051.365\\
							\hline
							\texttt{INLA.STint} & 0.37 & N.A. & 0.7676521 & 206537.9 & 2131.546 & 206562.2 & 2098.806\\
							\hline
							\hline
							\texttt{ST.CARsepspatial} & 743.2 & 312.18 & 0.8648668 & 229463.915    &  1759.231  &  229485.080   &   1727.627 \\      
							\hline
							our model & 5472 & $\approx 5000$ & 0.41761 & 240632.6 & 36262.85 & 250083.2 &  23388.25\\ 
							\hline  
							\texttt{INLA.ST1} & 0.08 & N.A. & 0.8761719 & 228936.3 & 506.5141 & 228947.8 & 514.6062\\
							\hline
							\texttt{INLA.STint} & 0.29 & N.A. & 0.8522007 & 227920.5 & 2151.013 & 227970.1 & 2141.197\\
							\hline
							\hline
							\texttt{ST.CARanova} & 571.4 & 1674.8 & 0.8549369 & 225248.6631  &    931.3869 &  225253.3885   &   925.1672 \\      
							\hline
							our model & 6732 & 1889.0 & 0.6350912 & 229445.7 & 19782.83 &  251996.4 & 21661.22\\ 
							\hline 
							\texttt{INLA.ST1} & 0.08 & N.A. & 0.8306218 & 219293.5 & 517.8178 & 219295.3 & 516.1339\\
							\hline
							\texttt{INLA.STint} & 0.29 & N.A. & 0.8192538 & 219275.8 & 1580.859 & 219299.1 &  1572.868\\
							\hline
							\hline
							\texttt{ST.CARar} & 732.6 & 168.17 & 0.7703520 & 213100.853  &    6325.963  &  213398.226   &   6088.797\\      
							\hline
							our model & 6444 & 2025 & 0.5986256 & 213717.5 & 19371.1 & 243104.1 &  24950.1 \\ 
							\hline  
							\texttt{INLA.ST1} & 0.08 & N.A. & 0.8936 & 226253.8 & 871.5126 & 226446 & 1051.778\\
							\hline
							\texttt{INLA.STint} & 0.41 & N.A. & 0.7496993 & 223827 & 13036.55 & 224140.1 & 11520.21\\
							\hline
							\hline
							\texttt{ST.CARadaptive} & 26260.2 & 592.73 & 2.626 & 466999.97  &    28277.32  &   468224.07   &   23389.50 \\      
							\hline
							\texttt{ST.CARlocalised} & 4322.2 & 1101.763 & 2.152 & 468133.15 &  32165.86  &   473010.07    &  28313.31\\      
							\hline
							our model & 6372 & 4183 & 1.610672 & 483729 & 58051.67 & 493412.3 & 34472.09\\ 
							\hline   
							\texttt{INLA.STint} & 1.08 & N.A. & 1.566 & 490315.3 & 54133.67 & 484575.3 & 35720.6\\
							\hline
					\end{tabular}}
				\end{center}
			\end{adjustwidth}
			\caption{Data simulated from five \texttt{CARBayesST} settings (\texttt{ST.CARadaptive} and \texttt{ST.CARlocalised} share the same simulated data set) with $T=50$ and $m=1600$: overall computation time (\texttt{runtime}) in seconds, mean response ESS ($\overline{\text{ESS}}$) at the 5000 kept post-burn-in MCMC iterations, and diagnostics metrics from the corresponding \texttt{CARBayesST} models, our model, \texttt{INLA.ST1}, and \texttt{INLA.STint}. \texttt{INLA.ST1} performs poorly on the last simulated data set, and that row is thus excluded from the table.}
			\label{m1600T50CARBayesSTtimeDiags}
		}
	\end{table}
	
	For each MCMC model, we recorded the overall computation time for $2\times 10^4$ MCMC iterations, denoted as \texttt{runtime}, and calculated the Effective Sample Size (ESS) for fitted response count values at the kept 5000 posterior iterations. Then, computation time per effective response sample, obtained by the \texttt{runtime}/$4\overline{\text{ESS}}$, where $\overline{\text{ESS}}$ denotes the mean of the $mT$ effective sample sizes corresponding to $y_t(\bm{s}_i)$s, would be an adequate metric to gauge the model's computational efficiency. We calculated MAE and considered model fit criteria including the Deviance Information Criterion (DIC), the Watanabe-Akaike Information Criterion (WAIC), and their estimated effective numbers of parameters (p.d for DIC and p.w. for WAIC) to evaluate model fitting performance.   \Cref{m1600T50poisGammaTimeDiags,m1600T50CARBayesSTtimeDiags,m900T50poisGammaTimeDiags,m900T50CARBayesSTtimeDiags} suggest that our MCMC model is comparable in computation time and modeling performance to and yields in all settings clearly higher ESS than \texttt{CARBayesST} MCMC models. Our framework also performs better than or comparably to \texttt{INLA.ST1} and \texttt{INLA.STint}.
	
	We further investigate the prediction capabilities of these Bayesian spatiotemporal count models. Our package, \texttt{spatempBayesCounts}, has a function \texttt{predictNewLocTime()} for out-of-sample predictions executing our methodology in \Cref{sec:prediction}, which takes a fitted model object as one of its arguments and thus needs not refit a model. The \texttt{CARBayesST} package contains no built-in functions for out-of-sample predictions and can only obtain predictions on new data by refitting a model with new responses set to \texttt{NA}. Among the six models, only \texttt{ST.CARlinear}, \texttt{ST.CARanova}, and \texttt{ST.CARar} permit missing (\texttt{NA}) values in the response data and thus support out-of-sample predictions, as they are the simplest in terms of parsimony. \texttt{ST.CARsepspatial}, \texttt{ST.CARadaptive}, and \texttt{ST.CARlocalised} exhibit more complex forms and cannot recover missing data well \citep{Lee2018}. \texttt{INLA.ST1} and \texttt{INLA.STint} also require refitting the model with response data corresponding to new spatial locations and/or time intervals set to \texttt{NA} in order to make out-of-sample predictions.
	
	We simulated data from \texttt{ST.CARlinear}, \texttt{ST.CARanova}, and \texttt{ST.CARar} with the same mechanisms as in the previous simulation experiments. We considered a total of $N=50$ consecutive time intervals of equal length and $K=39^2=1521$ spatial locations $\bm{s}=(i_1,i_2)$ on an equispaced 2-dimensional grid, where $(i_1,i_2)\in\{1,\ldots,39\}\times\{1,\ldots,39\}$. The last $h=2$ of the $N=50$ time intervals were taken as future time periods and had their corresponding generated responses at all locations set to \texttt{NA} when fitting the three \texttt{CARBayesST} models and the two \texttt{R-INLA} models. 
	$r=21$ out of the $K=1521$ locations, whose set of coordinates is $\{10,20,30\}\times \{10,20,30\}\cup \{5,15,25,35\}\times\{15,25\}\cup\{15,25\}\times \{5,35\}$, were set as testing locations (see \Cref{plot:simu}\textbf{A}) and had their corresponding generated responses at all time intervals set to \texttt{NA} when fitting the three \texttt{CARBayesST} models and the two \texttt{R-INLA} models. 
	Simulated response data corresponding to the first $T=N-h=48$ time intervals and the $m=K-r=1500$ training locations were used to fit our model. The function \texttt{predictNewLocTime()} in our package, which takes in the fitted model object and the distance matrix between the $m=1500$ reference and the $r=21$ testing locations, then predicts at the $2$ future time intervals and the $r$ testing locations following \Cref{sec:prediction}.
	
	\begin{table}[h!]
		{
			\begin{adjustwidth}{0cm}{-0cm}
				\begin{center}
					\scalebox{0.84}{\begin{tabular}{|*{7}{c|}}
							\hline
							\multirow{3}{*}{\backslashbox{\textbf{Model}}{\textbf{MAE}}} & \multicolumn{3}{|c|}{$i=1:m$} & \multicolumn{2}{|c|}{$i=(m+1):(m+r)$} & $t=T+1,T+2,\:i=$\\
							\cline{2-6}
							&\multirow{2}{*}{$t=1:T$} &\multirow{2}{*}{$t=T+1$} & $t=(T+1)$ &  \multirow{2}{*}{$t=1:T$} & $t=1:$ &  $1:m$;\hspace{1mm}   $t=1:(T+2)$,\\
							& &  & $:(T+2)$ & & $(T+2)$ & $i=(m+1):(m+r)$\\
							\hline
							\texttt{ST.CARlinear} & 0.7525225 & 0.759521 & 0.7533761 & 0.7884273 & 0.7929358 & 0.7636323\\
							\hline
							\texttt{our model} & 0.5759387 & 0.7593725 & 0.7522299 & 0.794803 & 0.798284 & 0.7641698\\
							\hline
							\texttt{INLA.ST1} & 0.754968 & 1.0900247 & 1.0911743 & 1.1091403 & 1.1122884 & 1.0966483 \\
							\hline
							\texttt{INLA.STint} & 0.7442695 & 1.1031692 & 1.1060514 & 1.1167986 & 1.1204805 & 1.1097923\\
							\hline
							\hline
							\texttt{ST.CARanova} & 0.7715878 & 0.89888 & 0.89300 & 0.7825770 & 0.7818363 & 0.8641791\\
							\hline
							\texttt{our model} & 0.5662226 & 0.8951150 & 0.8867734 & 0.786961 & 0.785446 & 0.8605033\\
							\hline
							\texttt{INLA.ST1} & 0.7730499 & 1.2358673 & 1.2379976 & 1.07116 & 1.07961 & 1.1969354\\
							\hline
							\texttt{INLA.STint} & 0.7624068 & 1.2450 & 1.248512 & 1.075794 & 1.084680 & 1.206037\\
							\hline
							\hline
							\texttt{ST.CARar} & 0.7835062 & 0.871665 & 0.874366 & 0.8882452 & 0.8947664 & 0.8796553\\
							\hline
							\texttt{our model} & 0.6192472 & 0.8716193 & 0.8795939 & 0.936035 & 0.938046 & 0.8947482\\
							\hline
							\texttt{INLA.ST1} & 0.8894 & 1.0689401 & 1.1037297 & 1.1291291 & 1.1471930 & 1.1149979\\
							\hline
							\texttt{INLA.STint} & 0.7428802 & 1.209911 & 1.2661711 & 1.1727726 & 1.1981913 & 1.2485467\\
							\hline
					\end{tabular}}
				\end{center}
			\end{adjustwidth}
			\caption{Data simulated from \texttt{ST.CARlinear}, \texttt{ST.CARanova}, and \texttt{ST.CARar} with $T+h=48+2=50$ consecutive time periods and $m+r=1500+21=1521=39^2$ spatial locations: MAEs for in-sample (the second column) and out-of-sample (the third column to the last column) predictions from the corresponding \texttt{CARBayesST} models, our model, \texttt{INLA.ST1}, and \texttt{INLA.STint}. The last column presents MAE for all testing response data.}
			\label{K1521N50CARBayesSTpredMAE}
		}
	\end{table}
	
	As the prediction MAEs in \Cref{K1521N50CARBayesSTpredMAE} and Figures \ref{plot:simu}\textbf{B}, \ref{plot:absErrorBoxplotSTCARlinear} and \ref{plot:absErrorBoxplotSTCARanova} demonstrate, our model matches the performance of \texttt{ST.CARlinear}, \texttt{ST.CARanova}, and \texttt{ST.CARar} and significantly outperforms \texttt{INLA.ST1} and \texttt{INLA.STint} in predicting responses at future time intervals and unseen spatial locations on data simulated from the three \texttt{CARBayesST} models that support out-of-sample predictions. Despite the fast computation speed and adequate in-sample prediction performance, INLA appears to introduce an approximation error that may not disappear with large samples, thus producing evidently higher MAEs for out-of-sample predictions.
	
	\begin{figure}[h!]
		\centering
		\includegraphics[width=1\textwidth]{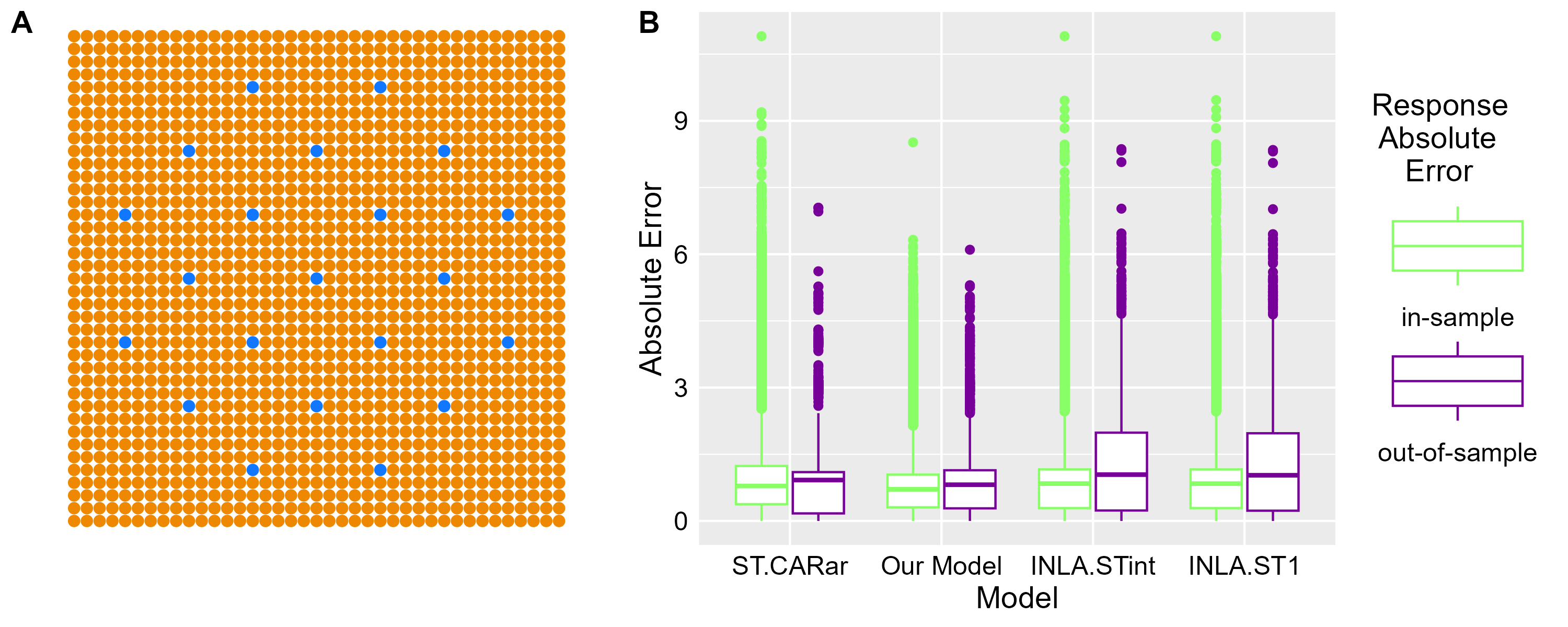}     
		\caption{(\textbf{A}) Data simulated from \texttt{ST.CARlinear}, \texttt{ST.CARanova}, and \texttt{ST.CARar} with $50$ consecutive time periods and $m+r=1521=39^2$ spatial locations: the $r=21$ blue locations are testing ones, and the $m=1500$ orange locations are training ones. (\textbf{B}) Boxplot for the last of the three simulated data sets from \textbf{A} (from \texttt{ST.CARar}): distribution of the absolute errors $\big|y_t(\bm{s}_i)-\hat{y}_t(\bm{s}_i)\big|$, where $\hat{y}_t(\bm{s}_i)$ is the posterior mean for each $(t,i)$.}
		\label{plot:simu} 
	\end{figure}

	\vspace{-1cm}
	\section{World Weekly COVID-19 Data} \label{sec:covid}
	We compare the performance of our proposed Bayesian spatiotemporal count model with several popular LGP-based and CAR-based models in the literature, including two methods fit by \texttt{INLA} \citep{Rueetal09,Blangiardo2013} and six models from the R package \texttt{CARBayesST} \citep{Lee2018}, which also adopts MCMC simulation, on COVID-19 data.
	
	We downloaded from \href{https://data.who.int/dashboards/covid19/data}{the WHO website} Weekly COVID-19 Cases and Deaths data of 240 countries/territories since January 2020. We considered the 145 countries/territories with neither NA cases nor NA deaths in all 123 weeks from May 24, 2020, to September 25, 2022, both ends included. Negative counts for weekly new cases/deaths were taken as recording errors and changed to their absolute values. After extensive trials, population turned out to be an external covariate that should be included in the model, and yearly population by country/territory data was obtained from \href{https://data.worldbank.org/indicator/SP.POP.TOTL?end=2022&start=1960}{the World Bank website}. 
	
	As geographical proximity does not amount to heavy bilateral human traffic between two countries, it is likely inappropriate to specify each country's neighbor set based on geographical distances between countries, which was corroborated by our tests and experiments. Instead, we retrieved from \href{https://wits.worldbank.org/CountryProfile/en/Country/CHN/Year/2022/TradeFlow/EXPIMP/Partner/by-country/Product/Total}{World Integrated Trade Solution (WITS)} all products' import \& export data in the latest available year for all countries of interest and specified for each country $\bm{s}_i$ its $h_s=8$ neighbor countries as the ones among our retained $m=123$ countries that contribute the most to $\bm{s}_i$'s import dollars. For each country $\bm{s}_i$, $w_{ij}$ in \eqref{ZsepmodelEq} was calculated as $\bm{s}_i$'s latest available year import dollars from its $j$th neighbor $N(\bm{s}_i)[j]$ divided by $\bm{s}_i$'s latest available year import dollars from all its $h_s$ neighbor countries.
	
	Due to zero new COVID-19 cases for many consecutive weeks since May 24, 2020, suspicious new counts fluctuations, a lack of reliable population data, and/or unavailable recent yearly WITS import \& export data, we excluded 22 countries/territories, whose ISO 3 codes are WLF, XKX, COK, NIU, PCN, TKL, ASM, PRK, MHL, NRU, FSM, SLB, TON, TUV, VUT, PLW, WSM, KIR, GUM, PRI, MNP, and HTI, from the 145 originally retained countries/territories. Our training data contain weekly new COVID-19 cases and deaths of $m=123$ countries/territories in $T=105$ weeks, from May 24, 2020, to May 22, 2022, both ends included. 
	After comprehensive tests and detailed analyses, we decided to include fixed scaled population as the only external covariate by fixing $\bm{\beta}$ and setting $\exp\{\bm{x}_t(\bm{s}_i)\bm{\beta}\}=N_t(\bm{s}_i) / k$ in \eqref{yPois}, where $N_t(\bm{s}_i)$ denotes country/territory $\bm{s}_i$'s population in the year among 2020, 2021, and 2022 within which the $t$th week falls, and $k$ is a constant specified as 5000 for new cases and $5\times 10^5$ for new deaths. 
	
	For each training data set of either weekly new cases or weekly new deaths, we fit our model with $2\times 2=4$ settings for spatiotemporal frailties $U^t(\bm{s}_i)$'s -- with ($\rho> 0$, \texttt{autoreg}) or without ($\rho = 0$, \texttt{noself}) autoregression; specify each country $\bm{s}_i$'s neighbor set and $w_{ij}$'s in \eqref{ZsepmodelEq} either based on the country's latest available year import data as detailed in the third paragraph of this section (\texttt{WITS}) or based on the correlations of the actual observed responses in the $T$ weeks between $\bm{s}_i$ and the other $m-1$ countries (\texttt{Cor}). The four settings are denoted as \texttt{autoregWITS}, \texttt{autoregCor}, \texttt{noselfWITS}, \texttt{noselfCor}, respectively. 
	We set $h_s=8$, $\alpha = 1.0001$, and $(\alpha_c, \theta_c) = (2,10)$. For \texttt{autoreg}, i.e., we follow our standard spatiotemporal frailty model \eqref{UmodelEq} and \eqref{ZsepmodelEq} so that all $m$ diagonal entries in the $\bfV$ matrix in \eqref{eq:ARG1} equal $\rho>0$, we specified $(a_{\kappa}, b_{\kappa}) = (0.55,1)$ and $(a_{\rho}, b_{\rho}) = (0.4,1)$; for \texttt{noself}, i.e., $\rho=0$ in \eqref{eq:V.rho.kappa} so that all $m$ diagonal entries in  $\bfV$ equal 0, we specified $(a_{\kappa}, b_{\kappa}) = (0.9,1)$.
	We ran each MCMC chain for $2\times 10^4$ burn-in iterations and $10^4$ post-burn-in iterations, thinned to 5000 posterior samples for analysis. 	
	
	\begin{figure}[h!]
		\centering
		\includegraphics[width=0.9\textwidth]{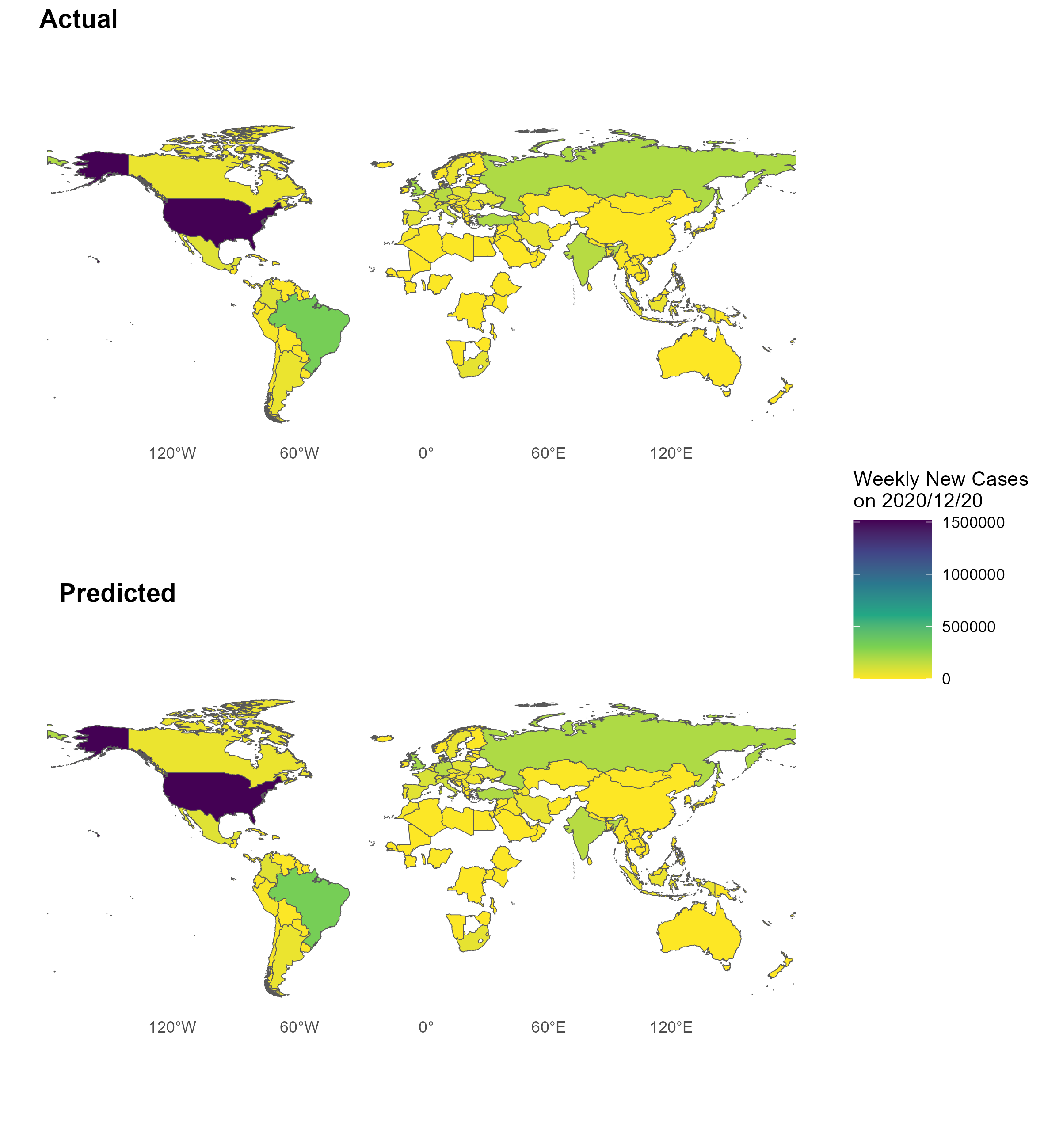}     
		\caption{Actual and predicted (from the \texttt{noselfWITS} setting) new COVID-19 cases in the $m=123$ training countries/territories in the week reported on December 20, 2020. The shapefile is obtained from \href{https://public.opendatasoft.com/explore/dataset/world-administrative-boundaries/export/}{Opendatasoft}.}
		\label{covidCases}
	\end{figure}
	
	\Cref{covidCasesPara,covidDeathsPara} present summary statistics of actual and predicted count responses and posterior samples of scalar parameters $c$, $\kappa$, and $\rho$. All four settings of our model produce good fits, as demonstrated by the close alignment between the estimated responses and the actual counts and the consistent and adequate scalar parameter estimates. \Cref{covidCases,covidDeaths} further showcase our model setting \texttt{noselfWITS}'s decent in-sample prediction performance in all 123 training countries/territories in the week with reporting date December 20, 2020. The other three settings of our model lead to satisfactory outcomes, too.
	
	As a comparison, we implemented the eight alternative Bayesian spatiotemporal count models carried out in \Cref{sec:simulation} on the two data sets, which all included population (in 2022, 2023, and 2024) by country/territory as the only external covariate and adopted the Poisson family. 
	For all six \texttt{CARBayesST} models, we ran $3\times 10^4$ MCMC iterations, with the first $2\times 10^4$ discarded as burn-in, and specified two choices for the $m\times m$ adjacency or weight matrix $W$, which is required by \texttt{CARBayesST} to be symmetric, nonnegative, and produce a positive sum for each row. For each of the two data sets, we first calculate the $m\times m$ correlation matrix $(\rho_{lj})$ of the actual observed counts in the $T$ weeks for the $m$ countries/territories. One $W_{m\times m}$ choice simply replaces all negative entries in the correlation matrix by 0 and is denoted as \texttt{Wcorr}. The other $W = (w_{lj})$ choice, denoted as \texttt{Wbinary}, is an $m\times m$ binary matrix whose diagonal entries are all 0. For all $l\neq j$, $w_{lj}=1$ if and only if $\rho_{lj}$ is among the $h_s=8$ largest values in $\{\rho_{lk}:k=1,\ldots,m,\,k\neq l\}$ or $\{\rho_{hj}:h=1,\ldots,m,\,h\neq j\}$. For both data sets, we specified \texttt{AR=1} for \texttt{ST.CARar} and \texttt{interaction = TRUE} for \texttt{ST.CARanova}. For \texttt{ST.CARlocalised}, we specified \texttt{G=6} for new cases and \texttt{G=3} for new deaths.    
	For both \texttt{INLA.ST1} and \texttt{INLA.STint}, the \texttt{graph} input under a spatial BYM specification \citep{Besag1991} is produced from the \texttt{Wbinary} matrix. A logGamma$(1,0.1)$ prior is imposed on the four precisions in \texttt{INLA.ST1} and the five precisions in \texttt{INLA.STint}. 
	
	To evaluate the adequacy of the aforementioned models, we calculated MAE and Mean Absolute Percentage Error (MAPE). Let $\hat{y}_t(\bm{s}_i)$ denote the estimated expectation of the response, where the posterior mean is taken for each MCMC method, in time interval $t$ and country/territory $\bm{s}_i$ for each $(i,t)$, then $\text{MAE} = \frac{1}{mT}\sum_{t=1}^T\sum_{i=1}^m \big|y_t(\bm{s}_i)-\hat{y}_t(\bm{s}_i)\big|$ and $\text{MAPE} = \frac{100\%}{N}\sum_{(i,t) \text{ s.t. } y_t(\bm{s}_i)>0}\Big|\frac{y_t(\bm{s}_i)-\hat{y}_t(\bm{s}_i)}{y_t(\bm{s}_i)}\Big|$, where $N\leq mT,\,N\in\mathbb{N}$ is the number of $(i,t)$ pairs with non-zero actual $y_t(\bm{s}_i)$'s. Out of the $mT=123\times 105=12915$ actual $y_t(\bm{s}_i)$'s, 178 weekly COVID-19 new cases are 0, and 1758 weekly COVID-19 new deaths are 0. We further considered DIC, WAIC, and their estimated effective numbers of parameters. 
	
	No matter whether the covariate population is unscaled or scaled, \texttt{ST.CARlinear}, \texttt{ST.CARanova}, and \texttt{ST.CARsepspatial}, with both aforementioned choices for $W_{m\times m}$, and \texttt{INLA.ST1} all produced very unsatisfactory model fitting results and were thus excluded from our analyses. With population scaled, among \texttt{ST.CARar}, \texttt{ST.CARadaptive}, and \texttt{ST.CARlocalised}, with both \texttt{Wbinary} and \texttt{Wcorr} for $W_{m\times m}$, and \texttt{INLA.STint}, the only setting that produced an adequate model fit significantly better than its counterpart with unscaled population is \texttt{INLA.STint} applied to the weekly COVID-19 new deaths data. Most other model settings led to similarly good results with unscaled and scaled population. We thus only present results with the original unscaled population from settings except \texttt{INLA.STint} for the new deaths data, which divided population by $5\times 10^5$.
	
	\Cref{covidCasesDiags,covidDeathsDiags} present for both data sets metrics obtained from more suitable ones among all model settings we experimented with. These adequate settings include \texttt{autoregWITS}, \texttt{autoregCor}, \texttt{noselfWITS}, and \texttt{noselfCor} from our model, \texttt{ST.CARar}, \texttt{ST.CARadaptive}, and \texttt{ST.CARlocalised}, with both \texttt{Wbinary} and \texttt{Wcorr} for $W_{m\times m}$, from \texttt{CARBayesST}, and \texttt{INLA.STint} fit by \texttt{R-INLA}. \Cref{covidCasesDiags,covidDeathsDiags} suggest that for both data sets, our model's four settings yield performance comparable to that of the best-suited models from R packages \texttt{CARBayesST} and \texttt{R-INLA}. Moreover, our model clearly beats \texttt{CARBayesST} models for the world weekly COVID-19 new cases data set (\Cref{covidCasesDiags}).
	
	\begin{table}[h!]
		{
			\begin{adjustwidth}{-0cm}{-0cm}
				\begin{center}
					\scalebox{0.86}{\begin{tabular}{|*{7}{c|}}
							\hline
							\backslashbox{\textbf{Method}}{\textbf{Metric}} & MAE & MAPE & DIC & p.d. & WAIC & p.w. \\
							\hline
							\texttt{autoregWITS} & 2.248796 & 1.35548 & 150922.8 & 12612.84 & 151018.5 & 6656.463\\
							\hline
							\texttt{autoregCor} & 2.272264 & 1.354982 & 150932.4 & 12617.63 & 150843.5 & 6565.742\\
							\hline
							\texttt{noselfWITS} & 1.968359 & 1.422978 & 150924.3 & 12620.48 & 150888.7 & 6557.625\\
							\hline
							\texttt{noselfCor} & 2.005899 & 1.41713 & 150925.8 & 12620.46 & 150870.1 & 6574.339\\
							\hline
							\texttt{INLA.STint} & 0.632495 & 1.91249 & 151429.8 & 12868.29 & 148388.9 & 6991.496\\
							\hline                      
							\texttt{ST.CARarWbinary} & 4.920659 & 2.988448 & 153774.1   &  12346.52  &  155789.8  &    9610.862 \\
							\hline                      
							\texttt{ST.CARarWcorr} & 4.531402 & 1.714859 & 158076.23  &   12547.22 &   160361.6 &     9826.647 \\
							\hline
							\texttt{ST.CARadaptiveWbinary} & 5.874778 & 2.233535 & 151254.15   &   12487.26 &    148417.63   &    6918.49  \\
							\hline
							\texttt{ST.CARadaptiveWcorr} & 15.64428 & 2.579549 & 158116.5   &   15781.17  &   216239.24  &    40352.00  \\
							\hline
							\texttt{ST.CARlocalisedWbinary} & 7.098764 & 0.5548995 & 190523.03 &     32576.98   & 1028955.88 &    445788.63\\
							\hline
							\texttt{ST.CARlocalisedWcorr} & 16.8224 & 0.7096886 & 311732.1   &   92927.66 &  20371494.59 &  10115176.62\\
							\hline
					\end{tabular}}
				\end{center}
			\end{adjustwidth}
			\caption{World weekly COVID-19 new cases: diagnostics metrics from 11 adequate model settings.}
			\label{covidCasesDiags}
		}
	\end{table}
	
	We now consider predicting in $h=18$ new weeks, from May 29, 2022, to September 25, 2022, both ends included, in the $m=123$ training countries/territories. We also consider predicting in the $T+h=105+18=123$ weeks in $r=4$ new countries/territories--whose ISO 3 codes are HTI, GUM, MNP, and PRI--that do not exhibit zero new COVID-19 cases for multiple consecutive weeks since May 24, 2020. For all settings of our model, each new location $\bm{s}$'s neighbor set $N(\bm{s})$ and $w_{j}(\bm{s})$'s in \eqref{predZU} are obtained based on the correlations of the actual observed responses in the $T+h$ weeks between $\bm{s}$ and the $m$ reference countries/territories. 
	As delineated in \Cref{sec:simulation}, among all models in \texttt{CARBayesST}, only \texttt{ST.CARlinear}, \texttt{ST.CARanova}, and \texttt{ST.CARar} support out-of-sample predictions. 
	Hence, \texttt{ST.CARar} is the only \texttt{CARBayesST} model that fits the world COVID-19 new cases/deaths data adequately and is capable of predicting in new countries/territories and/or future weeks. It, however, performs very poorly (\Cref{covidCasesPred,covidDeathsPred}) in out-of-sample predictions. 
	
	
	\begin{table}[h!]
		{
			\begin{adjustwidth}{0cm}{-0cm}
				\begin{center}
					\scalebox{0.82}{\begin{tabular}{|*{7}{c|}}
							\hline
							\multirow{3}{*}{\backslashbox{\textbf{Method}}{\textbf{MedAE}}} & \multicolumn{4}{|c|}{$i=1:m$} & \multicolumn{2}{|c|}{$i=(m+1):(m+r)$} \\
							\cline{2-7}
							&\multirow{2}{*}{$t=T+1$} & $t=(T+1)$ &  $t=(T+1)$ & $t=(T+1)$ & \multirow{2}{*}{$t=1:T$} & \multirow{2}{*}{$t=1:(T+h)$} \\
							&   & $:(T+2)$ & $:(T+5)$ & $:(T+h)$ & & \\
							\hline
							\texttt{autoregWITS} & 3024.25 & 3651.0 & 5980.5 & 10499.6 & 600.25 & 601.00\\
							\hline
							\texttt{autoregCor} & 2933.86 & 3672.5 & 6031.9 & 10468.6 & 598.19 & 604.88 \\
							\hline
							\texttt{noselfWITS} & 15065.6 & 15685.7 & 16437.2 & 15415.3 & 733.73 & 723.05\\
							\hline
							\texttt{noselfCor} & 9878 & 10160 & 12211 & 12128 & 470.67 & 471.36\\
							\hline
							\texttt{INLA.STint} & 1362.1 & 1435.3 & 1727.7 & 2764.3 & 183.83 & 211.84\\
							\hline                      
							\texttt{ST.CARarWbinary} & $4.547\times 10^8$ & $4.610\times 10^8$ & $4.887\times 10^8$ & $4.718\times 10^8$ & $8.267\times 10^7$ & $9.014\times 10^7$ \\
							\hline                      
							\texttt{ST.CARarWcorr} & $5.382\times 10^8$ & $5.364\times 10^8$ & $5.392\times 10^8$ & $4.958\times 10^8$ & $7.874\times 10^7$ & $7.766\times 10^7$\\
							\hline
					\end{tabular}}
				\end{center}
			\end{adjustwidth}
			\caption{World weekly COVID-19 new cases: MedAE for out-of-sample predictions from 7 adequate model settings capable of estimating at new time intervals and spatial locations.}
			\label{covidCasesPred}
		}
	\end{table}
	\begin{table}[h!]
		{
			\begin{adjustwidth}{0cm}{-0cm}
				\begin{center}
					\scalebox{0.86}{\begin{tabular}{|*{7}{c|}}
							\hline
							\multirow{3}{*}{\backslashbox{\textbf{Method}}{\textbf{MedAE}}} & \multicolumn{4}{|c|}{$i=1:m$} & \multicolumn{2}{|c|}{$i=(m+1):(m+r)$} \\
							\cline{2-7}
							&\multirow{2}{*}{$t=T+1$} & $t=(T+1)$ &  $t=(T+1)$ & $t=(T+1)$ & \multirow{2}{*}{$t=1:T$} &  \multirow{2}{*}{$t=1:(T+h)$} \\
							&   & $:(T+2)$ & $:(T+5)$ & $:(T+h)$ & & \\
							\hline
							\texttt{autoregWITS} & 17.257 & 20.798 & 31.278 & 56.406 & 5.2839 & 5.2839\\
							\hline
							\texttt{autoregCor} & 17.4479 & 20.1491 & 31.547 & 56.599 & 5.35063 & 5.35063\\
							\hline
							\texttt{noselfWITS} & 81.698 & 100.437 & 122.267 & 127.852 & 6.7714 & 4.7059 \\
							\hline
							\texttt{noselfCor} & 74.75 & 86.27 & 99.818 & 103.541 & 6.7408 & 5.8654\\
							\hline
							\texttt{INLA.STint} & 8.898 & 8.630 & 8.536 & 10.857 & 2.1010 & 2.3678\\
							\hline                      
							\texttt{ST.CARarWbinary} & 3632411 & 3537483 & 3705387 & 3758511 & 424443 & 424443\\
							\hline                      
							\texttt{ST.CARarWcorr} & 4377532 & 4426582 & 4361354 & 4310342 & 563653 & 563653\\
							\hline
					\end{tabular}}
				\end{center}
			\end{adjustwidth}
			\caption{World weekly COVID-19 new deaths: MedAE for out-of-sample predictions from 7 adequate model settings capable of estimating at new time intervals and spatial locations.}
			\label{covidDeathsPred}
		}
	\end{table}
	
	\Cref{covidCasesPred,covidDeathsPred} present the Median Absolute Error $\text{MedAE} = \text{median} \big(\big|y_t(\bm{s}_i)-\hat{y}_t(\bm{s}_i)\big|\big)$ for out-of-sample predictions in future weeks and/or new countries/territories. Although our model underperforms relative to \texttt{INLA.STint}, the performance remains comparable in scale. Our model beats \texttt{CARBayesST}, as each of its functions either cannot conduct or performs very unsatisfactorily in out-of-sample predictions. \Cref{covidCasesPredm127NewT1,covidDeathsPredm127NewT1} demonstrate our model setting \texttt{autoregWITS}'s adequate out-of-sample prediction performance in all 127 training and testing countries/territories in the first future week with reporting date May 29, 2022. The other three settings of our model produce satisfactory results as well.
	


	\section{Discussion} \label{sec:discussion}
	
	The current work can be extended in several directions. First, although we have modeled the count data using Poisson random variables, the proposed Bayesian framework and sampling strategies are readily applicable to related distributions, such as the zero-inflated Poisson distribution for datasets with many zero responses and the negative binomial distribution viewed as a gamma mixture of Poisson distributions. Second, our theory provides sufficient conditions on the weighting matrix to ensure stationarity of the frailty process. So far, we have investigated the empirical performance of only one such graph built upon nearest neighbors. In general, one can specify other types of sparse graphs for the frailties across spatial and temporal dimensions, as long as the process remains stationary. Third, our current spatiotemporal model is defined on a finite set of time intervals and spatial locations connected through a graph of weights,  which does not directly extend to a valid stochastic process over a continuous spatial domain. To achieve such an extension, one possible approach is to represent the latent gamma frailties as the sum of squares of several independent latent Gaussian processes, as in the recent work on Poisson point processes by \citet{Moretal24}. We leave these directions to future research.

    	\appendix
    	\section{Deductions Leading to Posterior Sampling Steps for the Model Given by \eqref{yPois}, \eqref{UmodelEq}, and \eqref{ZsepmodelEq}} \label{appenA}
    	We deduce in detail posterior sampling steps for the model given by \eqref{yPois}, \eqref{UmodelEq}, and \eqref{ZsepmodelEq}.
    	\begin{enumerate}
    		\item Sampling from the full conditional distribution of $c$ with prior $\pi(c)\sim\IG(\alpha_c,\theta_c)$:\\
    		If $m>1$ and $T>1$, then
    		\begin{align}\label{postCsupp}
    			f(c|\cdot) &\propto \pi(c) \times \prod_{i=1}^mf\left(U^1(\bm{s}_i)~|~c\right) \times\prod_{t=2}^T\prod_{i=1}^m  f\left(U^t(\bm{s}_i)~|~Z^{t-1}_{i0},Z^{t-1}_{i1},\ldots,Z^{t-1}_{i|N(\bm{s}_i)|},c\right) \times\nonumber\\
    			& \quad \prod_{t=2}^T\prod_{i=1}^m f\left(Z_{i0}^{t-1}~|~U^{t-1}(\bm{s}_i),\rho,c\right) \times\prod_{t=2}^T\prod_{i=1}^m \prod_{j=1}^{|N(\bm{s}_i)|}f\left(Z^{t-1}_{ij}~|~U^{t-1}_{N(\bm{s}_i)[j]},\kappa,c\right) \nonumber\\
    			& \propto \left(\frac{1}{c}\right)^{\alpha_c+1}\exp\left\lbrace-\frac{\theta_c}{c}\right\rbrace \times\left(\frac{1}{c}\right)^{m\alpha}\exp\left\lbrace-\frac{1}{c}\sum_{i=1}^m U^1(\bm{s}_i)\right\rbrace \times \nonumber\\
    			&\quad\prod_{t=2}^T \prod_{i=1}^m \left(\frac{1}{c}\right)^{\alpha+\sum_{j=0}^{|N(\bm{s}_i)|}Z^{t-1}_{ij}}\exp\left\lbrace -\frac{U^t(\bm{s}_i)}{c}\right\rbrace \cdot \left(\frac{1}{c}\right)^{Z^{t-1}_{i0}}\exp\left\lbrace -\frac{1}{c}\rho U^{t-1}(\bm{s}_i)\right\rbrace \times \nonumber\\
    			&\quad \prod_{t=2}^T \prod_{i=1}^m\prod_{j=1}^{|N(\bm{s}_i)|} \left(\frac{1}{c}\right)^{Z^{t-1}_{ij}}\exp\left\lbrace -\frac{1}{c}\left[\kappa w_{ij}U^{t-1}_{N(\bm{s}_i)[j]}\right]\right\rbrace \nonumber\\
    			&\propto\left(\frac{1}{c}\right)^{Tm\alpha+2\sum_{t=2}^{T}\sum_{i=1}^m \sum_{j=0}^{|N(\bm{s}_i)|} Z^{t-1}_{ij}+\alpha_c+1} \times \exp\left\lbrace -\frac{\theta_c^{\text{post}}}{c}\right\rbrace\nonumber\\
    			&\sim \IG\left(\alpha_c+ Tm\alpha+2\sum_{t=1}^{T-1}\sum_{i=1}^m \sum_{j=0}^{|N(\bm{s}_i)|} Z^t_{ij},\: \theta_c^{\text{post}}\right),
    		\end{align}
    		where $\theta_c^{\text{post}}=\theta_c+\sum\limits_{t=1}^T\sum\limits_{i=1}^m U^t(\bm{s}_i) + \rho\cdot\sum\limits_{t=1}^{T-1}\sum\limits_{i=1}^m U^t(\bm{s}_i) + \kappa\cdot\sum\limits_{i=1}^m  \sum\limits_{j=1}^{|N(\bm{s}_i)|}w_{ij}\sum\limits_{t=1}^{T-1}U^t_{N(\bm{s}_i)[j]}$.
    		
    		If $T=1$, then
    		\begin{align}
    			f(c|\cdot) & \propto \pi(c) \times \prod_{i=1}^mf\left(U^1(\bm{s}_i)~|~c\right)  \nonumber\\
    			&\propto \left(\frac{1}{c}\right)^{m\alpha}\exp\left\lbrace-\frac{1}{c}\sum_{i=1}^m U^1(\bm{s}_i)\right\rbrace\times \left(\frac{1}{c}\right)^{\alpha_c+1}\exp\left\lbrace-\frac{\theta_c}{c}\right\rbrace  \nonumber\\
    			&\propto\left(\frac{1}{c}\right)^{m\alpha+\alpha_c+1}  \exp\left\lbrace -\frac{1}{c}\left[\sum_{i=1}^m U^1(\bm{s}_i)+\theta_c\right]\right\rbrace\nonumber\\ 
    			&\sim  \IG\left(m\alpha+\alpha_c,\: \sum_{i=1}^m U^1(\bm{s}_i) +\theta_c\right). \nonumber
    		\end{align}
    		If $m=1$ and $T>1$, then
    		\begin{align*}
    			f(c|\cdot) &\propto \pi(c) \times f\left(U^1(\bm{s}_1)~|~c\right) \times\prod_{t=2}^T  f\left(U^t(\bm{s}_1)~|~Z^{t-1}_{10},c\right) \times \prod_{t=1}^{T-1} f\left(Z_{10}^{t}~|~U^{t}(\bm{s}_1),\rho,c\right)  \\
    			&\propto\left(\frac{1}{c}\right)^{T\alpha+2\sum_{t=1}^{T-1}Z_{10}^t+\alpha_c+1}  \exp\left\lbrace -\frac{1}{c}\left[\sum_{t=1}^T U^t(\bm{s}_1)+\rho\cdot\sum_{t=1}^{T-1}U^t(\bm{s}_1)+\theta_c\right]\right\rbrace \\
    			&\sim  \IG\left(T\alpha+2\sum_{t=1}^{T-1}Z_{10}^t+\alpha_c,\: \sum_{t=1}^T U^t(\bm{s}_1)+\rho\cdot\sum_{t=1}^{T-1}U^t(\bm{s}_1)+\theta_c\right).
    		\end{align*}
    		\item Sampling from the full conditional distributions of $U^{1:T}(\bm{s}_{1:m})$:\par Fix any arbitrary $i\in\{1,2,\ldots,m\}$. For each $l\in\{1,2,\ldots,m\}$ such that $\bm{s}_i\in N(\bm{s}_l)$, we denote $k_l(\bm{s}_i)$ as the positive integer less than or equal to $h_s$ such that $\bm{s}_i=N(\bm{s}_l)[k_l(\bm{s}_i)]$. If $m>1$ and $T>1$, then for any $t\in\{2,3,\ldots,T-1\}$,
    		\begin{align}\label{postUsupp}
    			&f\left(U^t(\bm{s}_i)|\cdot\right)  \propto f\left(y_t(\bm{s}_i)~|~U^t(\bm{s}_i),\bm{\beta}\right) \times f(U^t(\bm{s}_i)~|~Z^{t-1}_{i0},Z^{t-1}_{i1},\ldots,Z^{t-1}_{i|N(\bm{s}_i)|},c)\times\nonumber\\
    			&\qquad f\big(Z^t_{i0}~|~U^t(\bm{s}_i),\rho,c\big)\times\prod_{l:\bm{s}_i\in N(\bm{s}_l)}^{1\leq l\leq m} f\big(Z^t_{lk_l(\bm{s}_i)}~|~U^t(\bm{s}_i),\kappa,c\big) \nonumber\\
    			&\quad \propto U^t(\bm{s}_i)^{y_t(\bm{s}_i)}\exp\left\lbrace - U^t(\bm{s}_i)\left[ e^{\bm{x}_t(\bm{s}_i)^{\T}\bm{\beta}}\right]\right\rbrace \times U^t(\bm{s}_i)^{\alpha+\sum\limits_{j=0}^{|N(\bm{s}_i)|}Z^{t-1}_{ij}-1}\exp\left\lbrace -\frac{U^t(\bm{s}_i)}{c}\right\rbrace  \times      \nonumber \\
    			& \qquad U^t(\bm{s}_i)^{Z^t_{i0}}\exp\left\lbrace-U^t(\bm{s}_i) \cdot \frac{\rho}{c} \right\rbrace  \times \prod_{l:\bm{s}_i\in N(\bm{s}_l)}^{1\leq l\leq m} U^t(\bm{s}_i)^{Z^t_{lk_l(\bm{s}_i)}}\exp\left\lbrace-w_{lk_l(\bm{s}_i)}U^t(\bm{s}_i) \cdot \frac{\kappa}{c} \right\rbrace  \nonumber\\
    			&\quad \propto U^t(\bm{s}_i)^{y_t(\bm{s}_i)+\alpha+\sum_{j=0}^{|N(\bm{s}_i)|}Z^{t-1}_{ij}-1+ Z^t_{i0}+\sum_{l:\bm{s}_i\in N(\bm{s}_l)}^{1\leq l\leq m}Z^t_{lk_l(\bm{s}_i)}}\times\nonumber\\
    			&\qquad \exp\left\lbrace -U^t(\bm{s}_i)\left[ e^{\bm{x}_t(\bm{s}_i)^{\T}\bm{\beta}}+\frac{1}{c}\left(1+\rho+\kappa\sum_{l:\bm{s}_i\in N(\bm{s}_l)}^{1\leq l\leq m}w_{lk_l(\bm{s}_i)}\right)\right]\right\rbrace  \nonumber\\
    			&\quad \sim \Ga\left(y_t(\bm{s}_i)+\alpha+\sum_{j=0}^{|N(\bm{s}_i)|}Z^{t-1}_{ij}+ Z^t_{i0}+\sum_{l:\bm{s}_i\in N(\bm{s}_l)}^{1\leq l\leq m}Z^t_{lk_l(\bm{s}_i)},\: \text{rate}_{U^t(\bm{s}_i)}\right), 
    		\end{align}
    		where $\text{rate}_{U^t(\bm{s}_i)}= e^{\bm{x}_t(\bm{s}_i)^{\T}\bm{\beta}}+\frac{1}{c}\left(1+\rho+\kappa\sum_{l:\bm{s}_i\in N(\bm{s}_l)}^{1\leq l\leq m}w_{lk_l(\bm{s}_i)}\right)$.       
    		\begin{align*}             
    			&f\left(U^t(\bm{s}_i)|\cdot\right) \sim \Ga\left(y_t(\bm{s}_i)+\alpha+ Z^t_{i0}+\sum_{l:\bm{s}_i\in N(\bm{s}_l)}^{1\leq l\leq m}Z^t_{lk_l(\bm{s}_i)},\: \text{rate}_{U^t(\bm{s}_i)}\right)\text{  for } t=1,\text{ and}\\           
    			&f\left(U^t(\bm{s}_i)|\cdot\right) \sim \Ga\left(y_t(\bm{s}_i)+\alpha+\sum_{j=0}^{|N(\bm{s}_i)|}Z^{t-1}_{ij},\: e^{\bm{x}_t(\bm{s}_i)^{\T}\bm{\beta}}+\frac{1}{c}\right)\text{ for }t=T.\\
    			&\text{If }T=1,\text{ then }
    			f\left(U^1(\bm{s}_i)|\cdot\right) \sim \Ga\left(y_1(\bm{s}_i)+\alpha,\: e^{\bm{x}_1(\bm{s}_i)^{\T}\bm{\beta}}+\frac{1}{c}\right).\\
    			&\text{If }m=1\text{ and }T>1,\text{ then for any }t\in\{2,\ldots,T-1\},\\
    			&f\left(U^t(\bm{s}_1)|\cdot\right) \sim \Ga\left(y_t(\bm{s}_1)+\alpha+Z_{10}^{t-1}+Z_{10}^t,\: e^{\bm{x}_t(\bm{s}_1)^{\T}\bm{\beta}}+\frac{1}{c}(1+\rho)\right),\\
    			&f\left(U^t(\bm{s}_1)|\cdot\right) \sim \Ga\left(y_t(\bm{s}_1)+\alpha+Z_{10}^{t},\: e^{\bm{x}_t(\bm{s}_1)^{\T}\bm{\beta}}+\frac{1}{c}(1+\rho)\right)\text{ for }t=1,\text{ and }\\
    			&f\left(U^t(\bm{s}_1)|\cdot\right) \sim \Ga\left(y_t(\bm{s}_1)+\alpha+Z_{10}^{t-1},\: e^{\bm{x}_t(\bm{s}_1)^{\T}\bm{\beta}}+\frac{1}{c}\right)\text{ for }t=T.		
    		\end{align*}
    		\item Sampling from the full conditional distribution of $\rho$ with prior \\$\pi(\rho)\sim\tGa(a_{\rho},b_{\rho},0,1)$ when $T>1$:
    		\begin{align}\label{postRhoSupp}
    			f\left(\rho|\cdot\right) &\propto \pi(\rho) \times \prod_{t=1}^{T-1}\prod_{i=1}^m f\left(Z_{i0}^t~|~U^t(\bm{s}_i),\rho,c\right)  \nonumber\\
    			&\propto \rho^{a_{\rho}-1}\exp\{-b_{\rho}\cdot\rho\}\mathbbm{1}_{\{0<\rho<1\}} \times \prod_{t=1}^{T-1}\prod_{i=1}^m \rho^{Z^t_{i0}}\exp\left\lbrace-\frac{\rho}{c}\cdot U^t(\bm{s}_i)\right\rbrace   \nonumber \\
    			&\propto \rho^{a_{\rho}+\sum_{t=1}^{T-1}\sum_{i=1}^m Z^t_{i0}-1}\exp\left\lbrace -\rho\left[b_{\rho}+\frac{1}{c}\sum_{t=1}^{T-1}\sum_{i=1}^mU^{t}(\bm{s}_i)\right]\right\rbrace \mathbbm{1}_{\{0<\rho<1\}} \nonumber \\
    			& \sim \tGa\left(a_{\rho}+\sum_{t=1}^{T-1}\sum_{i=1}^m Z^t_{i0}, \: b_{\rho}+\frac{1}{c}\sum_{t=1}^{T-1}\sum_{i=1}^mU^{t}(\bm{s}_i),\:0,\:1\right).
    		\end{align}
    		\item Sampling from the full conditional distribution of $\kappa$ with prior \\$\pi(\kappa)\sim\tGa(a_{\kappa},b_{\kappa},0,\,1-\rho)$ when $m>1$ and $T>1$:
    		\begin{align}\label{postKappaSupp}
    			&f\left(\kappa|\cdot\right) \propto \pi(\kappa) \times \prod_{t=1}^{T-1}\prod_{i=1}^m \prod_{j=1}^{|N(\bm{s}_i)|} f\left(Z_{ij}^t~|~U^t_{N(\bm{s}_i)[j]},\kappa,c\right) \nonumber\\
    			&\propto \kappa^{a_{\kappa}-1}\exp\{-b_{\kappa}\cdot\kappa\}\mathbbm{1}_{\{0<\kappa<1-\rho\}} \times \prod_{t=1}^{T-1}\prod_{i=1}^m \prod_{j=1}^{|N(\bm{s}_i)|} \kappa^{Z^t_{ij}}\exp\left\lbrace-\frac{\kappa}{c}\cdot w_{ij}U^t_{N(\bm{s}_i)[j]}\right\rbrace   \nonumber \\
    			&\propto \kappa^{a_{\kappa}+\sum_{t=1}^{T-1}\sum_{i=1}^m \sum_{j=1}^{|N(\bm{s}_i)|}Z^t_{ij}-1}\exp\left\lbrace -\kappa\left[b_{\kappa}+\frac{1}{c}\sum_{i=1}^m\sum_{j=1}^{|N(\bm{s}_i)|}w_{ij}\sum_{t=1}^{T-1}U^{t}_{N(\bm{s}_i)[j]}\right]\right\rbrace \mathbbm{1}_{\{0<\kappa<1-\rho\}}\nonumber \\
    			& \sim \tGa\left(a_{\kappa}+\sum_{t=1}^{T-1}\sum_{i=1}^m \sum_{j=1}^{|N(\bm{s}_i)|}Z^t_{ij}, \: b_{\kappa}+\frac{1}{c}\sum_{i=1}^m\sum_{j=1}^{|N(\bm{s}_i)|}w_{ij}\sum_{t=1}^{T-1}U^t_{N(\bm{s}_i)[j]},\:0,\:1-\rho\right). 
    		\end{align}
    		\item Sampling from the full conditional distributions of $Z^t_{ij}$'s when $T>1$:\par
    		Fix any arbitrary $t\in\{1,2,\ldots,T-1\}$, $i\in\{1,2,\ldots,m\}$. Then 
    		\begin{align}\label{postZi0Supp}
    			f\left(Z^t_{i0}|\cdot\right)&\propto f\left(U^{t+1}(\bm{s}_i)~|~Z^t_{i0},Z^t_{i1},\ldots,Z^t_{i|N(\bm{s}_i)|},c\right) \times  f\left(Z^t_{i0}~|~U^t(\bm{s}_i),\rho,c\right)\nonumber\\
    			&\propto\frac{\left(\frac{U^{t+1}(\bm{s}_i)}{c}\right)^{Z^t_{i0}}}{\Gamma\left(\alpha+\sum_{k=0}^{|N(\bm{s}_i)|}Z^t_{ik}\right)} \times \frac{\left[\frac{\rho}{c}\cdot U^t(\bm{s}_i)\right]^{Z^t_{i0}}}{Z^t_{i0}!}\nonumber\\
    			&\propto \frac{\left(\frac{\rho U^{t+1}(\bm{s}_i)U^{t}(\bm{s}_i)}{c^2}\right)^{Z^t_{i0}}}{\Gamma\left(\alpha+\sum_{k=0}^{|N(\bm{s}_i)|}Z^t_{ik}\right)\times Z^t_{i0}!}\nonumber\\ &\sim \Bessel\left(\alpha+\sum_{k=1}^{|N(\bm{s}_i)|}Z^t_{ik}-1, \frac{2}{c}\sqrt{\rho \cdot U^{t+1}(\bm{s}_i)U^t(\bm{s}_i)}\right). 
    		\end{align}
    		When $m>1$, we further fix any arbitrary $j\in\{1,2,\ldots,|N(\bm{s}_i)|\}$. Then 
    		\begin{align}\label{postZiSupp}
    			f\left(Z^t_{ij}|\cdot\right)&\propto f\left(U^{t+1}(\bm{s}_i)~|~Z^t_{i0},Z^t_{i1},\ldots,Z^t_{i|N(\bm{s}_i)|},c\right) \times  f\left(Z^t_{ij}~|~U^t_{N(\bm{s}_i)[j]},\kappa,c\right)\nonumber\\
    			&\propto\frac{\left(\frac{U^{t+1}(\bm{s}_i)}{c}\right)^{Z^t_{ij}}}{\Gamma\left(\alpha+\sum_{k=0}^{|N(\bm{s}_i)|}Z^t_{ik}\right)} \times \frac{\left(\frac{\kappa}{c}\cdot w_{ij}U^t_{N(\bm{s}_i)[j]}\right)^{Z^t_{ij}}}{Z^t_{ij}!}\nonumber\\
    			&\propto \frac{\left(\frac{\kappa U^{t+1}(\bm{s}_i)}{c^2}\cdot w_{ij}U^t_{N(\bm{s}_i)[j]}\right)^{Z^t_{ij}}}{\Gamma\left(\alpha+\sum_{k=0}^{|N(\bm{s}_i)|}Z^t_{ik}\right)\times Z^t_{ij}!} \nonumber\\
    			&\sim \Bessel\left(\alpha+\sum_{k=0,k\neq j}^{|N(\bm{s}_i)|}Z^t_{ik}-1,\: \frac{2}{c}\sqrt{\kappa w_{ij}U^{t+1}(\bm{s}_i)U^t_{N(\bm{s}_i)[j]}}\right), 
    		\end{align}
    		where $X\sim  \Bessel(\nu,a)$ has the probability mass function $\mathbb{P}(X=n\mid a,\nu)=\frac{(a/2)^{2n+\nu}}{I_{\nu}(a)\Gamma(n+\nu+1)n!}$ for $n\in \NN$, $a>0,\nu>-1$, and $I_{\nu}(\cdot)$ is the modified Bessel function of the first kind. We generate Bessel random variables via the rejection algorithms in Section 3 of \citet{Devroye2002}.
    		\item Sampling $\bm{\beta}_{p\times 1}$ via a Metropolis step with prior $\pi(\bm{\beta})\sim N_p\left(\bm{\mu}_{0\bm{\beta}}, {\Sigma}_{0\bm{\beta}}\right)$:
    		\begin{align}\label{postBetaSupp}
    			&f\left(\bm{\beta}|\cdot\right) \propto \pi(\bm{\beta}) \times \prod_{i=1}^m\prod_{t=1}^T f\left(y_t(\bm{s}_i)~|~U^t(\bm{s}_i),\bm{\beta}\right) \\
    			&\;\propto \exp\left\lbrace -\frac{1}{2}\left(\bm{\beta}-\bm{\mu}_{0\bm{\beta}}\right)^{\T}\Sigma_{0\bm{\beta}}^{-1}\left(\bm{\beta}-\bm{\mu}_{0\bm{\beta}}\right)\right\rbrace   \times\nonumber \\
    			&\quad\prod_{i=1}^m\prod_{t=1}^T \exp\left\lbrace - U^t(\bm{s}_i) e^{\bm{x}_t(\bm{s}_i)^{\T}\bm{\beta}}\right\rbrace \times \left[ U^t(\bm{s}_i) e^{\bm{x}_t(\bm{s}_i)^{\T}\bm{\beta}}\right]^{y_t(\bm{s}_i)} \nonumber \\
    			&\;\propto \exp\left\lbrace -\frac{1}{2}\left(\bm{\beta}-\bm{\mu}_{0\bm{\beta}}\right)^{\T}\Sigma_{0\bm{\beta}}^{-1}\left(\bm{\beta}-\bm{\mu}_{0\bm{\beta}}\right) + \sum_{i=1}^m\sum_{t=1}^T \left[y_t(\bm{s}_i)\cdot \bm{x}_t(\bm{s}_i)^{\T}\bm{\beta} -  U^t(\bm{s}_i) e^{\bm{x}_t(\bm{s}_i)^{\T}\bm{\beta}} \right]\right\rbrace  \nonumber 
    		\end{align}
    		is not a standard parametric family, and we resort to a Metropolis step for sampling $\bm{\beta}$ from its full conditional distribution.
    		\par We consider a symmetric Markov kernel $Q:\mathbb{R}^p \times 
    		\mathbb{R}^p\rightarrow \mathbb{R}^+$ with $
    		Q\left(\bm{\beta}_a, \bm{\beta}_b\right)=Q(\bm{\beta}_b,\bm{\beta}_a)$ for all $\bm{\beta}_a,\bm{\beta}_b\in\mathbb{R}^p$. $
    		Q\left(\bm{\beta}_a, \bm{\beta}_b\right)$ equals the density of a $N_{p}(\bm{\beta}_a, V)$ random variable evaluated at $\bm{\beta}_b$ with probability $p$ and the density of a $N_{p}(\bm{\beta}_a, 100V)$ random variable evaluated at $\bm{\beta}_b$ with probability $1-p$, 
    		where $V$ is an estimate 
    		of the inverse Hessian of $\ln\left[f(\bm{\beta}|\cdot)\right]$ and $p\in(0,1)$ should be quite large to give less probability to big moves in the parameter space. 
    		At each MCMC iteration $w\in\mathbb{N}$, we let the current parameter estimate for $\bm{\beta}$ be $\bm{\beta}^{(w)}$. Since 
    		\begin{align}\label{sampleBetaHessian}
    			\frac{\partial^2}{\partial \bm{\beta}\partial \bm{\beta}^{\T}}\ln f(\bm{\beta}|\cdot) = -\Sigma_{0\bm{\beta}}^{-1} -  \sum_{i=1}^m\sum_{t=1}^T U^t(\bm{s}_i) e^{\bm{x}_t(\bm{s}_i)^{\T}\bm{\beta}}\bm{x}_t(\bm{s}_i)\bm{x}_t(\bm{s}_i)^{\T},
    		\end{align}
    		we set $V^{(w)} = \left[\Sigma_{0\bm{\beta}}^{-1} +  \sum_{i=1}^m\sum_{t=1}^T U^t(\bm{s}_i) e^{\bm{x}_t(\bm{s}_i)^{\T}\bm{\beta}^{(w)}}\bm{x}_t(\bm{s}_i)\bm{x}_t(\bm{s}_i)^{\T}\right]^{-1}$. We propose a new parameter value $\bm{\beta}^*$ from the density $q(\bm{\beta}|\bm{\beta}^{(w)})=Q\left(\bm{\beta}^{(w)},\bm{\beta}\right)$ and then set 
    		\begin{equation*}
    			\bm{\beta}^{(w+1)}=
    			\begin{cases}
    				\bm{\beta}^*,&\text{ with probability }\alpha\left(\bm{\beta}^{(w)},\bm{\beta}^*\right)\\ 
    				\bm{\beta}^{(w)},&\text{ with probability }1-\alpha\left(\bm{\beta}^{(w)},\bm{\beta}^*\right)
    			\end{cases},
    		\end{equation*}
    		where $\alpha\left(\bm{\beta}^{(w)},\bm{\beta}^*\right)=\min\left\lbrace\frac{f(\bm{\beta}^*|\cdot)}{f(\bm{\beta}^{(w)}|\cdot)},1\right\rbrace$ is the acceptance probability for this multivariate random walk Metropolis algorithm. 
    	\end{enumerate}

    	\section{Spatiotemporal Frailty Process Stationarity Condition for the Third Simulation Group Mentioned at the Start of \Cref{sec:simulation}}\label{appenB}
    	\begin{theorem} \label{thm:stationary2}
    		Fix any arbitrary $m\in \NN$ and spatial locations $\{\bm{s}_1,\ldots,\bm{s}_m\}$. For each $t=2,\ldots,T$ and $\bm{s}_i$ with $i=2,\ldots,m$, we assume
    		\begin{align} \label{eq:ARGsupp}
    			& U^{t}(\bm{s}_i) ~|~ Z^{t-1}(\bm{s}_i) \overset{\ind}{\sim} \Ga\left(\alpha + Z^{t-1}(\bm{s}_i), \frac{1}{c} \right), \nonumber \\
    			& Z^{t-1}(\bm{s}_i) ~|~  \left\{U^{t-1}(\bm{s}_j):~ j=1,\ldots,m\right\} \overset{\ind}{\sim} \Pois \left(\frac{1}{c} \sum_{j=1}^m v_{ij} U^{t-1}(\bm{s}_j)\right) 
    		\end{align}
    		independent of $U^t(\bm{s}_1)\sim \Ga(\alpha,\frac{1}{c})$ for some hyperparameter $\alpha>1$, parameter $c>0$, and nonnegative weights $\bfV=\{v_{ij}\geq 0:i,j=1,\ldots,m\}$ such that $\sum_{j=1}^m v_{ij}>0$ for all $i=2,\ldots,m$ and and $v_{1j}=0$ for all $j=1,\ldots,m$. 
    		
    		Assume $\bfV$ satisfies at least one of the two following conditions.
    		\begin{itemize}
    			\item \textbf{Condition 1:} $\sum_{i=1}^m v_{ij}\leq 1$ for all $j=1,\ldots,m$;  
    			\item \textbf{Condition 2:} $\sum_{j=1}^m v_{ij}\leq 1$ for all $i=1,\ldots,m$. 
    		\end{itemize}
    		Then the process $\{U^t(\bm{s}_{1:m}):t=1,2,\ldots,T\}$ defined above in \eqref{eq:ARGsupp} is a stationary process with an invariant distribution as $T\to\infty$.
    	\end{theorem}
    	
    	\begin{proof}[Proof of \Cref{thm:stationary2}]
    		For any vector $\bm{x}=(x_1,\ldots,x_m)^\T \in \RR^{m}$ such that $x_i\geq 0$ for all $i=1,\ldots,m$, we have 
    		\begin{align} \label{eq:CAR1.1supp}
    			&~\quad \mathbb{E}\left[\exp\left\{-\bm{x}^\T U^t(\bm{s}_{1:m})\right\}  ~|~ U^{t-1}(\bm{s}_{1:m}) \right] \nonumber \\
    			& = \mathbb{E}\left(\mathbb{E}\left[\exp\left\{-\bm{x}^\T U^t(\bm{s}_{1:m})\right\} ~|~ Z^{t-1}(\bm{s}_{2:m})\right] ~|~ U^{t-1}(\bm{s}_{1:m})  \right) \nonumber \\
    			&= \mathbb{E}\left( \mathbb{E}\left[\exp\left\{-x_1 U^t(\bm{s}_1)\right\}\right]\times\prod_{i=2}^m \mathbb{E}\left[\exp\left\{-x_i U^t(\bm{s}_i)\right\} ~|~ Z^{t-1}(\bm{s}_i)\right] ~\Big|~ U^{t-1}(\bm{s}_{1:m})  \right) \nonumber \\
    			&= \mathbb{E}\left( \frac{1}{(cx_1+1)^{\alpha}}\times\prod_{i=2}^m \frac{1}{(cx_i+1)^{\alpha+Z^{t-1}(\bm{s}_i)}} ~\Big|~ U^{t-1}(\bm{s}_{1:m})  \right) \nonumber \\
    			&= \frac{1}{(cx_1+1)^{\alpha}}\times\prod_{i=2}^m\mathbb{E}\left(\frac{1}{(cx_i+1)^{\alpha+Z^{t-1}(\bm{s}_i)}} ~\Big|~ U^{t-1}(\bm{s}_{1:m})  \right) \nonumber \\
    			&= \frac{1}{(cx_1+1)^{\alpha}}\times\prod_{i=2}^m \sum_{k=0}^{\infty} \frac{\left(\frac{1}{c} \sum_{j=1}^m v_{ij} U^{t-1}(\bm{s}_j)\right)^k}{k!} \exp\left\{-\frac{1}{c} \sum_{j=1}^m v_{ij} U^{t-1}(\bm{s}_j)\right\} \cdot \frac{1}{(cx_i+1)^{\alpha+k}}   \nonumber \\
    			&= \prod_{i=1}^m \frac{1}{(cx_i+1)^{\alpha}}\times \exp\nonumber \left\{-\frac{1}{c} \sum_{j=1}^m v_{ij} U^{t-1}(\bm{s}_j)\left[1 - \frac{1}{cx_i+1}\right]\right\}\\
    			&= \exp\left\{-\sum_{j=1}^m \left(\sum_{i=1}^m \frac{x_iv_{ij}}{cx_i+1} \right) U^{t-1}(\bm{s}_j) -\alpha\sum_{i=1}^m \log (cx_i+1)\right\}.
    		\end{align}
    		If we define vector-valued function $a(\bm{x}) = \left(\sum_{i=1}^m \frac{x_iv_{i1}}{cx_i+1},\ldots,\sum_{i=1}^m \frac{x_iv_{im}}{cx_i+1}\right)^\T \in \RR^m$ and scalar-valued function $b(\bm{x}) = -\alpha\sum_{i=1}^m \log (cx_i+1) \in \RR$, then \eqref{eq:CAR1.1} implies 
    		\begin{align*}
    			\mathbb{E}\left[\exp\left\{-\bm{x}^\T U^t(\bm{s}_{1:m})\right\}  ~|~ U^{t-1}(\bm{s}_{1:m}) \right] &= \exp\left\{- a(\bm{x})^\T U^{t-1}(\bm{s}_{1:m}) + b(\bm{x})\right\}.
    		\end{align*}
    		Since $v_{ij}\geq 0$ for all $(i,j)$ and $\sum_{j=1}^m v_{ij}>0$ for all $i=2,\ldots,m$, $a(\bm{x})\neq \bm{0}_{m\times 1}$ for all $\bm{x}\in\mathbb{R}^m$ with $x_i\geq 0$ for all $i$ and $\sum_{i=2}^m x_i>0$. Hence, based on Definition 1 in \citet{Daretal06}, the vector process $\{U^t(\bm{s}_{1:m}):t=1,2,\ldots\}$ is a temporal compound autoregressive process of order 1 for a fixed $m$ and given spatial locations $\{\bm{s}_1,\ldots,\bm{s}_m\}$. 	
    		By Proposition 6 in \citet{Daretal06}, the process $\{U^t(\bm{s}_{1:m}):t=1,2,\ldots\}$ converges to a stationary distribution as $T\to\infty$ if and only if $\lim_{h\to\infty} a^{\circ h}(\bm{x}) =\bm{0}_{m\times 1}$ for all $\bm{x}\in\RR^m$ with $x_i\geq 0$ for all $i=1,\ldots,m$ and $\sum_{i=2}^m x_i>0$, where $a^{\circ h}(\bm{x})$ represents the function $a(\bm{x})$ compounded $h$ times with itself. The rest of the proof is exactly the same as that in \Cref{thm:stationary1}.
    	\end{proof}
    	
    	\clearpage
    	\section{Simulation Results Complementary to 
    		\Cref{sec:simulation}}
    	\begin{figure}[h!]
    		\centering
    		\includegraphics[width=0.9\textwidth]{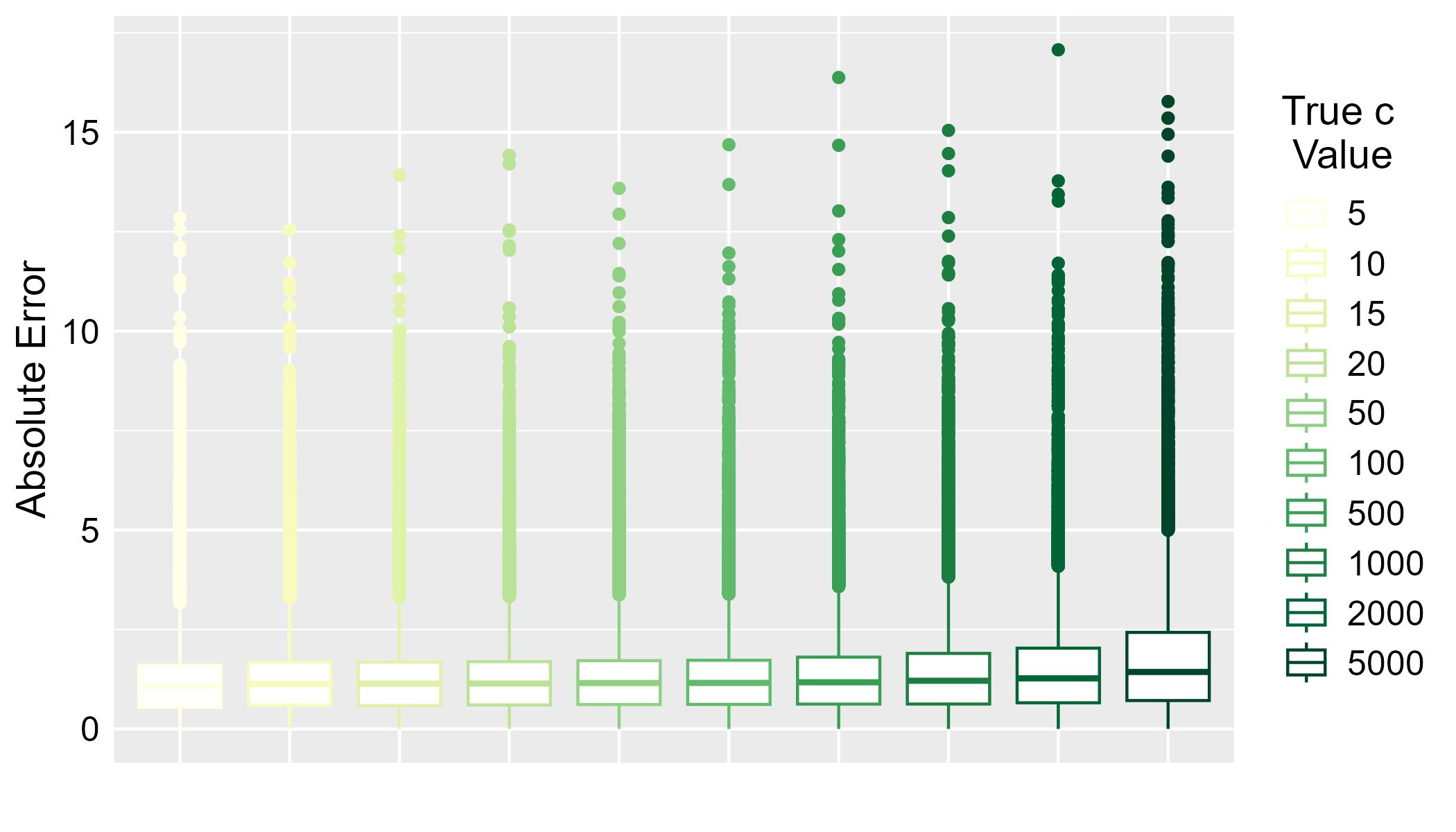}     
    		\caption{Boxplot: distribution of the $mT$ absolute errors $\big|y_t(\bm{s}_i)-\hat{y}_t(\bm{s}_i)\big|$, where $\hat{y}_t(\bm{s}_i)$ denotes the posterior mean in time interval $t$ and spatial location $\bm{s}_i$ for each $(t,i)$, for our model fit with \texttt{hypara1} to data simulated from the first simulation group (mentioned at the start of \Cref{sec:simulation}) and each of the ten true $c$ parameter values.}
    		\label{plot:absErrorBoxplot_hyparaDefault} 
    	\end{figure}
    	\begin{figure}[h!]
    		\centering
    		\includegraphics[width=0.9\textwidth]{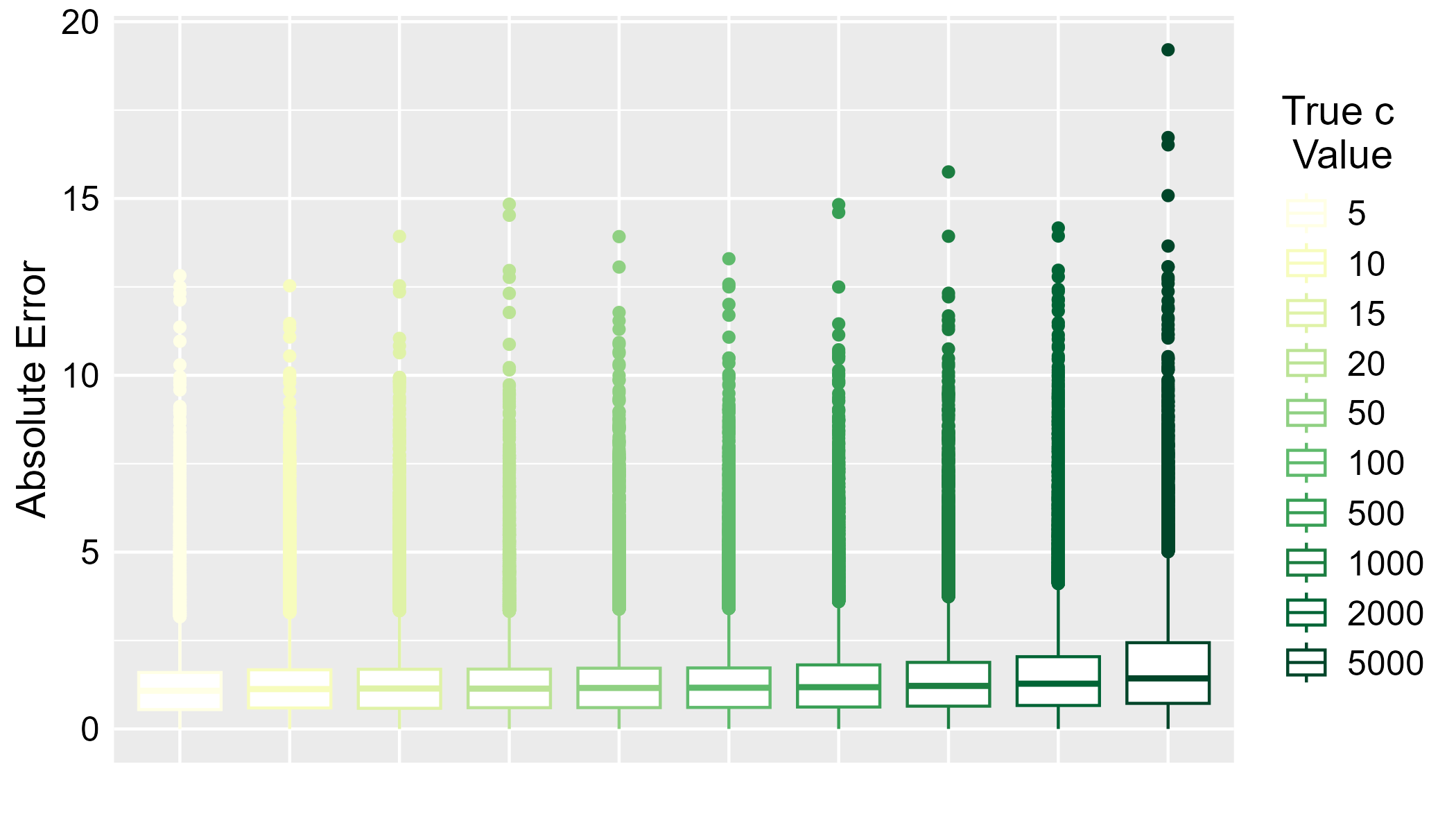}     
    		\caption{Boxplot: distribution of the $mT$ absolute errors $\big|y_t(\bm{s}_i)-\hat{y}_t(\bm{s}_i)\big|$, where $\hat{y}_t(\bm{s}_i)$ denotes the posterior mean in time interval $t$ and spatial location $\bm{s}_i$ for each $(t,i)$, for our model fit with \texttt{hypara2} to data simulated from the first simulation group (mentioned at the start of \Cref{sec:simulation}) and each of the ten true $c$ parameter values.}
    		\label{plot:absErrorBoxplot_hyparaMiddle} 
    	\end{figure}
    	
    	\begin{table}[H]
    		{
    			\begin{adjustwidth}{-0cm}{-0cm}
    				\begin{center}
    					\scalebox{0.85}{\begin{tabular}{|*{9}{c|}}
    							\hline
    							\multirow{3}{*}{\textbf{True $c$}} & \multicolumn{4}{|c|}{\textbf{True $\kappa=0.35$}} & \multicolumn{4}{|c|}{\textbf{True $\kappa=0.7$}}\\
    							\cline{2-9}
    							& \multicolumn{2}{|c|}{{$2^{\text{nd}}$ simulation group}} & \multicolumn{2}{|c|}{{$3^{\text{rd}}$ simulation group}} & \multicolumn{2}{|c|}{{$2^{\text{nd}}$ simulation group}} &
    							\multicolumn{2}{|c|}{{$3^{\text{rd}}$ simulation group}}\\
    							\cline{2-9}
    							\textbf{Value}  & \texttt{hypara3} & \texttt{hypara4} & \texttt{hypara3} & \texttt{hypara4} & \texttt{hypara3} & \texttt{hypara4} & \texttt{hypara3} & \texttt{hypara4}\\
    							\hline
    							$c=5$ & 0.739245 & 0.738330 & 0.723956 & 0.721725 & 1.009475 & 1.007144 & 0.942227 & 0.939763\\
    							\hline
    							$c=10$ & 0.766734 & 0.76466 & 0.750299 & 0.748261 & 1.041290 & 1.040639 & 0.976482 & 0.974455\\
    							\hline
    							$c=15$ & 0.775026 & 0.773128 & 0.757853 & 0.756802 & 1.049664 & 1.049141 & 0.979377 & 0.978869\\
    							\hline
    							$c=20$ & 0.779227 & 0.778088 & 0.762157 & 0.760419 & 1.052398 & 1.052027 & 0.987071 & 0.985950\\
    							\hline
    							$c=50$ & 0.789033 & 0.788211 & 0.774082 & 0.773347 & 1.067943 & 1.06694 & 1.00299 & 1.00351\\
    							\hline
    							$c=100$ & 0.794015 & 0.794397 & 0.779964 & 0.781459 & 1.074548 & 1.079313 & 1.006237 & 1.007725\\
    							\hline
    							$c=500$ & 0.818645 & 0.818356 & 0.805213 & 0.800317 & 1.114351 & 1.11053 & 1.039833 & 1.029073\\
    							\hline
    							$c=1000$ & 0.847104 & 0.847053 & 0.825729 & 0.831444 & 1.160870 & 1.15175 & 1.08213 & 1.076918\\
    							\hline
    							$c=2000$ & 0.894481 & 0.891117 & 0.874451 & 0.880667 & 1.228211 & 1.23757 & 1.139403 & 1.149787\\
    							\hline
    							$c=5000$ & 1.018098 & 1.016191 & 1.004142 & 1.00886 & 1.426702 & 1.450810 & 1.324379 & 1.326022\\
    							\hline
    					\end{tabular}}
    				\end{center}
    			\end{adjustwidth}
    			\caption{Mean Absolute Error (MAE) for each of the $2\times 2\times 2\times 10=80$ simulation settings for the second and the third simulation groups.}
    			\label{MAE}
    		}
    	\end{table}    
    	
    	\begin{table}[H]
    		{
    			\begin{adjustwidth}{-0cm}{-0cm}
    				\begin{center}
    					\scalebox{0.84}{\begin{tabular}{|*{9}{c|}}
    							\hline
    							\multirow{3}{*}{\textbf{True $c$}} & \multicolumn{4}{|c|}{\textbf{True $\kappa=0.35$}} & \multicolumn{4}{|c|}{\textbf{True $\kappa=0.7$}}\\
    							\cline{2-9}
    							& \multicolumn{2}{|c|}{{$2^{\text{nd}}$ simulation group}} & \multicolumn{2}{|c|}{{$3^{\text{rd}}$ simulation group}} & \multicolumn{2}{|c|}{{$2^{\text{nd}}$ simulation group}} &
    							\multicolumn{2}{|c|}{{$3^{\text{rd}}$ simulation group}}\\
    							\cline{2-9}
    							\textbf{Value}  & \texttt{hypara3} & \texttt{hypara4} & \texttt{hypara3} & \texttt{hypara4} & \texttt{hypara3} & \texttt{hypara4} & \texttt{hypara3} & \texttt{hypara4}\\
    							\hline
    							$c=5$ & 0.3722 & 0.3696 & 0.3399 & 0.3343 & 0.7010 & 0.7003 & 0.7175 & 0.7164\\
    							\hline
    							$c=10$ & 0.3711 & 0.3680 & 0.3408 & 0.3379 & 0.6990 & 0.6989 & 0.7175 & 0.7171\\
    							\hline
    							$c=15$ & 0.3699 & 0.3684 & 0.3371 & 0.3357 & 0.7010 & 0.7012 & 0.7103 & 0.7099 \\
    							\hline
    							$c=20$ & 0.3677 & 0.3674 & 0.3431 & 0.3389 & 0.6995 & 0.6991 & 0.7149 & 0.7147\\
    							\hline
    							$c=50$ & 0.3690 & 0.3662 & 0.3416 & 0.3407 & 0.7000 & 0.6998 & 0.7147 & 0.7152\\
    							\hline
    							$c=100$ & 0.3696 & 0.3675 & 0.3468 & 0.3468 & 0.6988 & 0.6986 & 0.7142 & 0.7137\\
    							\hline
    							$c=500$ & 0.3679 & 0.3685 & 0.3426 & 0.3425 & 0.6992 & 0.6992 & 0.7141 & 0.7138\\
    							\hline
    							$c=1000$ & 0.3685 & 0.3688 & 0.3417 & 0.3404 & 0.6990 & 0.6993 & 0.7137 & 0.7136\\
    							\hline
    							$c=2000$ & 0.3672 & 0.3672 & 0.3399 & 0.3412 & 0.6993 & 0.6993 & 0.7141 & 0.7139\\
    							\hline
    							$c=5000$ & 0.3674 & 0.3682 & 0.3410 & 0.3406 & 0.6996 & 0.6993 & 0.7143 & 0.7138\\
    							\hline
    					\end{tabular}}
    				\end{center}
    			\end{adjustwidth}
    			\caption{Mean value of the kept $10^4$ posterior samples of the parameter $\kappa$ for each of the $2\times 2\times 2\times 10=80$ simulation settings for the second and the third simulation groups.}
    			\label{postMeanKappa}
    		}
    	\end{table}
    	
    	\begin{table}[H]
    		{
    			\begin{adjustwidth}{-0cm}{-0cm}
    				\begin{center}
    					\scalebox{0.84}{\begin{tabular}{|*{9}{c|}}
    							\hline
    							\multirow{3}{*}{\textbf{True $c$}} & \multicolumn{4}{|c|}{\textbf{True $\kappa=0.35$}} & \multicolumn{4}{|c|}{\textbf{True $\kappa=0.7$}}\\
    							\cline{2-9}
    							& \multicolumn{2}{|c|}{{$2^{\text{nd}}$ simulation group}} & \multicolumn{2}{|c|}{{$3^{\text{rd}}$ simulation group}} & \multicolumn{2}{|c|}{{$2^{\text{nd}}$ simulation group}} &
    							\multicolumn{2}{|c|}{{$3^{\text{rd}}$ simulation group}}\\
    							\cline{2-9}
    							\textbf{Value}  & \texttt{hypara3} & \texttt{hypara4} & \texttt{hypara3} & \texttt{hypara4} & \texttt{hypara3} & \texttt{hypara4} & \texttt{hypara3} & \texttt{hypara4}\\
    							\hline
    							$c=5$ & 4.941 & 4.964 & 5.217 & 5.257 & 5.067 & 5.081 & 4.808 & 4.828 \\
    							\hline
    							$c=10$ & 9.849 & 9.901 & 10.414 & 10.459 & 10.164 & 10.170 & 9.725 & 9.747 \\
    							\hline
    							$c=15$ & 14.83 & 14.87 & 15.75 & 15.78 & 15.21 & 15.20 & 14.74 & 14.75 \\
    							\hline
    							$c=20$ & 19.86 & 19.87 & 20.85 & 20.97 & 20.35 & 20.38 & 19.49 & 19.50\\
    							\hline
    							$c=50$ & 49.58 & 49.79 & 52.20 & 52.25 & 50.80 & 50.87 & 48.72 & 48.64\\
    							\hline
    							$c=100$ & 99.07 & 99.38 & 103.31 & 103.26 & 102.02 & 102.1 & 97.48 & 97.54\\
    							\hline
    							$c=500$ & 496.2 & 495.6 & 520.4 & 520.6 & 509.9 & 509.8 & 488.1 & 488.6\\
    							\hline
    							$c=1000$ & 991.2 & 990.9 & 1042.3 & 1043.9 & 1019.9 & 1019.4 & 978.1 & 978.0\\
    							\hline
    							$c=2000$ & 1987 & 1987 & 2089 & 2086 & 2037 & 2038 & 1954 & 1955\\
    							\hline
    							$c=5000$ & 4967 & 4959  & 5215 & 5218 & 5087 & 5093 & 4884 & 4889\\
    							\hline
    					\end{tabular}}
    				\end{center}
    			\end{adjustwidth}
    			\caption{Mean value of the kept $10^4$ posterior samples of the parameter $c$ for each of the $2\times 2\times 2\times 10=80$ simulation settings for the second and the third simulation groups.}
    			\label{postMeanC}
    		}
    	\end{table}
    	\clearpage   
    	
    	\begin{table}[H]
    		{
    			\begin{adjustwidth}{-0cm}{-0cm}
    				\begin{center}
    					\scalebox{0.95}{\begin{tabular}{|*{8}{c|}}
    							\hline
    							\multicolumn{2}{|c|}{} & Minimum & 1st Quantile & Median & Mean & 3rd Quantile & Maximum \\
    							\hline
    							\multicolumn{2}{|c|}{Actual $y_t(\bm{s}_i)$'s} & 0.00 &  11.00 &  21.00 &  24.72  & 34.00 & 127.00 \\
    							\hline
    							\multirow{2}{*}{$U_t(\bm{s}_i)$'s} & \texttt{hypara1} &  0.00 &  11.92 &  21.36 &  24.72  & 33.93 & 156.96  \\
    							\cline{2-8}
    							& \texttt{hypara2} &  0.00  & 11.92  & 21.36 &  24.72 &  33.93 & 160.24\\
    							\hline
    							\multirow{2}{*}{$c$} & \texttt{hypara1} & 4.743 &  4.933  & 4.986 &  4.987 &  5.041  & 5.249  \\
    							\cline{2-8}
    							& \texttt{hypara2} & 4.654 &  4.945 &  4.997 &  4.997 &  5.049 &  5.409\\
    							\hline
    							\multirow{2}{*}{$\kappa$} & \texttt{hypara1} & 0.3673 & 0.3958 & 0.4022 & 0.4023 & 0.4090 & 0.4411    \\
    							\cline{2-8}
    							& \texttt{hypara2} &  0.3682 & 0.3966 & 0.4030 & 0.4030 & 0.4092 & 0.4383\\
    							\hline
    							\multirow{2}{*}{$\rho$} & \texttt{hypara1} & 0.3662 & 0.3975 & 0.4041 & 0.4040 & 0.4104 & 0.4406     \\
    							\cline{2-8}
    							& \texttt{hypara2} & 0.3671 & 0.3970 & 0.4030 & 0.4031 & 0.4094 & 0.4387\\
    							\hline
    					\end{tabular}}
    				\end{center}
    			\end{adjustwidth}
    			\caption{Summary statistics of the actual observed counts $y_{1:T}(\bm{s}_{1:m})$ and the kept $10^4$ posterior samples of the parameters $c,\kappa, \rho,U_{1:T}(\bm{s}_{1:m})$ for hyperparameter specifications \texttt{hypara1} and \texttt{hypara2} with true parameter values $(\kappa,\rho,c) = (0.4, 0.4, 5)$ and undirected spatial locations from the first simulation group.}
    			\label{c5complete}
    		}
    	\end{table}
    	
    	\begin{table}[h]
    		{
    			\begin{adjustwidth}{-0cm}{-0cm}
    				\begin{center}
    					\scalebox{0.95}{\begin{tabular}{|*{8}{c|}}
    							\hline
    							\multicolumn{2}{|c|}{} & Minimum & 1st Quantile & Median & Mean & 3rd Quantile & Maximum \\
    							\hline
    							\multicolumn{2}{|c|}{Actual $y_t(\bm{s}_i)$'s} & 0.00  & 23.00  & 42.00 &  49.46  & 68.00 & 258.00 \\
    							\hline
    							\multirow{2}{*}{$U_t(\bm{s}_i)$'s} & \texttt{hypara1} &  0.00  & 23.88 &  42.73 &  49.46 &  67.79 & 315.09  \\
    							\cline{2-8}
    							& \texttt{hypara2} & 0.00  & 23.88  & 42.73 &  49.46 &  67.79 & 317.33 \\
    							\hline
    							\multirow{2}{*}{$c$} & \texttt{hypara1} & 9.405 &  9.834 &  9.934 &  9.936 & 10.034 & 10.656 \\
    							\cline{2-8}
    							& \texttt{hypara2} & 9.413 &  9.853 &  9.946 &  9.949 & 10.043 & 10.480\\
    							\hline
    							\multirow{2}{*}{$\kappa$} & \texttt{hypara1} & 0.3660 & 0.3970 & 0.4031 & 0.4031 & 0.4092 & 0.4396    \\
    							\cline{2-8}
    							& \texttt{hypara2} & 0.3630 & 0.3963 & 0.4025 & 0.4025 & 0.4086 & 0.4393 \\
    							\hline
    							\multirow{2}{*}{$\rho$} & \texttt{hypara1} & 0.3726 & 0.3975 & 0.4034 & 0.4035 & 0.4096 & 0.4376      \\
    							\cline{2-8}
    							& \texttt{hypara2} & 0.3636 & 0.3978 & 0.4042 & 0.4040 & 0.4102 & 0.4371 \\
    							\hline
    					\end{tabular}}
    				\end{center}
    			\end{adjustwidth}
    			\caption{Summary statistics of the actual observed counts $y_{1:T}(\bm{s}_{1:m})$ and the kept $10^4$ posterior samples of the parameters $c,\kappa, \rho,U_{1:T}(\bm{s}_{1:m})$ for hyperparameter specifications \texttt{hypara1} and \texttt{hypara2} with true parameter values $(\kappa,\rho,c) = (0.4, 0.4, 10)$ and undirected spatial locations from the first simulation group.}
    			\label{c10complete}
    		}
    	\end{table}
    	
    	\begin{table}[h]
    		{
    			\begin{adjustwidth}{-0cm}{-0cm}
    				\begin{center}
    					\scalebox{0.95}{\begin{tabular}{|*{8}{c|}}
    							\hline
    							\multicolumn{2}{|c|}{} & Minimum & 1st Quantile & Median & Mean & 3rd Quantile & Maximum \\
    							\hline
    							\multicolumn{2}{|c|}{Actual $y_t(\bm{s}_i)$'s} &  0.00  & 35.00 &  64.00  & 74.13 & 102.00 & 382.00 \\
    							\hline
    							\multirow{2}{*}{$U_t(\bm{s}_i)$'s} & \texttt{hypara1} & 0.00  & 35.87 &  64.03 &  74.13 & 101.36 & 459.63  \\
    							\cline{2-8}
    							& \texttt{hypara2} &  0.00  & 35.87  & 64.03  & 74.13 & 101.36 & 443.33 \\
    							\hline
    							\multirow{2}{*}{$c$} & \texttt{hypara1} &  14.19 &  14.75  & 14.91 &  14.91  & 15.06 &  15.69   \\
    							\cline{2-8}
    							& \texttt{hypara2} &  14.07  & 14.77 &  14.93 &  14.92 &  15.08  & 15.67 \\
    							\hline
    							\multirow{2}{*}{$\kappa$} & \texttt{hypara1} & 0.3722 & 0.4027 & 0.4088 & 0.4088 & 0.4151 & 0.4412   \\
    							\cline{2-8}
    							& \texttt{hypara2} & 0.3715 & 0.4010 & 0.4076 & 0.4074 & 0.4138 & 0.4418\\
    							\hline
    							\multirow{2}{*}{$\rho$} & \texttt{hypara1} & 0.3604 & 0.3919 & 0.3976 & 0.3979 & 0.4039 & 0.4348    \\
    							\cline{2-8}
    							& \texttt{hypara2} & 0.3645 & 0.3928 & 0.3990 & 0.3991 & 0.4054 & 0.4287 \\
    							\hline
    					\end{tabular}}
    				\end{center}
    			\end{adjustwidth}
    			\caption{Summary statistics of the actual observed counts $y_{1:T}(\bm{s}_{1:m})$ and the kept $10^4$ posterior samples of the parameters $c,\kappa, \rho,U_{1:T}(\bm{s}_{1:m})$ for hyperparameter specifications \texttt{hypara1} and \texttt{hypara2} with true parameter values $(\kappa,\rho,c) = (0.4, 0.4, 15)$ and undirected spatial locations from the first simulation group.}
    			\label{c15complete}
    		}
    	\end{table}
    	
    	\begin{table}[h]
    		{
    			\begin{adjustwidth}{-0cm}{-0cm}
    				\begin{center}
    					\scalebox{0.95}{\begin{tabular}{|*{8}{c|}}
    							\hline
    							\multicolumn{2}{|c|}{} & Minimum & 1st Quantile & Median & Mean & 3rd Quantile & Maximum \\
    							\hline
    							\multicolumn{2}{|c|}{Actual $y_t(\bm{s}_i)$'s} & 0.00  & 47.00 &  85.00 &  98.84 & 135.00 & 517.00\\
    							\hline
    							\multirow{2}{*}{$U_t(\bm{s}_i)$'s} & \texttt{hypara1} &  0.00 &  47.56 &  85.25 &  98.84 & 135.15 & 596.63 \\
    							\cline{2-8}
    							& \texttt{hypara2} & 0.00 &  47.55 &  85.25  & 98.84  & 135.15  & 604.63 \\
    							\hline
    							\multirow{2}{*}{$c$} & \texttt{hypara1} & 18.91 &  19.70 &  19.88 &  19.89 &  20.08 &  21.29   \\
    							\cline{2-8}
    							& \texttt{hypara2} &  18.90  & 19.68 &  19.87 &  19.87 & 20.06 &  20.95 \\
    							\hline
    							\multirow{2}{*}{$\kappa$} & \texttt{hypara1} & 0.3728 & 0.3998 & 0.4068 & 0.4064 & 0.4132 & 0.4402   \\
    							\cline{2-8}
    							& \texttt{hypara2} & 0.3748 & 0.4008 & 0.4070 & 0.4072 & 0.4134 & 0.4432 \\
    							\hline
    							\multirow{2}{*}{$\rho$} & \texttt{hypara1} & 0.3682 & 0.3937 & 0.3998 & 0.4001 & 0.4063 & 0.4357 \\
    							\cline{2-8}
    							& \texttt{hypara2} & 0.3642 & 0.3933 & 0.3995 & 0.3995 & 0.4058 & 0.4320  \\
    							\hline
    					\end{tabular}}
    				\end{center}
    			\end{adjustwidth}
    			\caption{Summary statistics of the actual observed counts $y_{1:T}(\bm{s}_{1:m})$ and the kept $10^4$ posterior samples of the parameters $c,\kappa, \rho,U_{1:T}(\bm{s}_{1:m})$ for hyperparameter specifications \texttt{hypara1} and \texttt{hypara2} with true parameter values $(\kappa,\rho,c) = (0.4, 0.4, 20)$ and undirected spatial locations from the first simulation group.}
    			\label{c20complete}
    		}
    	\end{table}
    	
    	\begin{table}[h]
    		{
    			\begin{adjustwidth}{-0cm}{-0cm}
    				\begin{center}
    					\scalebox{0.95}{\begin{tabular}{|*{8}{c|}}
    							\hline
    							\multicolumn{2}{|c|}{} & Minimum & 1st Quantile & Median & Mean & 3rd Quantile & Maximum \\
    							\hline
    							\multicolumn{2}{|c|}{Actual $y_t(\bm{s}_i)$'s} &  0.0 &  118.0  & 213.0  & 247.1  & 338.0 & 1282.0  \\
    							\hline
    							\multirow{2}{*}{$U_t(\bm{s}_i)$'s} & \texttt{hypara1} & 0.0  & 118.6 &  213.2 &  247.1 &  338.8 & 1405.3   \\
    							\cline{2-8}
    							& \texttt{hypara2} & 0.0 &  118.6 &  213.2 &  247.1 &  338.8 & 1418.8\\
    							\hline
    							\multirow{2}{*}{$c$} & \texttt{hypara1} & 47.18 &  49.37  & 49.84  & 49.86 &  50.33 &  52.88     \\
    							\cline{2-8}
    							& \texttt{hypara2} & 47.24 &  49.35 &  49.84 &  49.84 &  50.33  & 52.37  \\
    							\hline
    							\multirow{2}{*}{$\kappa$} & \texttt{hypara1}  & 0.3713 & 0.3999 & 0.4060 & 0.4061 & 0.4121 & 0.4390\\
    							\cline{2-8}
    							& \texttt{hypara2} & 0.3765 & 0.4007 & 0.4068 & 0.4069 & 0.4132 & 0.4390  \\
    							\hline
    							\multirow{2}{*}{$\rho$} & \texttt{hypara1} & 0.3643 & 0.3938 & 0.3999 & 0.3998 & 0.4060 & 0.4304  \\
    							\cline{2-8}
    							& \texttt{hypara2} &  0.3679 & 0.3931 & 0.3992 & 0.3991 & 0.4051 & 0.4313  \\
    							\hline
    					\end{tabular}}
    				\end{center}
    			\end{adjustwidth}
    			\caption{Summary statistics of the actual observed counts $y_{1:T}(\bm{s}_{1:m})$ and the kept $10^4$ posterior samples of the parameters $c,\kappa, \rho,U_{1:T}(\bm{s}_{1:m})$ for the two hyperparameter specifications (\texttt{hypara1} and \texttt{hypara2}) when the true parameter values are $(\kappa,\rho,c) = (0.4, 0.4, 50)$ and the spatial locations are undirected from the first simulation group.}
    			\label{c50complete}
    		}
    	\end{table}
    	
    	\begin{table}[h]
    		{
    			\begin{adjustwidth}{-0cm}{-0cm}
    				\begin{center}
    					\scalebox{0.95}{\begin{tabular}{|*{8}{c|}}
    							\hline
    							\multicolumn{2}{|c|}{} & Minimum & 1st Quantile & Median & Mean & 3rd Quantile & Maximum \\
    							\hline
    							\multicolumn{2}{|c|}{Actual $y_t(\bm{s}_i)$'s} &  0.0  & 237.0  & 426.0  & 494.2 &  674.0 & 2454.0 \\
    							\hline
    							\multirow{2}{*}{$U_t(\bm{s}_i)$'s} & \texttt{hypara1} & 0.0 &  237.9 &  426.5 &  494.2  & 675.2 & 2631.5  \\
    							\cline{2-8}
    							& \texttt{hypara2} &  0.0  & 237.9  & 426.5  & 494.2 &  675.2 & 2637.0 \\
    							\hline
    							\multirow{2}{*}{$c$} & \texttt{hypara1} &  94.82 &  98.58  & 99.56 &  99.57 & 100.55 & 105.36   \\
    							\cline{2-8}
    							& \texttt{hypara2} & 94.58 &  98.55  & 99.48  & 99.53 & 100.50 & 104.58 \\
    							\hline
    							\multirow{2}{*}{$\kappa$} & \texttt{hypara1} &  0.3696 & 0.3999 & 0.4063 & 0.4061 & 0.4125 & 0.4369   \\
    							\cline{2-8}
    							& \texttt{hypara2} & 0.3733 & 0.4008 & 0.4067 & 0.4068 & 0.4128 & 0.4387\\
    							\hline
    							\multirow{2}{*}{$\rho$} & \texttt{hypara1} & 0.3685 & 0.3941 & 0.4000 & 0.4002 & 0.4059 & 0.4332    \\
    							\cline{2-8}
    							& \texttt{hypara2} & 0.3653 & 0.3937 & 0.3995 & 0.3995 & 0.4052 & 0.4318 \\
    							\hline
    					\end{tabular}}
    				\end{center}
    			\end{adjustwidth}
    			\caption{Summary statistics of the actual observed counts $y_{1:T}(\bm{s}_{1:m})$ and the kept $10^4$ posterior samples of the parameters $c,\kappa, \rho,U_{1:T}(\bm{s}_{1:m})$ for hyperparameter specifications \texttt{hypara1} and \texttt{hypara2} with true parameter values $(\kappa,\rho,c) = (0.4, 0.4, 100)$ and undirected spatial locations from the first simulation group.}
    			\label{c100complete}
    		}
    	\end{table}
    	
    	\begin{table}[h]
    		{
    			\begin{adjustwidth}{-0cm}{-0cm}
    				\begin{center}
    					\scalebox{0.95}{\begin{tabular}{|*{8}{c|}}
    							\hline
    							\multicolumn{2}{|c|}{} & Minimum & 1st Quantile & Median & Mean & 3rd Quantile & Maximum \\
    							\hline
    							\multicolumn{2}{|c|}{Actual $y_t(\bm{s}_i)$'s} & 0  &  1188 &   2130  &  2471  &  3378  & 12125  \\
    							\hline
    							\multirow{2}{*}{$U_t(\bm{s}_i)$'s} & \texttt{hypara1} & 0  &  1189  &  2129 &   2471 &   3378  & 12547 \\
    							\cline{2-8}
    							& \texttt{hypara2} & 0  &  1189  &  2129  &  2471  &  3378 &  12556  \\
    							\hline
    							\multirow{2}{*}{$c$} & \texttt{hypara1} & 470.3 &  492.3 &  497.3 &  497.5 &  502.5  & 525.9  \\
    							\cline{2-8}
    							& \texttt{hypara2} &  466.7  & 492.8 &  497.7 &  497.7 &  502.6  & 531.9 \\
    							\hline
    							\multirow{2}{*}{$\kappa$} & \texttt{hypara1} & 0.3697 & 0.4005 & 0.4066 & 0.4065 & 0.4127 & 0.4417\\
    							\cline{2-8}
    							& \texttt{hypara2} & 0.3710 & 0.4001 & 0.4062 & 0.4062 & 0.4123 & 0.4374  \\
    							\hline
    							\multirow{2}{*}{$\rho$} & \texttt{hypara1} & 0.3658 & 0.3940 & 0.3999 & 0.3999 & 0.4058 & 0.4371   \\
    							\cline{2-8}
    							& \texttt{hypara2} & 0.3715 & 0.3942 & 0.4000 & 0.4001 & 0.4058 & 0.4338 \\
    							\hline
    					\end{tabular}}
    				\end{center}
    			\end{adjustwidth}
    			\caption{Summary statistics of the actual observed counts $y_{1:T}(\bm{s}_{1:m})$ and the kept $10^4$ posterior samples of the parameters $c,\kappa, \rho,U_{1:T}(\bm{s}_{1:m})$ for hyperparameter specifications \texttt{hypara1} and \texttt{hypara2} with true parameter values $(\kappa,\rho,c) = (0.4, 0.4, 500)$ and undirected spatial locations from the first simulation group.}
    			\label{c500complete}
    		}
    	\end{table}
    	
    	\begin{table}[h]
    		{
    			\begin{adjustwidth}{-0cm}{-0cm}
    				\begin{center}
    					\scalebox{0.95}{\begin{tabular}{|*{8}{c|}}
    							\hline
    							\multicolumn{2}{|c|}{} & Minimum & 1st Quantile & Median & Mean & 3rd Quantile & Maximum \\
    							\hline
    							\multicolumn{2}{|c|}{Actual $y_t(\bm{s}_i)$'s} & 1  &  2369 &   4252  &  4941  &  6736  & 24221 \\
    							\hline
    							\multirow{2}{*}{$U_t(\bm{s}_i)$'s} & \texttt{hypara1} & 0.019 & 2373.657 & 4255.920 & 4941.258 & 6739.421 & 24875.179  \\
    							\cline{2-8}
    							& \texttt{hypara2} & 0.008 & 2373.594 & 4255.815 & 4941.266 & 6739.462 & 24775.287  \\
    							\hline
    							\multirow{2}{*}{$c$} & \texttt{hypara1} & 945.6  & 986.8  & 997.2 &  996.9 & 1007.2 & 1044.5   \\
    							\cline{2-8}
    							& \texttt{hypara2} & 937.7  & 985.6  & 994.8  & 994.9 & 1004.0 & 1046.7 \\
    							\hline
    							\multirow{2}{*}{$\kappa$} & \texttt{hypara1} & 0.3727 & 0.4007 & 0.4066 & 0.4067 & 0.4128 & 0.4399    \\
    							\cline{2-8}
    							& \texttt{hypara2} & 0.3755 & 0.4009 & 0.4072 & 0.4073 & 0.4134 & 0.4400 \\
    							\hline
    							\multirow{2}{*}{$\rho$} & \texttt{hypara1} & 0.3672 & 0.3935 & 0.3993 & 0.3993 & 0.4052 & 0.4348 \\
    							\cline{2-8}
    							& \texttt{hypara2} & 0.3696 & 0.3930 & 0.3992 & 0.3991 & 0.4053 & 0.4260 \\
    							\hline
    					\end{tabular}}
    				\end{center}
    			\end{adjustwidth}
    			\caption{Summary statistics of the actual observed counts $y_{1:T}(\bm{s}_{1:m})$ and the kept $10^4$ posterior samples of the parameters $c,\kappa, \rho,U_{1:T}(\bm{s}_{1:m})$ for hyperparameter specifications \texttt{hypara1} and \texttt{hypara2} with true parameter values $(\kappa,\rho,c) = (0.4, 0.4, 1000)$ and undirected spatial locations from the first simulation group.}
    			\label{c1000complete}
    		}
    	\end{table}
    	
    	\begin{table}[h]
    		{
    			\begin{adjustwidth}{-0cm}{-0cm}
    				\begin{center}
    					\scalebox{0.95}{\begin{tabular}{|*{8}{c|}}
    							\hline
    							\multicolumn{2}{|c|}{} & Minimum & 1st Quantile & Median & Mean & 3rd Quantile & Maximum \\
    							\hline
    							\multicolumn{2}{|c|}{Actual $y_t(\bm{s}_i)$'s} & 5   & 4767 &   8529  &  9882 &  13464 &  48245 \\
    							\hline
    							\multirow{2}{*}{$U_t(\bm{s}_i)$'s} & \texttt{hypara1} & 0.63 & 4756.08 & 8533.04 & 9882.26 & 13479.89 & 49040.28  \\
    							\cline{2-8}
    							& \texttt{hypara2} & 0.82 & 4756.05 & 8533.17 & 9882.24 & 13479.89 & 49053.01\\
    							\hline
    							\multirow{2}{*}{$c$} & \texttt{hypara1} & 1886  &  1971  &  1990 &   1990  &  2009  &  2100\\
    							\cline{2-8}
    							& \texttt{hypara2} & 1904  &  1974  &  1994  &  1994  &  2013  &  2105\\
    							\hline
    							\multirow{2}{*}{$\kappa$} & \texttt{hypara1} &  0.3673 & 0.3999 & 0.4061 & 0.4061 & 0.4123 & 0.4407  \\
    							\cline{2-8}
    							& \texttt{hypara2} & 0.3702 & 0.4006 & 0.4072 & 0.4072 & 0.4138 & 0.4389\\
    							\hline
    							\multirow{2}{*}{$\rho$} & \texttt{hypara1} & 0.3694 & 0.3941 & 0.4000 & 0.4001 & 0.4062 & 0.4398\\
    							\cline{2-8}
    							& \texttt{hypara2} & 0.3634 & 0.3927 & 0.3989 & 0.3989 & 0.4051 & 0.4328\\
    							\hline
    					\end{tabular}}
    				\end{center}
    			\end{adjustwidth}
    			\caption{Summary statistics of the actual observed counts $y_{1:T}(\bm{s}_{1:m})$ and the kept $10^4$ posterior samples of the parameters $c,\kappa, \rho,U_{1:T}(\bm{s}_{1:m})$ for hyperparameter specifications \texttt{hypara1} and \texttt{hypara2} with true parameter values $(\kappa,\rho,c) = (0.4, 0.4, 2000)$ and undirected spatial locations from the first simulation group.}
    			\label{c2000complete}
    		}
    	\end{table}
    	
    	\begin{table}[h]
    		{
    			\begin{adjustwidth}{-0cm}{-0cm}
    				\begin{center}
    					\scalebox{0.95}{\begin{tabular}{|*{8}{c|}}
    							\hline
    							\multicolumn{2}{|c|}{} & Minimum & 1st Quantile & Median & Mean & 3rd Quantile & Maximum \\
    							\hline
    							\multicolumn{2}{|c|}{Actual $y_t(\bm{s}_i)$'s} & 11  & 11893 &  21324  & 24706 &  33732 & 121817 \\
    							\hline
    							\multirow{2}{*}{$U_t(\bm{s}_i)$'s} & \texttt{hypara1} & 3.52 & 11889.96 & 21319.67 & 24706.13 & 33733.67 & 123024.56  \\
    							\cline{2-8}
    							& \texttt{hypara2} & 2.17 & 11890.13 & 21319.76 & 24706.12 & 33733.53 & 123114.89 \\
    							\hline
    							\multirow{2}{*}{$c$} & \texttt{hypara1} & 4730  &  4933 &   4981  &  4983  &  5033  &  5258\\
    							\cline{2-8}
    							& \texttt{hypara2} & 4736   & 4926 &   4972  &  4973  &  5020 &   5244\\
    							\hline
    							\multirow{2}{*}{$\kappa$} & \texttt{hypara1} &  0.3788 & 0.4012 & 0.4074 & 0.4072 & 0.4133 & 0.4432 \\
    							\cline{2-8}
    							& \texttt{hypara2} & 0.3727 & 0.4003 & 0.4064 & 0.4064 & 0.4125 & 0.4380 \\
    							\hline
    							\multirow{2}{*}{$\rho$} & \texttt{hypara1} & 0.3594 & 0.3933 & 0.3988 & 0.3989 & 0.4047 & 0.4283\\
    							\cline{2-8}
    							& \texttt{hypara2} & 0.3672 & 0.3939 & 0.4000 & 0.3999 & 0.4058 & 0.4338  \\
    							\hline
    					\end{tabular}}
    				\end{center}
    			\end{adjustwidth}
    			\caption{Summary statistics of the actual observed counts $y_{1:T}(\bm{s}_{1:m})$ and the kept $10^4$ posterior samples of the parameters $c,\kappa, \rho,U_{1:T}(\bm{s}_{1:m})$ for hyperparameter specifications \texttt{hypara1} and \texttt{hypara2} with true parameter values $(\kappa,\rho,c) = (0.4, 0.4, 5000)$ and undirected spatial locations from the first simulation group.}
    			\label{c5000complete}
    		}
    	\end{table}
    	\clearpage   
    	\begin{table}[H]
    		{
    			\begin{adjustwidth}{-0cm}{-0cm}
    				\begin{center}
    					\scalebox{0.95}{\begin{tabular}{|*{8}{c|}}
    							\hline
    							\multicolumn{2}{|c|}{} & Minimum & 1st Quantile & Median & Mean & 3rd Quantile & Maximum \\
    							\hline
    							\multicolumn{2}{|c|}{Actual $y_t(\bm{s}_i)$'s} & 0.00  &  7.00  & 14.00 &  16.56  & 23.00 &  92.00   \\
    							\hline
    							\multirow{2}{*}{$U_t(\bm{s}_i)$'s} & \texttt{hypara3} & 0.000 &  7.261 & 14.010 & 16.557 & 23.091 & 122.156  \\
    							\cline{2-8}
    							& \texttt{hypara4} & 0.000  & 7.258 & 14.009 & 16.557 & 23.093 & 116.952 \\
    							\hline
    							\multirow{2}{*}{$c$} & \texttt{hypara3} & 4.725 & 5.008 & 5.065 &   5.067 & 5.124 & 5.448   \\
    							\cline{2-8}
    							& \texttt{hypara4} & 4.743  & 5.019  & 5.080 &  5.081 &  5.143 &  5.430 \\
    							\hline
    							\multirow{2}{*}{$\kappa$} & \texttt{hypara3} &  0.6682 & 0.6961 & 0.7010 & 0.7010 & 0.7061 & 0.7317    \\
    							\cline{2-8}
    							& \texttt{hypara4} & 0.6732 & 0.6951 & 0.7002 & 0.7003 & 0.7054 & 0.7312 \\
    							\hline
    					\end{tabular}}
    				\end{center}
    			\end{adjustwidth}
    			\caption{Summary statistics of the actual observed counts $y_{1:T}(\bm{s}_{1:m})$ and the kept $10^4$ posterior samples of the parameters $c,\kappa, U_{1:T}(\bm{s}_{1:m})$ for the two hyperparameter specifications (\texttt{hypara3} and \texttt{hypara4}) when the true parameter values are $(\kappa,c) = (0.7, 5)$ and the spatial locations are undirected from the second simulation group.}
    			\label{undirectedkappa0.7c5}
    		}
    	\end{table}
    	
    	\begin{table}[H]
    		{
    			\begin{adjustwidth}{-0cm}{-0cm}
    				\begin{center}
    					\scalebox{0.95}{\begin{tabular}{|*{8}{c|}}
    							\hline
    							\multicolumn{2}{|c|}{} & Minimum & 1st Quantile & Median & Mean & 3rd Quantile & Maximum \\
    							\hline
    							\multicolumn{2}{|c|}{Actual $y_t(\bm{s}_i)$'s} & 0.00  & 14.00 &  28.00  & 33.11  & 46.00 & 181.00 \\
    							\hline
    							\multirow{2}{*}{$U_t(\bm{s}_i)$'s} & \texttt{hypara3} & 0.00  & 14.50 &  27.98  & 33.11 &  46.12 & 222.57  \\
    							\cline{2-8}
    							& \texttt{hypara4} & 0.00  & 14.50  & 27.99  & 33.11  & 46.12 & 221.75 \\
    							\hline
    							\multirow{2}{*}{$c$} & \texttt{hypara3} & 9.455 & 10.051 & 10.162 & 10.164 & 10.275 & 10.819  \\
    							\cline{2-8}
    							& \texttt{hypara4} & 9.657 & 10.053 & 10.170 & 10.170 & 10.283 & 10.791\\
    							\hline
    							\multirow{2}{*}{$\kappa$} & \texttt{hypara3} &  0.6717 & 0.6941 & 0.6990 & 0.6990 & 0.7040 & 0.7281   \\
    							\cline{2-8}
    							& \texttt{hypara4} & 0.6720 & 0.6941 & 0.6989 & 0.6989 & 0.7038 & 0.7244  \\
    							\hline
    					\end{tabular}}
    				\end{center}
    			\end{adjustwidth}
    			\caption{Summary statistics of the actual observed counts $y_{1:T}(\bm{s}_{1:m})$ and the kept $10^4$ posterior samples of the parameters $c,\kappa, U_{1:T}(\bm{s}_{1:m})$ for the two hyperparameter specifications (\texttt{hypara3} and \texttt{hypara4}) when the true parameter values are $(\kappa,c) = (0.7, 10)$ and the spatial locations are undirected from the second simulation group.}
    			\label{undirectedkappa0.7c10}
    		}
    	\end{table}
    	
    	\begin{table}[h]
    		{
    			\begin{adjustwidth}{-0cm}{-0cm}
    				\begin{center}
    					\scalebox{0.95}{\begin{tabular}{|*{8}{c|}}
    							\hline
    							\multicolumn{2}{|c|}{} & Minimum & 1st Quantile & Median & Mean & 3rd Quantile & Maximum \\
    							\hline
    							\multicolumn{2}{|c|}{Actual $y_t(\bm{s}_i)$'s} & 0.0  &  21.0  &  42.0  &  49.8  &  70.0  & 287.0    \\
    							\hline
    							\multirow{2}{*}{$U_t(\bm{s}_i)$'s} & \texttt{hypara3} & 0.00 &  21.94 &  42.29  & 49.80 &  69.36 & 342.79  \\
    							\cline{2-8}
    							& \texttt{hypara4} & 0.00 &  21.94 &  42.29  & 49.80 &  69.36 & 341.63 \\
    							\hline
    							\multirow{2}{*}{$c$} & \texttt{hypara3} & 14.30  & 15.04  & 15.21 &  15.21 &  15.38 &  16.24  \\
    							\cline{2-8}
    							& \texttt{hypara4} & 14.25  & 15.04  & 15.21  & 15.20 &  15.37 &  16.22 \\
    							\hline
    							\multirow{2}{*}{$\kappa$} & \texttt{hypara3} &  0.6732 & 0.6961 & 0.7010 & 0.7010 & 0.7058 & 0.7284    \\
    							\cline{2-8}
    							& \texttt{hypara4} & 0.6737 & 0.6964 & 0.7012 & 0.7012 & 0.7060 & 0.7265  \\
    							\hline
    					\end{tabular}}
    				\end{center}
    			\end{adjustwidth}
    			\caption{Summary statistics of the actual observed counts $y_{1:T}(\bm{s}_{1:m})$ and the kept $10^4$ posterior samples of the parameters $c,\kappa, U_{1:T}(\bm{s}_{1:m})$ for the two hyperparameter specifications (\texttt{hypara3} and \texttt{hypara4}) when the true parameter values are $(\kappa,c) = (0.7, 15)$ and the spatial locations are undirected from the second simulation group.}
    			\label{undirectedkappa0.7c15}
    		}
    	\end{table}
    	
    	\begin{table}[h]
    		{
    			\begin{adjustwidth}{-0cm}{-0cm}
    				\begin{center}
    					\scalebox{0.95}{\begin{tabular}{|*{8}{c|}}
    							\hline
    							\multicolumn{2}{|c|}{} & Minimum & 1st Quantile & Median & Mean & 3rd Quantile & Maximum \\
    							\hline
    							\multicolumn{2}{|c|}{Actual $y_t(\bm{s}_i)$'s} & 0.00  & 28.00   &56.00 &  66.32 &  92.00 & 381.00    \\
    							\hline
    							\multirow{2}{*}{$U_t(\bm{s}_i)$'s} & \texttt{hypara3} & 0.00 &  29.06  & 56.08  & 66.32 &  92.60 & 440.38   \\
    							\cline{2-8}
    							& \texttt{hypara4} & 0.00 &  29.06 &  56.08  & 66.32 &  92.61 & 444.60   \\
    							\hline
    							\multirow{2}{*}{$c$} & \texttt{hypara3} & 19.27  & 20.13   & 20.35 &  20.35  & 20.57  & 21.66  \\
    							\cline{2-8}
    							& \texttt{hypara4} & 19.13  & 20.16 &  20.38 &  20.38 &  20.60 &  21.73 \\
    							\hline
    							\multirow{2}{*}{$\kappa$} & \texttt{hypara3} & 0.6719 & 0.6947 & 0.6994 & 0.6995 & 0.7044 & 0.7262    \\
    							\cline{2-8}
    							& \texttt{hypara4} & 0.6714 & 0.6942 & 0.6991 & 0.6991 & 0.7040 & 0.7283  \\
    							\hline
    					\end{tabular}}
    				\end{center}
    			\end{adjustwidth}
    			\caption{Summary statistics of the actual observed counts $y_{1:T}(\bm{s}_{1:m})$ and the kept $10^4$ posterior samples of the parameters $c,\kappa, U_{1:T}(\bm{s}_{1:m})$ for the two hyperparameter specifications (\texttt{hypara3} and \texttt{hypara4}) when the true parameter values are $(\kappa,c) = (0.7, 20)$ and the spatial locations are undirected from the second simulation group.}
    			\label{undirectedkappa0.7c20}
    		}
    	\end{table}
    	
    	\begin{table}[h]
    		{
    			\begin{adjustwidth}{-0cm}{-0cm}
    				\begin{center}
    					\scalebox{0.95}{\begin{tabular}{|*{8}{c|}}
    							\hline
    							\multicolumn{2}{|c|}{} & Minimum & 1st Quantile & Median & Mean & 3rd Quantile & Maximum \\
    							\hline
    							\multicolumn{2}{|c|}{Actual $y_t(\bm{s}_i)$'s} & 0.0  &  72.0 &  140.0  & 165.8 &  233.0  & 976.0  \\
    							\hline
    							\multirow{2}{*}{$U_t(\bm{s}_i)$'s} & \texttt{hypara3} & 0.00  & 73.22 & 140.37 & 165.75 & 232.27 & 1081.72 \\
    							\cline{2-8}
    							& \texttt{hypara4} & 0.00  & 73.22 & 140.37 & 165.75 & 232.27 & 1102.49  \\
    							\hline
    							\multirow{2}{*}{$c$} & \texttt{hypara3} &  48.10  & 50.23 &  50.79 &  50.80  & 51.33 &  53.55 \\
    							\cline{2-8}
    							& \texttt{hypara4} & 47.44 &  50.34  & 50.87  & 50.87  & 51.40 &  54.16\\
    							\hline
    							\multirow{2}{*}{$\kappa$} & \texttt{hypara3} & 0.6718 & 0.6952 & 0.7001 & 0.7000 & 0.7048 & 0.7295 \\
    							\cline{2-8}
    							& \texttt{hypara4} & 0.6732 & 0.6949 & 0.6998 & 0.6998 & 0.7046 & 0.7304 \\
    							\hline
    					\end{tabular}}
    				\end{center}
    			\end{adjustwidth}
    			\caption{Summary statistics of the actual observed counts $y_{1:T}(\bm{s}_{1:m})$ and the kept $10^4$ posterior samples of the parameters $c,\kappa, U_{1:T}(\bm{s}_{1:m})$ for the two hyperparameter specifications (\texttt{hypara3} and \texttt{hypara4}) when the true parameter values are $(\kappa,c) = (0.7, 50)$ and the spatial locations are undirected from the second simulation group.}
    			\label{undirectedkappa0.7c50}
    		}
    	\end{table}
    	
    	\begin{table}[h]
    		{
    			\begin{adjustwidth}{-0cm}{-0cm}
    				\begin{center}
    					\scalebox{0.95}{\begin{tabular}{|*{8}{c|}}
    							\hline
    							\multicolumn{2}{|c|}{} & Minimum & 1st Quantile & Median & Mean & 3rd Quantile & Maximum \\
    							\hline
    							\multicolumn{2}{|c|}{Actual $y_t(\bm{s}_i)$'s} & 0.0  & 145.0 &  279.0  & 331.8  & 463.0 & 1936.0\\
    							\hline
    							\multirow{2}{*}{$U_t(\bm{s}_i)$'s} & \texttt{hypara3} & 0.0 & 145.7 &  280.3 & 331.8 & 462.9 & 2094.2 \\
    							\cline{2-8}
    							& \texttt{hypara4} & 0.0 &  145.7 &  280.3  & 331.8 &  462.9 & 2109.4\\
    							\hline
    							\multirow{2}{*}{$c$} & \texttt{hypara3} & 96.44 & 100.96 & 101.99 & 102.02 & 103.06 & 107.52 \\
    							\cline{2-8}
    							& \texttt{hypara4} & 95.7 &  101.0  & 102.0  & 102.1  & 103.1  & 108.6  \\
    							\hline
    							\multirow{2}{*}{$\kappa$} & \texttt{hypara3} & 0.6717 & 0.6940 & 0.6989 & 0.6988 & 0.7037 & 0.7237  \\
    							\cline{2-8}
    							& \texttt{hypara4} & 0.6722 & 0.6938 & 0.6986 & 0.6986 & 0.7035 & 0.7292 \\
    							\hline
    					\end{tabular}}
    				\end{center}
    			\end{adjustwidth}
    			\caption{Summary statistics of the actual observed counts $y_{1:T}(\bm{s}_{1:m})$ and the kept $10^4$ posterior samples of the parameters $c,\kappa, U_{1:T}(\bm{s}_{1:m})$ for the two hyperparameter specifications (\texttt{hypara3} and \texttt{hypara4}) when the true parameter values are $(\kappa,c) = (0.7, 100)$ and the spatial locations are undirected from the second simulation group.}
    			\label{undirectedkappa0.7c100}
    		}
    	\end{table}
    	
    	\begin{table}[h]
    		{
    			\begin{adjustwidth}{-0cm}{-0cm}
    				\begin{center}
    					\scalebox{0.95}{\begin{tabular}{|*{8}{c|}}
    							\hline
    							\multicolumn{2}{|c|}{} & Minimum & 1st Quantile & Median & Mean & 3rd Quantile & Maximum \\
    							\hline
    							\multicolumn{2}{|c|}{Actual $y_t(\bm{s}_i)$'s} & 0  &   725 &   1401  &  1659 &   2327  &  9047  \\
    							\hline
    							\multirow{2}{*}{$U_t(\bm{s}_i)$'s} & \texttt{hypara3} & 0.0  & 726.5 & 1402.4 & 1658.8 & 2324.5 & 9380.4 \\
    							\cline{2-8}
    							& \texttt{hypara4} & 0.0  & 726.6 & 1402.3 & 1658.8 & 2324.5 & 9428.0 \\
    							\hline
    							\multirow{2}{*}{$c$} & \texttt{hypara3} & 481.8  & 504.6 &  509.8  & 509.9  & 515.1 &  539.0  \\
    							\cline{2-8}
    							& \texttt{hypara4} & 484.4  & 504.4  & 509.8 &  509.8  & 515.0  & 541.7 \\
    							\hline
    							\multirow{2}{*}{$\kappa$} & \texttt{hypara3} & 0.6730 & 0.6944 &  0.6992 & 0.6992 & 0.7040 & 0.7293  \\
    							\cline{2-8}
    							& \texttt{hypara4} & 0.6676 & 0.6943 & 0.6992 & 0.6992 & 0.7041 & 0.7245  \\
    							\hline
    					\end{tabular}}
    				\end{center}
    			\end{adjustwidth}
    			\caption{Summary statistics of the actual observed counts $y_{1:T}(\bm{s}_{1:m})$ and the kept $10^4$ posterior samples of the parameters $c,\kappa, U_{1:T}(\bm{s}_{1:m})$ for the two hyperparameter specifications (\texttt{hypara3} and \texttt{hypara4}) when the true parameter values are $(\kappa,c) = (0.7, 500)$ and the spatial locations are undirected from the second simulation group.}
    			\label{undirectedkappa0.7c500}
    		}
    	\end{table}
    	
    	\begin{table}[h]
    		{
    			\begin{adjustwidth}{-0cm}{-0cm}
    				\begin{center}
    					\scalebox{0.95}{\begin{tabular}{|*{8}{c|}}
    							\hline
    							\multicolumn{2}{|c|}{} & Minimum & 1st Quantile & Median & Mean & 3rd Quantile & Maximum \\
    							\hline
    							\multicolumn{2}{|c|}{Actual $y_t(\bm{s}_i)$'s} & 0  &  1451 &   2814 &   3318   & 4644  & 18528  \\
    							\hline
    							\multirow{2}{*}{$U_t(\bm{s}_i)$'s} & \texttt{hypara3} & 0 &   1451  &  2809  &  3318  &  4644  & 19012  \\
    							\cline{2-8}
    							& \texttt{hypara4} & 0  &  1451  &  2809  &  3318 &   4644 &  18996 \\
    							\hline
    							\multirow{2}{*}{$c$} & \texttt{hypara3} & 964.3 & 1009.6 & 1019.7 & 1019.9 & 1029.9 & 1081.7  \\
    							\cline{2-8}
    							& \texttt{hypara4} & 964.8 & 1009.2 & 1019.2 & 1019.4 & 1029.3 & 1078.8 \\
    							\hline
    							\multirow{2}{*}{$\kappa$} & \texttt{hypara3} & 0.6718 & 0.6943 & 0.6991 & 0.6990 & 0.7038 & 0.7289  \\
    							\cline{2-8}
    							& \texttt{hypara4} & 0.6735 & 0.6944 & 0.6992 & 0.6993 & 0.7041 & 0.7235\\
    							\hline
    					\end{tabular}}
    				\end{center}
    			\end{adjustwidth}
    			\caption{Summary statistics of the actual observed counts $y_{1:T}(\bm{s}_{1:m})$ and the kept $10^4$ posterior samples of the parameters $c,\kappa, U_{1:T}(\bm{s}_{1:m})$ for the two hyperparameter specifications (\texttt{hypara3} and \texttt{hypara4}) when the true parameter values are $(\kappa,c) = (0.7, 1000)$ and the spatial locations are undirected from the second simulation group.}
    			\label{undirectedkappa0.7c1000}
    		}
    	\end{table}
    	
    	\begin{table}[h]
    		{
    			\begin{adjustwidth}{-0cm}{-0cm}
    				\begin{center}
    					\scalebox{0.95}{\begin{tabular}{|*{8}{c|}}
    							\hline
    							\multicolumn{2}{|c|}{} & Minimum & 1st Quantile & Median & Mean & 3rd Quantile & Maximum \\
    							\hline
    							\multicolumn{2}{|c|}{Actual $y_t(\bm{s}_i)$'s} & 1  &  2909  &  5608  &  6636 &   9279  & 36824 \\
    							\hline
    							\multirow{2}{*}{$U_t(\bm{s}_i)$'s} & \texttt{hypara3} & 0.02 & 2906.28 & 5611.68 & 6635.58 & 9283.75  & 37839.95   \\
    							\cline{2-8}
    							& \texttt{hypara4} & 0.01 & 2906.44 & 5611.76 & 6635.57 & 9283.54 & 37524.77 \\
    							\hline
    							\multirow{2}{*}{$c$} & \texttt{hypara3} & 1914  &  2017  &  2036  &  2037  &  2057  &  2170  \\
    							\cline{2-8}
    							& \texttt{hypara4} & 1922  &  2017   & 2038  &  2038  &  2059 &   2136  \\
    							\hline
    							\multirow{2}{*}{$\kappa$} & \texttt{hypara3} & 0.6704 & 0.6946 & 0.6994 & 0.6993 & 0.7041 & 0.7252  \\
    							\cline{2-8}
    							& \texttt{hypara4} & 0.6731 & 0.6944 & 0.6993 & 0.6993 & 0.7041 & 0.7265  \\
    							\hline
    					\end{tabular}}
    				\end{center}
    			\end{adjustwidth}
    			\caption{Summary statistics of the actual observed counts $y_{1:T}(\bm{s}_{1:m})$ and the kept $10^4$ posterior samples of the parameters $c,\kappa, U_{1:T}(\bm{s}_{1:m})$ for the two hyperparameter specifications (\texttt{hypara3} and \texttt{hypara4}) when the true parameter values are $(\kappa,c) = (0.7, 2000)$ and the spatial locations are undirected from the second simulation group.}
    			\label{undirectedkappa0.7c2000}
    		}
    	\end{table}
    	
    	\begin{table}[h]
    		{
    			\begin{adjustwidth}{-0cm}{-0cm}
    				\begin{center}
    					\scalebox{0.95}{\begin{tabular}{|*{8}{c|}}
    							\hline
    							\multicolumn{2}{|c|}{} & Minimum & 1st Quantile & Median & Mean & 3rd Quantile & Maximum \\
    							\hline
    							\multicolumn{2}{|c|}{Actual $y_t(\bm{s}_i)$'s} & 5  &  7286  & 14010  & 16589 &  23207 &  92135   \\
    							\hline
    							\multirow{2}{*}{$U_t(\bm{s}_i)$'s} & \texttt{hypara3} & 0.68 & 7274.32 & 14023.42 & 16588.77 & 23197.09 & 93506.65  \\
    							\cline{2-8}
    							& \texttt{hypara4} & 0.81 & 7274.13 & 14023.49 & 16588.73 & 23197.07 & 93122.49 \\
    							\hline
    							\multirow{2}{*}{$c$} & \texttt{hypara3} & 4803  &  5035  &  5087  &  5087  &  5139  &  5353   \\
    							\cline{2-8}
    							& \texttt{hypara4} & 4836   & 5041 &   5092  &  5093  &  5146  &  5396 \\
    							\hline
    							\multirow{2}{*}{$\kappa$} & \texttt{hypara3} &  0.6716 & 0.6949 & 0.6997 & 0.6996 & 0.7044 & 0.7246  \\
    							\cline{2-8}
    							& \texttt{hypara4} & 0.6695 & 0.6944 & 0.6993 & 0.6993 & 0.7041 & 0.7259 \\
    							\hline
    					\end{tabular}}
    				\end{center}
    			\end{adjustwidth}
    			\caption{Summary statistics of the actual observed counts $y_{1:T}(\bm{s}_{1:m})$ and the kept $10^4$ posterior samples of the parameters $c,\kappa, U_{1:T}(\bm{s}_{1:m})$ for the two hyperparameter specifications (\texttt{hypara3} and \texttt{hypara4}) when the true parameter values are $(\kappa,c) = (0.7, 5000)$ and the spatial locations are undirected from the second simulation group.}
    			\label{undirectedkappa0.7c5000}
    		}
    	\end{table}
    	\clearpage  
    	
    	\begin{table}[H]
    		{
    			\begin{adjustwidth}{-0cm}{-0cm}
    				\begin{center}
    					\scalebox{0.95}{\begin{tabular}{|*{8}{c|}}
    							\hline
    							\multicolumn{2}{|c|}{} & Minimum & 1st Quantile & Median & Mean & 3rd Quantile & Maximum \\
    							\hline
    							\multicolumn{2}{|c|}{Actual $y_t(\bm{s}_i)$'s} & 0.000  & 2.000 &  6.000  & 7.829 & 11.000 & 63.000 \\
    							\hline
    							\multirow{2}{*}{$U_t(\bm{s}_i)$'s} & \texttt{hypara3} & 0.000  & 2.431 &  5.701  & 7.829 & 10.994 & 87.799  \\
    							\cline{2-8}
    							& \texttt{hypara4} & 0.000  & 2.430 &  5.700 &  7.829 & 10.993 & 89.944  \\
    							\hline
    							\multirow{2}{*}{$c$} & \texttt{hypara3} & 4.472  & 4.850 &  4.938  & 4.941 &  5.030 &  5.416  \\
    							\cline{2-8}
    							& \texttt{hypara4} & 4.511  & 4.874  & 4.961 &  4.964  & 5.049  & 5.569 \\
    							\hline
    							\multirow{2}{*}{$\kappa$} & \texttt{hypara3} & 0.3089 & 0.3604 & 0.3722 & 0.3722 & 0.3841 & 0.4362 \\
    							\cline{2-8}
    							& \texttt{hypara4} & 0.3047 & 0.3585 & 0.3698 & 0.3696 & 0.3811 & 0.4336  \\
    							\hline
    					\end{tabular}}
    				\end{center}
    			\end{adjustwidth}
    			\caption{Summary statistics of the actual observed counts $y_{1:T}(\bm{s}_{1:m})$ and the kept $10^4$ posterior samples of the parameters $c,\kappa, U_{1:T}(\bm{s}_{1:m})$ for the two hyperparameter specifications (\texttt{hypara3} and \texttt{hypara4}) when the true parameter values are $(\kappa,c) = (0.35, 5)$ and the spatial locations are undirected from the second simulation group.}
    			\label{undirectedkappa0.35c5}
    		}
    	\end{table}
    	
    	\begin{table}[H]
    		{
    			\begin{adjustwidth}{-0cm}{-0cm}
    				\begin{center}
    					\scalebox{0.95}{\begin{tabular}{|*{8}{c|}}
    							\hline
    							\multicolumn{2}{|c|}{} & Minimum & 1st Quantile & Median & Mean & 3rd Quantile & Maximum \\
    							\hline
    							\multicolumn{2}{|c|}{Actual $y_t(\bm{s}_i)$'s} & 0.00  &  4.00 &  11.00  & 15.59  & 22.00 & 135.00 \\
    							\hline
    							\multirow{2}{*}{$U_t(\bm{s}_i)$'s} & \texttt{hypara3} & 0.000 &  4.828 & 11.338 & 15.589 & 21.872 & 175.830   \\
    							\cline{2-8}
    							& \texttt{hypara4} & 0.000  & 4.826 & 11.338 & 15.589 & 21.873 & 170.927 \\
    							\hline
    							\multirow{2}{*}{$c$} & \texttt{hypara3} & 9.005 & 9.673 &  9.844  & 9.849 & 10.024 & 10.819   \\
    							\cline{2-8}
    							& \texttt{hypara4} & 9.168  & 9.725 &  9.890 &  9.901 & 10.062 & 11.089 \\
    							\hline
    							\multirow{2}{*}{$\kappa$} & \texttt{hypara3} & 0.3100 & 0.3598 & 0.3711 & 0.3711 & 0.3824 & 0.4283  \\
    							\cline{2-8}
    							& \texttt{hypara4} & 0.3005 & 0.3575 & 0.3686 & 0.3680 & 0.3792 & 0.4226 \\
    							\hline
    					\end{tabular}}
    				\end{center}
    			\end{adjustwidth}
    			\caption{Summary statistics of the actual observed counts $y_{1:T}(\bm{s}_{1:m})$ and the kept $10^4$ posterior samples of the parameters $c,\kappa, U_{1:T}(\bm{s}_{1:m})$ for the two hyperparameter specifications (\texttt{hypara3} and \texttt{hypara4}) when the true parameter values are $(\kappa,c) = (0.35, 10)$ and the spatial locations are undirected from the second simulation group.}
    			\label{undirectedkappa0.35c10}
    		}
    	\end{table}
    	
    	\begin{table}[h]
    		{
    			\begin{adjustwidth}{-0cm}{-0cm}
    				\begin{center}
    					\scalebox{0.95}{\begin{tabular}{|*{8}{c|}}
    							\hline
    							\multicolumn{2}{|c|}{} & Minimum & 1st Quantile & Median & Mean & 3rd Quantile & Maximum \\
    							\hline
    							\multicolumn{2}{|c|}{Actual $y_t(\bm{s}_i)$'s} & 0.00  &  7.00 &  17.00 &  23.43 &  33.00 & 184.00 \\
    							\hline
    							\multirow{2}{*}{$U_t(\bm{s}_i)$'s} & \texttt{hypara3} & 0.000 &  7.277 & 17.021 & 23.430 & 32.701 & 232.022   \\
    							\cline{2-8}
    							& \texttt{hypara4} & 0.000 &  7.276 & 17.020 & 23.432 & 32.705 & 233.137 \\
    							\hline
    							\multirow{2}{*}{$c$} & \texttt{hypara3} & 13.47 &  14.58 &  14.83 &  14.83  & 15.07  & 16.38   \\
    							\cline{2-8}
    							& \texttt{hypara4} & 13.55 &  14.61  & 14.86 &  14.87 &  15.12 &  16.22  \\
    							\hline
    							\multirow{2}{*}{$\kappa$} & \texttt{hypara3} & 0.3102 & 0.3594 & 0.3704 & 0.3699 & 0.3809 & 0.4235   \\
    							\cline{2-8}
    							& \texttt{hypara4} & 0.3078 & 0.3580 & 0.3688 & 0.3684 & 0.3794 & 0.4218\\
    							\hline
    					\end{tabular}}
    				\end{center}
    			\end{adjustwidth}
    			\caption{Summary statistics of the actual observed counts $y_{1:T}(\bm{s}_{1:m})$ and the kept $10^4$ posterior samples of the parameters $c,\kappa, U_{1:T}(\bm{s}_{1:m})$ for the two hyperparameter specifications (\texttt{hypara3} and \texttt{hypara4}) when the true parameter values are $(\kappa,c) = (0.35, 15)$ and the spatial locations are undirected from the second simulation group.}
    			\label{undirectedkappa0.35c15}
    		}
    	\end{table}
    	
    	\begin{table}[h]
    		{
    			\begin{adjustwidth}{-0cm}{-0cm}
    				\begin{center}
    					\scalebox{0.95}{\begin{tabular}{|*{8}{c|}}
    							\hline
    							\multicolumn{2}{|c|}{} & Minimum & 1st Quantile & Median & Mean & 3rd Quantile & Maximum \\
    							\hline
    							\multicolumn{2}{|c|}{Actual $y_t(\bm{s}_i)$'s} &  0.00  &  9.00 &  22.00  & 31.26  & 43.00 & 246.00  \\
    							\hline
    							\multirow{2}{*}{$U_t(\bm{s}_i)$'s} & \texttt{hypara3} & 0.000  & 9.672 & 22.692 & 31.255 & 43.562 & 304.002  \\
    							\cline{2-8}
    							& \texttt{hypara4} & 0.000 &  9.669 & 22.692 & 31.255 & 43.562 &301.296 \\
    							\hline
    							\multirow{2}{*}{$c$} & \texttt{hypara3} & 17.81 &  19.52 &  19.86 &  19.86  & 20.20 &  21.58\\
    							\cline{2-8}
    							& \texttt{hypara4} & 18.30  & 19.53 &  19.85  & 19.87 &  20.19 &  21.81  \\
    							\hline
    							\multirow{2}{*}{$\kappa$} & \texttt{hypara3} & 0.3155 & 0.3568 & 0.3678 & 0.3677 & 0.3786 & 0.4340   \\
    							\cline{2-8}
    							& \texttt{hypara4} & 0.3059 & 0.3570 & 0.3679 & 0.3674 & 0.3782 & 0.4232 \\
    							\hline
    					\end{tabular}}
    				\end{center}
    			\end{adjustwidth}
    			\caption{Summary statistics of the actual observed counts $y_{1:T}(\bm{s}_{1:m})$ and the kept $10^4$ posterior samples of the parameters $c,\kappa, U_{1:T}(\bm{s}_{1:m})$ for the two hyperparameter specifications (\texttt{hypara3} and \texttt{hypara4}) when the true parameter values are $(\kappa,c) = (0.35, 20)$ and the spatial locations are undirected from the second simulation group.}
    			\label{undirectedkappa0.35c20}
    		}
    	\end{table}
    	
    	\begin{table}[h]
    		{
    			\begin{adjustwidth}{-0cm}{-0cm}
    				\begin{center}
    					\scalebox{0.95}{\begin{tabular}{|*{8}{c|}}
    							\hline
    							\multicolumn{2}{|c|}{} & Minimum & 1st Quantile & Median & Mean & 3rd Quantile & Maximum \\
    							\hline
    							\multicolumn{2}{|c|}{Actual $y_t(\bm{s}_i)$'s} & 0.00  & 24.00 &  57.00 &  78.22 & 109.00 & 604.00  \\
    							\hline
    							\multirow{2}{*}{$U_t(\bm{s}_i)$'s} & \texttt{hypara3} & 0.00 &  24.16 &  56.96 &  78.22 & 109.14 & 692.00    \\
    							\cline{2-8}
    							& \texttt{hypara4} & 0.00  & 24.16  & 56.96  & 78.22 & 109.14 & 687.09 \\
    							\hline
    							\multirow{2}{*}{$c$} & \texttt{hypara3} & 44.94  & 48.83  & 49.55 &  49.58 &  50.33  & 54.34    \\
    							\cline{2-8}
    							& \texttt{hypara4} & 45.71 &  49.00  & 49.77 &  49.79 &  50.53 &  54.50\\
    							\hline
    							\multirow{2}{*}{$\kappa$} & \texttt{hypara3} & 0.3062 & 0.3594 & 0.3694 & 0.3690 & 0.3793 & 0.4321    \\
    							\cline{2-8}
    							& \texttt{hypara4} & 0.3035 & 0.3564 & 0.3665 & 0.3662 & 0.3764 & 0.4174 \\
    							\hline
    					\end{tabular}}
    				\end{center}
    			\end{adjustwidth}
    			\caption{Summary statistics of the actual observed counts $y_{1:T}(\bm{s}_{1:m})$ and the kept $10^4$ posterior samples of the parameters $c,\kappa, U_{1:T}(\bm{s}_{1:m})$ for the two hyperparameter specifications (\texttt{hypara3} and \texttt{hypara4}) when the true parameter values are $(\kappa,c) = (0.35, 50)$ and the spatial locations are undirected from the second simulation group.}
    			\label{undirectedkappa0.35c50}
    		}
    	\end{table}
    	
    	\begin{table}[h]
    		{
    			\begin{adjustwidth}{-0cm}{-0cm}
    				\begin{center}
    					\scalebox{0.95}{\begin{tabular}{|*{8}{c|}}
    							\hline
    							\multicolumn{2}{|c|}{} & Minimum & 1st Quantile & Median & Mean & 3rd Quantile & Maximum \\
    							\hline
    							\multicolumn{2}{|c|}{Actual $y_t(\bm{s}_i)$'s} &  0.0  &  47.0 &  114.0  & 156.4 &  218.0 & 1176.0 \\
    							\hline
    							\multirow{2}{*}{$U_t(\bm{s}_i)$'s} & \texttt{hypara3} & 0.00 &  47.94 & 114.01 & 156.36 & 217.71 & 1314.87     \\
    							\cline{2-8}
    							& \texttt{hypara4} & 0.00  & 47.94 & 114.02 & 156.36 & 217.72 & 1302.91\\
    							\hline
    							\multirow{2}{*}{$c$} & \texttt{hypara3} & 91.18 & 97.44 &  99.02  & 99.07 & 100.65 & 107.99    \\
    							\cline{2-8}
    							& \texttt{hypara4} & 92.23  & 97.75  & 99.33 &  99.38 & 100.86 & 109.62\\
    							\hline
    							\multirow{2}{*}{$\kappa$} & \texttt{hypara3} & 0.3092 & 0.3593 & 0.3700 & 0.3696 & 0.3802 & 0.4243     \\
    							\cline{2-8}
    							& \texttt{hypara4} & 0.3070 & 0.3578 & 0.3680 & 0.3675 & 0.3779 & 0.4164\\
    							\hline
    					\end{tabular}}
    				\end{center}
    			\end{adjustwidth}
    			\caption{Summary statistics of the actual observed counts $y_{1:T}(\bm{s}_{1:m})$ and the kept $10^4$ posterior samples of the parameters $c,\kappa, U_{1:T}(\bm{s}_{1:m})$ for the two hyperparameter specifications (\texttt{hypara3} and \texttt{hypara4}) when the true parameter values are $(\kappa,c) = (0.35, 100)$ and the spatial locations are undirected from the second simulation group.}
    			\label{undirectedkappa0.35c100}
    		}
    	\end{table}
    	
    	\begin{table}[h]
    		{
    			\begin{adjustwidth}{-0cm}{-0cm}
    				\begin{center}
    					\scalebox{0.95}{\begin{tabular}{|*{8}{c|}}
    							\hline
    							\multicolumn{2}{|c|}{} & Minimum & 1st Quantile & Median & Mean & 3rd Quantile & Maximum \\
    							\hline
    							\multicolumn{2}{|c|}{Actual $y_t(\bm{s}_i)$'s} & 0.0 &  241.0  & 569.5 &  781.4 & 1091.0 & 5862.0\\
    							\hline
    							\multirow{2}{*}{$U_t(\bm{s}_i)$'s} & \texttt{hypara3} & 0.0  & 240.8  & 569.2  & 781.4 & 1089.8 & 6162.7      \\
    							\cline{2-8}
    							& \texttt{hypara4} & 0.0  & 240.8 &  569.2 &  781.4 & 1089.7 & 6130.4  \\
    							\hline
    							\multirow{2}{*}{$c$} & \texttt{hypara3} & 456.6 &  488.3 &  495.7 &  496.2 &  503.9  & 548.7    \\
    							\cline{2-8}
    							& \texttt{hypara4} & 457.1  & 488.1 &  495.6 &  495.6  & 502.9  & 538.8  \\
    							\hline
    							\multirow{2}{*}{$\kappa$} & \texttt{hypara3} & 0.3070 & 0.3580 & 0.3683 & 0.3679 & 0.3782 & 0.4190      \\
    							\cline{2-8}
    							& \texttt{hypara4} & 0.3148 & 0.3589 & 0.3687 & 0.3685 & 0.3781 & 0.4273   \\
    							\hline
    					\end{tabular}}
    				\end{center}
    			\end{adjustwidth}
    			\caption{Summary statistics of the actual observed counts $y_{1:T}(\bm{s}_{1:m})$ and the kept $10^4$ posterior samples of the parameters $c,\kappa, U_{1:T}(\bm{s}_{1:m})$ for the two hyperparameter specifications (\texttt{hypara3} and \texttt{hypara4}) when the true parameter values are $(\kappa,c) = (0.35, 500)$ and the spatial locations are undirected from the second simulation group.}
    			\label{undirectedkappa0.35c500}
    		}
    	\end{table}
    	
    	\begin{table}[h]
    		{
    			\begin{adjustwidth}{-0cm}{-0cm}
    				\begin{center}
    					\scalebox{0.95}{\begin{tabular}{|*{8}{c|}}
    							\hline
    							\multicolumn{2}{|c|}{} & Minimum & 1st Quantile & Median & Mean & 3rd Quantile & Maximum \\
    							\hline
    							\multicolumn{2}{|c|}{Actual $y_t(\bm{s}_i)$'s} &  0   &  480  &  1137 &   1562 &   2178  & 11699 \\
    							\hline
    							\multirow{2}{*}{$U_t(\bm{s}_i)$'s} & \texttt{hypara3} & 0.0 &  479.9 & 1138.1 & 1562.5 & 2177.0 & 12164.6      \\
    							\cline{2-8}
    							& \texttt{hypara4} & 0.0  & 479.9 & 1138.1 & 1562.5  & 2177.0 &12094.0 \\
    							\hline
    							\multirow{2}{*}{$c$} & \texttt{hypara3} & 912.7  & 975.2  & 991.0 &  991.2 & 1006.4 & 1111.3    \\
    							\cline{2-8}
    							& \texttt{hypara4} & 908.1 &  975.8 &  990.4  & 990.9 & 1005.4 & 1097.2 \\
    							\hline
    							\multirow{2}{*}{$\kappa$} & \texttt{hypara3} & 0.2967 & 0.3583 & 0.3687 & 0.3685 & 0.3790 & 0.4226      \\
    							\cline{2-8}
    							& \texttt{hypara4} & 0.3046 & 0.3594 & 0.3690 & 0.3688 & 0.3785 & 0.4316 \\
    							\hline
    					\end{tabular}}
    				\end{center}
    			\end{adjustwidth}
    			\caption{Summary statistics of the actual observed counts $y_{1:T}(\bm{s}_{1:m})$ and the kept $10^4$ posterior samples of the parameters $c,\kappa, U_{1:T}(\bm{s}_{1:m})$ for the two hyperparameter specifications (\texttt{hypara3} and \texttt{hypara4}) when the true parameter values are $(\kappa,c) = (0.35, 1000)$ and the spatial locations are undirected from the second simulation group.}
    			\label{undirectedkappa0.35c1000}
    		}
    	\end{table}
    	
    	\begin{table}[h]
    		{
    			\begin{adjustwidth}{-0cm}{-0cm}
    				\begin{center}
    					\scalebox{0.95}{\begin{tabular}{|*{8}{c|}}
    							\hline
    							\multicolumn{2}{|c|}{} & Minimum & 1st Quantile & Median & Mean & 3rd Quantile & Maximum \\
    							\hline
    							\multicolumn{2}{|c|}{Actual $y_t(\bm{s}_i)$'s} & 0  &   964 &   2279  &  3126   & 4364  & 23165  \\
    							\hline
    							\multirow{2}{*}{$U_t(\bm{s}_i)$'s} & \texttt{hypara3} & 0.0 &  961.2 & 2277.8 & 3125.6 & 4361.2 & 23735.7      \\
    							\cline{2-8}
    							& \texttt{hypara4} & 0.0  & 961.2 & 2277.7 & 3125.6 & 4361.1 & 23786.4\\
    							\hline
    							\multirow{2}{*}{$c$} & \texttt{hypara3} & 1823  &  1955 &   1985  &  1987  &  2018  &  2173    \\
    							\cline{2-8}
    							& \texttt{hypara4} & 1828  &  1957  &  1987  &  1987 &   2017   & 2142 \\
    							\hline
    							\multirow{2}{*}{$\kappa$} & \texttt{hypara3} & 0.3044 & 0.3571 & 0.3678 & 0.3672 & 0.3778 & 0.4208      \\
    							\cline{2-8}
    							& \texttt{hypara4} &  0.3181 & 0.3574 & 0.3671 & 0.3672 & 0.3768 & 0.4270 \\
    							\hline
    					\end{tabular}}
    				\end{center}
    			\end{adjustwidth}
    			\caption{Summary statistics of the actual observed counts $y_{1:T}(\bm{s}_{1:m})$ and the kept $10^4$ posterior samples of the parameters $c,\kappa, U_{1:T}(\bm{s}_{1:m})$ for the two hyperparameter specifications (\texttt{hypara3} and \texttt{hypara4}) when the true parameter values are $(\kappa,c) = (0.35, 2000)$ and the spatial locations are undirected from the second simulation group.}
    			\label{undirectedkappa0.35c2000}
    		}
    	\end{table}
    	
    	\begin{table}[h]
    		{
    			\begin{adjustwidth}{-0cm}{-0cm}
    				\begin{center}
    					\scalebox{0.95}{\begin{tabular}{|*{8}{c|}}
    							\hline
    							\multicolumn{2}{|c|}{} & Minimum & 1st Quantile & Median & Mean & 3rd Quantile & Maximum \\
    							\hline
    							\multicolumn{2}{|c|}{Actual $y_t(\bm{s}_i)$'s} & 0  &  2407 &   5694  &  7814  & 10878  & 58052  \\
    							\hline
    							\multirow{2}{*}{$U_t(\bm{s}_i)$'s} & \texttt{hypara3} & 0  &  2403  &  5695  &  7814 &  10881  & 59029     \\
    							\cline{2-8}
    							& \texttt{hypara4} & 0  &  2403  &  5695  &  7814  & 10881 &  58955 \\
    							\hline
    							\multirow{2}{*}{$c$} & \texttt{hypara3} & 4575  &  4884 &   4964  &  4967  &  5045  &  5428   \\
    							\cline{2-8}
    							& \texttt{hypara4} & 4533  &  4879  &  4958  &  4959  &  5036 &   5496\\
    							\hline
    							\multirow{2}{*}{$\kappa$} & \texttt{hypara3} & 0.3060 & 0.3574 & 0.3677 & 0.3674 & 0.3778 & 0.4168      \\
    							\cline{2-8}
    							& \texttt{hypara4} & 0.2977 & 0.3577 & 0.3681 & 0.3682 & 0.3787 & 0.4184\\
    							\hline
    					\end{tabular}}
    				\end{center}
    			\end{adjustwidth}
    			\caption{Summary statistics of the actual observed counts $y_{1:T}(\bm{s}_{1:m})$ and the kept $10^4$ posterior samples of the parameters $c,\kappa, U_{1:T}(\bm{s}_{1:m})$ for the two hyperparameter specifications (\texttt{hypara3} and \texttt{hypara4}) when the true parameter values are $(\kappa,c) = (0.35, 5000)$ and the spatial locations are undirected from the second simulation group.}
    			\label{undirectedkappa0.35c5000}
    		}
    	\end{table}
    	\clearpage    
    	
    	\begin{table}[H]
    		{
    			\begin{adjustwidth}{-0cm}{-0cm}
    				\begin{center}
    					\scalebox{0.95}{\begin{tabular}{|*{8}{c|}}
    							\hline
    							\multicolumn{2}{|c|}{} & Minimum & 1st Quantile & Median & Mean & 3rd Quantile & Maximum \\
    							\hline
    							\multicolumn{2}{|c|}{Actual $y_t(\bm{s}_i)$'s} & 0.00  &  5.00  & 11.00 &  13.73 &  19.00  & 74.00   \\
    							\hline
    							\multirow{2}{*}{$U_t(\bm{s}_i)$'s} & \texttt{hypara3} & 0.000 &  5.488 & 11.206 & 13.734 & 19.178 & 101.142      \\
    							\cline{2-8}
    							& \texttt{hypara4} & 0.000 &  5.486 & 11.206 & 13.734 & 19.180 & 100.537 \\
    							\hline
    							\multirow{2}{*}{$c$} & \texttt{hypara3} & 4.514  & 4.750  & 4.809 &  4.808  & 4.865 &  5.178    \\
    							\cline{2-8}
    							& \texttt{hypara4} & 4.535  & 4.769 &  4.826  & 4.828 &  4.885 &  5.173 \\
    							\hline
    							\multirow{2}{*}{$\kappa$} & \texttt{hypara3} & 0.6746 & 0.7113 & 0.7176 & 0.7175 & 0.7237 & 0.7505      \\
    							\cline{2-8}
    							& \texttt{hypara4} &  0.6754 & 0.7102 & 0.7165 & 0.7164 & 0.7226 & 0.7490 \\
    							\hline
    					\end{tabular}}
    				\end{center}
    			\end{adjustwidth}
    			\caption{Summary statistics of the actual observed counts $y_{1:T}(\bm{s}_{1:m})$ and the kept $10^4$ posterior samples of the parameters $c,\kappa, U_{1:T}(\bm{s}_{1:m})$ for the two hyperparameter specifications (\texttt{hypara3} and \texttt{hypara4}) when the true parameter values are $(\kappa,c) = (0.7, 5)$ and the spatial locations are directed from the third simulation group.}
    			\label{directedkappa0.7c5}
    		}
    	\end{table}
    	
    	\begin{table}[H]
    		{
    			\begin{adjustwidth}{-0cm}{-0cm}
    				\begin{center}
    					\scalebox{0.95}{\begin{tabular}{|*{8}{c|}}
    							\hline
    							\multicolumn{2}{|c|}{} & Minimum & 1st Quantile & Median & Mean & 3rd Quantile & Maximum \\
    							\hline
    							\multicolumn{2}{|c|}{Actual $y_t(\bm{s}_i)$'s} & 0.00  & 11.00 &  22.00 &  27.51 &  39.00 & 166.00 \\
    							\hline
    							\multirow{2}{*}{$U_t(\bm{s}_i)$'s} & \texttt{hypara3} & 0.00 &  10.95 &  22.37  & 27.51  & 38.53 & 208.85      \\
    							\cline{2-8}
    							& \texttt{hypara4} & 0.00  & 10.95  & 22.37 &  27.51 &  38.53 & 212.36 \\
    							\hline
    							\multirow{2}{*}{$c$} & \texttt{hypara3} & 9.175  & 9.616  & 9.729  & 9.725  & 9.833 & 10.395   \\
    							\cline{2-8}
    							& \texttt{hypara4} & 9.164 &  9.638  & 9.744 &  9.747 &  9.856 & 10.338 \\
    							\hline
    							\multirow{2}{*}{$\kappa$} & \texttt{hypara3} & 0.6868 & 0.7117 & 0.7176 & 0.7175 & 0.7234 & 0.7511      \\
    							\cline{2-8}
    							& \texttt{hypara4} & 0.6846 & 0.7111 & 0.7171 & 0.7171 & 0.7230 & 0.7488\\
    							\hline
    					\end{tabular}}
    				\end{center}
    			\end{adjustwidth}
    			\caption{Summary statistics of the actual observed counts $y_{1:T}(\bm{s}_{1:m})$ and the kept $10^4$ posterior samples of the parameters $c,\kappa, U_{1:T}(\bm{s}_{1:m})$ for the two hyperparameter specifications (\texttt{hypara3} and \texttt{hypara4}) when the true parameter values are $(\kappa,c) = (0.7, 10)$ and the spatial locations are directed from the third simulation group.}
    			\label{directedkappa0.7c10}
    		}
    	\end{table}
    	
    	\begin{table}[h]
    		{
    			\begin{adjustwidth}{-0cm}{-0cm}
    				\begin{center}
    					\scalebox{0.95}{\begin{tabular}{|*{8}{c|}}
    							\hline
    							\multicolumn{2}{|c|}{} & Minimum & 1st Quantile & Median & Mean & 3rd Quantile & Maximum \\
    							\hline
    							\multicolumn{2}{|c|}{Actual $y_t(\bm{s}_i)$'s} & 0.00 &  16.00 &  34.00  & 41.21  & 58.00 & 240.00  \\
    							\hline
    							\multirow{2}{*}{$U_t(\bm{s}_i)$'s} & \texttt{hypara3} & 0.00  & 16.30  & 33.70 &  41.21 &  57.78 & 292.69       \\
    							\cline{2-8}
    							& \texttt{hypara4} &  0.00 &  16.30 &  33.70  & 41.22 &  57.78 & 293.47 \\
    							\hline
    							\multirow{2}{*}{$c$} & \texttt{hypara3} & 13.81  & 14.58 &  14.74 &  14.74 &  14.91  & 15.75  \\
    							\cline{2-8}
    							& \texttt{hypara4} & 13.95 &  14.59  & 14.75 &  14.75 &  14.91 &  15.71 \\
    							\hline
    							\multirow{2}{*}{$\kappa$} & \texttt{hypara3} & 0.6768 & 0.7044 & 0.7104 & 0.7103 & 0.7162 & 0.7430     \\
    							\cline{2-8}
    							& \texttt{hypara4} &  0.6730 & 0.7040 & 0.7100 & 0.7099 & 0.7159 & 0.7409\\
    							\hline
    					\end{tabular}}
    				\end{center}
    			\end{adjustwidth}
    			\caption{Summary statistics of the actual observed counts $y_{1:T}(\bm{s}_{1:m})$ and the kept $10^4$ posterior samples of the parameters $c,\kappa, U_{1:T}(\bm{s}_{1:m})$ for the two hyperparameter specifications (\texttt{hypara3} and \texttt{hypara4}) when the true parameter values are $(\kappa,c) = (0.7, 15)$ and the spatial locations are directed from the third simulation group.}
    			\label{directedkappa0.7c15}
    		}
    	\end{table}
    	
    	\begin{table}[h]
    		{
    			\begin{adjustwidth}{-0cm}{-0cm}
    				\begin{center}
    					\scalebox{0.95}{\begin{tabular}{|*{8}{c|}}
    							\hline
    							\multicolumn{2}{|c|}{} & Minimum & 1st Quantile & Median & Mean & 3rd Quantile & Maximum \\
    							\hline
    							\multicolumn{2}{|c|}{Actual $y_t(\bm{s}_i)$'s} & 0.00 &  21.00 &  44.00  & 54.96  & 77.00 & 342.00 \\
    							\hline
    							\multirow{2}{*}{$U_t(\bm{s}_i)$'s} & \texttt{hypara3} & 0.00  & 21.82 &  44.82  & 54.96 &  76.92 & 413.15  \\
    							\cline{2-8}
    							& \texttt{hypara4} &  0.00 &  21.82 &  44.82 &  54.96 &  76.93 & 414.19\\
    							\hline
    							\multirow{2}{*}{$c$} & \texttt{hypara3} & 18.38  & 19.27 &  19.48  & 19.49 &  19.70  & 20.71   \\
    							\cline{2-8}
    							& \texttt{hypara4} & 18.22  & 19.29 &  19.50 &  19.50 &  19.72  & 20.60\\
    							\hline
    							\multirow{2}{*}{$\kappa$} & \texttt{hypara3} & 0.6856 & 0.7090 & 0.7150 & 0.7149 & 0.7207 & 0.7489  \\
    							\cline{2-8}
    							& \texttt{hypara4} &  0.6829 & 0.7089 & 0.7146 & 0.7147 & 0.7206 & 0.7488 \\
    							\hline
    					\end{tabular}}
    				\end{center}
    			\end{adjustwidth}
    			\caption{Summary statistics of the actual observed counts $y_{1:T}(\bm{s}_{1:m})$ and the kept $10^4$ posterior samples of the parameters $c,\kappa, U_{1:T}(\bm{s}_{1:m})$ for the two hyperparameter specifications (\texttt{hypara3} and \texttt{hypara4}) when the true parameter values are $(\kappa,c) = (0.7, 20)$ and the spatial locations are directed from the third simulation group.}
    			\label{directedkappa0.7c20}
    		}
    	\end{table}
    	
    	\begin{table}[h]
    		{
    			\begin{adjustwidth}{-0cm}{-0cm}
    				\begin{center}
    					\scalebox{0.95}{\begin{tabular}{|*{8}{c|}}
    							\hline
    							\multicolumn{2}{|c|}{} & Minimum & 1st Quantile & Median & Mean & 3rd Quantile & Maximum \\
    							\hline
    							\multicolumn{2}{|c|}{Actual $y_t(\bm{s}_i)$'s} & 0.0  &  54.0  & 112.0 &  137.5  & 192.0  & 822.0 \\
    							\hline
    							\multirow{2}{*}{$U_t(\bm{s}_i)$'s} & \texttt{hypara3} & 0.00 &  54.77 & 112.14 & 137.52 & 192.49 & 920.54      \\
    							\cline{2-8}
    							& \texttt{hypara4} & 0.00 &  54.76 & 112.13 & 137.52 & 192.49 & 918.93 \\
    							\hline
    							\multirow{2}{*}{$c$} & \texttt{hypara3} & 45.74  & 48.16  & 48.72 &  48.72 &  49.26 &  51.92  \\
    							\cline{2-8}
    							& \texttt{hypara4} & 45.70 &  48.13  & 48.63 &  48.64 &  49.17  & 51.68\\
    							\hline
    							\multirow{2}{*}{$\kappa$} & \texttt{hypara3} & 0.6830 & 0.7089 & 0.7147 & 0.7147 & 0.7205 & 0.7479     \\
    							\cline{2-8}
    							& \texttt{hypara4} & 0.6840 & 0.7093 & 0.7151 & 0.7152 & 0.7209 & 0.7528 \\
    							\hline
    					\end{tabular}}
    				\end{center}
    			\end{adjustwidth}
    			\caption{Summary statistics of the actual observed counts $y_{1:T}(\bm{s}_{1:m})$ and the kept $10^4$ posterior samples of the parameters $c,\kappa, U_{1:T}(\bm{s}_{1:m})$ for the two hyperparameter specifications (\texttt{hypara3} and \texttt{hypara4}) when the true parameter values are $(\kappa,c) = (0.7, 50)$ and the spatial locations are directed from the third simulation group.}
    			\label{directedkappa0.7c50}
    		}
    	\end{table}
    	
    	\begin{table}[h]
    		{
    			\begin{adjustwidth}{-0cm}{-0cm}
    				\begin{center}
    					\scalebox{0.95}{\begin{tabular}{|*{8}{c|}}
    							\hline
    							\multicolumn{2}{|c|}{} & Minimum & 1st Quantile & Median & Mean & 3rd Quantile & Maximum \\
    							\hline
    							\multicolumn{2}{|c|}{Actual $y_t(\bm{s}_i)$'s} & 0.0  & 108.0 &  225.0  & 274.9  & 385.0 & 1591.0 \\
    							\hline
    							\multirow{2}{*}{$U_t(\bm{s}_i)$'s} & \texttt{hypara3} & 0.0 &  108.8  & 225.3  & 274.9  & 384.7 & 1764.8      \\
    							\cline{2-8}
    							& \texttt{hypara4} & 0.0 &  108.8  & 225.3  & 274.9  & 384.7 & 1746.9 \\
    							\hline
    							\multirow{2}{*}{$c$} & \texttt{hypara3} & 92.21  & 96.38 &  97.42  & 97.48 &  98.55 & 105.09   \\
    							\cline{2-8}
    							& \texttt{hypara4} & 91.77  & 96.43 &  97.55 &  97.54 &  98.60 & 102.89 \\
    							\hline
    							\multirow{2}{*}{$\kappa$} & \texttt{hypara3} & 0.6791 & 0.7084 & 0.7143 & 0.7142 & 0.7201 & 0.7466     \\
    							\cline{2-8}
    							& \texttt{hypara4} &  0.6825 & 0.7079 & 0.7139 & 0.7137 & 0.7195 & 0.7464 \\
    							\hline
    					\end{tabular}}
    				\end{center}
    			\end{adjustwidth}
    			\caption{Summary statistics of the actual observed counts $y_{1:T}(\bm{s}_{1:m})$ and the kept $10^4$ posterior samples of the parameters $c,\kappa, U_{1:T}(\bm{s}_{1:m})$ for the two hyperparameter specifications (\texttt{hypara3} and \texttt{hypara4}) when the true parameter values are $(\kappa,c) = (0.7, 100)$ and the spatial locations are directed from the third simulation group.}
    			\label{directedkappa0.7c100}
    		}
    	\end{table}
    	
    	\begin{table}[h]
    		{
    			\begin{adjustwidth}{-0cm}{-0cm}
    				\begin{center}
    					\scalebox{0.95}{\begin{tabular}{|*{8}{c|}}
    							\hline
    							\multicolumn{2}{|c|}{} & Minimum & 1st Quantile & Median & Mean & 3rd Quantile & Maximum \\
    							\hline
    							\multicolumn{2}{|c|}{Actual $y_t(\bm{s}_i)$'s} &  0   &  550  &  1122  &  1375  &  1921  &  8447  \\
    							\hline
    							\multirow{2}{*}{$U_t(\bm{s}_i)$'s} & \texttt{hypara3} & 0.0 &  549.3 & 1122.8 & 1374.8 & 1922.7 & 8856.2       \\
    							\cline{2-8}
    							& \texttt{hypara4} & 0.0 &  549.4 & 1122.8 & 1374.8 & 1922.7 & 8787.0 \\
    							\hline
    							\multirow{2}{*}{$c$} & \texttt{hypara3} & 460.3  & 483.0 &  488.0 &  488.1  & 493.2 &  514.5   \\
    							\cline{2-8}
    							& \texttt{hypara4} & 463.1  & 483.3 &  488.4 &  488.6  & 493.9 &  519.5 \\
    							\hline
    							\multirow{2}{*}{$\kappa$} & \texttt{hypara3} & 0.6815 & 0.7084 & 0.7142 & 0.7141 & 0.7200 & 0.7504     \\
    							\cline{2-8}
    							& \texttt{hypara4} & 0.6815 & 0.7083 & 0.7139 & 0.7138 & 0.7195 & 0.7452 \\
    							\hline
    					\end{tabular}}
    				\end{center}
    			\end{adjustwidth}
    			\caption{Summary statistics of the actual observed counts $y_{1:T}(\bm{s}_{1:m})$ and the kept $10^4$ posterior samples of the parameters $c,\kappa, U_{1:T}(\bm{s}_{1:m})$ for the two hyperparameter specifications (\texttt{hypara3} and \texttt{hypara4}) when the true parameter values are $(\kappa,c) = (0.7, 500)$ and the spatial locations are directed from the third simulation group.}
    			\label{directedkappa0.7c500}
    		}
    	\end{table}
    	
    	\begin{table}[h]
    		{
    			\begin{adjustwidth}{-0cm}{-0cm}
    				\begin{center}
    					\scalebox{0.95}{\begin{tabular}{|*{8}{c|}}
    							\hline
    							\multicolumn{2}{|c|}{} & Minimum & 1st Quantile & Median & Mean & 3rd Quantile & Maximum \\
    							\hline
    							\multicolumn{2}{|c|}{Actual $y_t(\bm{s}_i)$'s} & 0  &  1095  &  2252   & 2750  &  3850  & 16684  \\
    							\hline
    							\multirow{2}{*}{$U_t(\bm{s}_i)$'s} & \texttt{hypara3} & 0 &   1096  &  2249  &  2750  &  3843 &  17186       \\
    							\cline{2-8}
    							& \texttt{hypara4} & 0  &  1096 &   2249  &  2750  &  3843 &  17235 \\
    							\hline
    							\multirow{2}{*}{$c$} & \texttt{hypara3} & 923.3  & 967.5 &  977.7  & 978.1 &  988.6 & 1036.0   \\
    							\cline{2-8}
    							& \texttt{hypara4} & 919.6  & 967.1  & 977.7  & 978.0  & 988.8 & 1033.4  \\
    							\hline
    							\multirow{2}{*}{$\kappa$} & \texttt{hypara3} & 0.6802 & 0.7080 & 0.7139 & 0.7137 & 0.7196 & 0.7459    \\
    							\cline{2-8}
    							& \texttt{hypara4} &  0.6827 & 0.7079 & 0.7137 & 0.7136 & 0.7195 & 0.7443 \\
    							\hline
    					\end{tabular}}
    				\end{center}
    			\end{adjustwidth}
    			\caption{Summary statistics of the actual observed counts $y_{1:T}(\bm{s}_{1:m})$ and the kept $10^4$ posterior samples of the parameters $c,\kappa, U_{1:T}(\bm{s}_{1:m})$ for the two hyperparameter specifications (\texttt{hypara3} and \texttt{hypara4}) when the true parameter values are $(\kappa,c) = (0.7, 1000)$ and the spatial locations are directed from the third simulation group.}
    			\label{directedkappa0.7c1000}
    		}
    	\end{table}
    	
    	\begin{table}[h]
    		{
    			\begin{adjustwidth}{-0cm}{-0cm}
    				\begin{center}
    					\scalebox{0.95}{\begin{tabular}{|*{8}{c|}}
    							\hline
    							\multicolumn{2}{|c|}{} & Minimum & 1st Quantile & Median & Mean & 3rd Quantile & Maximum \\
    							\hline
    							\multicolumn{2}{|c|}{Actual $y_t(\bm{s}_i)$'s} & 0   & 2199 &   4496 &   5499  &  7683  & 33144  \\
    							\hline
    							\multirow{2}{*}{$U_t(\bm{s}_i)$'s} & \texttt{hypara3} & 0 &   2195  &  4495  &  5499  &  7683 &  33899        \\
    							\cline{2-8}
    							& \texttt{hypara4} & 0  &  2195  &  4495  &  5499  &  7683 &  33818 \\
    							\hline
    							\multirow{2}{*}{$c$} & \texttt{hypara3} & 1848   & 1933 &   1954  &  1954  &  1976  &  2082    \\
    							\cline{2-8}
    							& \texttt{hypara4} &  1851  &  1934  &  1955  &  1955  &  1976  &  2060 \\
    							\hline
    							\multirow{2}{*}{$\kappa$} & \texttt{hypara3} & 0.6829 & 0.7083 & 0.7140 & 0.7141 & 0.7200 & 0.7442    \\
    							\cline{2-8}
    							& \texttt{hypara4} & 0.6810 & 0.7080 & 0.7139 & 0.7139 & 0.7198 & 0.7461 \\
    							\hline
    					\end{tabular}}
    				\end{center}
    			\end{adjustwidth}
    			\caption{Summary statistics of the actual observed counts $y_{1:T}(\bm{s}_{1:m})$ and the kept $10^4$ posterior samples of the parameters $c,\kappa, U_{1:T}(\bm{s}_{1:m})$ for the two hyperparameter specifications (\texttt{hypara3} and \texttt{hypara4}) when the true parameter values are $(\kappa,c) = (0.7, 2000)$ and the spatial locations are directed from the third simulation group.}
    			\label{directedkappa0.7c2000}
    		}
    	\end{table}
    	
    	\begin{table}[h]
    		{
    			\begin{adjustwidth}{-0cm}{-0cm}
    				\begin{center}
    					\scalebox{0.95}{\begin{tabular}{|*{8}{c|}}
    							\hline
    							\multicolumn{2}{|c|}{} & Minimum & 1st Quantile & Median & Mean & 3rd Quantile & Maximum \\
    							\hline
    							\multicolumn{2}{|c|}{Actual $y_t(\bm{s}_i)$'s} & 0  &  5498 &  11290 &  13749 &  19200 &  83046   \\
    							\hline
    							\multirow{2}{*}{$U_t(\bm{s}_i)$'s} & \texttt{hypara3} & 0 &   5491 &  11273 &  13749 &  19191 &  84107        \\
    							\cline{2-8}
    							& \texttt{hypara4} & 0  &  5491  & 11273  & 13749  & 19191 &  84140 \\
    							\hline
    							\multirow{2}{*}{$c$} & \texttt{hypara3} & 4627  &  4832 &   4884  &  4884   & 4935  &  5145     \\
    							\cline{2-8}
    							& \texttt{hypara4} & 4576  &  4836  &  4889  &  4889  &  4941 &   5167 \\
    							\hline
    							\multirow{2}{*}{$\kappa$} & \texttt{hypara3} & 0.6812 & 0.7086 & 0.7144 & 0.7143 & 0.7201 & 0.7430    \\
    							\cline{2-8}
    							& \texttt{hypara4} & 0.6780 & 0.7080 & 0.7138 & 0.7138 & 0.7195 & 0.7448 \\
    							\hline
    					\end{tabular}}
    				\end{center}
    			\end{adjustwidth}
    			\caption{Summary statistics of the actual observed counts $y_{1:T}(\bm{s}_{1:m})$ and the kept $10^4$ posterior samples of the parameters $c,\kappa, U_{1:T}(\bm{s}_{1:m})$ for the two hyperparameter specifications (\texttt{hypara3} and \texttt{hypara4}) when the true parameter values are $(\kappa,c) = (0.7, 5000)$ and the spatial locations are directed from the third simulation group.}
    			\label{directedkappa0.7c5000}
    		}
    	\end{table}
    	\clearpage   
    	
    	\begin{table}[H]
    		{
    			\begin{adjustwidth}{-0cm}{-0cm}
    				\begin{center}
    					\scalebox{0.95}{\begin{tabular}{|*{8}{c|}}
    							\hline
    							\multicolumn{2}{|c|}{} & Minimum & 1st Quantile & Median & Mean & 3rd Quantile & Maximum \\
    							\hline
    							\multicolumn{2}{|c|}{Actual $y_t(\bm{s}_i)$'s} & 0.000 &  2.000 &  5.000  & 7.748 & 11.000 & 76.000  \\
    							\hline
    							\multirow{2}{*}{$U_t(\bm{s}_i)$'s} & \texttt{hypara3} & 0.000 &  2.411 &  5.634  & 7.749 & 10.829 & 99.658        \\
    							\cline{2-8}
    							& \texttt{hypara4} & 0.000  & 2.409  & 5.632 &  7.748 & 10.827 & 98.062 \\
    							\hline
    							\multirow{2}{*}{$c$} & \texttt{hypara3} & 4.677 &  5.119 &  5.213  & 5.217  & 5.315 &  5.753     \\
    							\cline{2-8}
    							& \texttt{hypara4} & 4.734  & 5.148  & 5.256 &  5.257 &  5.362  & 5.923\\
    							\hline
    							\multirow{2}{*}{$\kappa$} & \texttt{hypara3} & 0.2690 & 0.3267 & 0.3405 & 0.3399 & 0.3530 & 0.4124  \\
    							\cline{2-8}
    							& \texttt{hypara4} &  0.2582 & 0.3204 & 0.3345 & 0.3343 & 0.3488 & 0.4051 \\
    							\hline
    					\end{tabular}}
    				\end{center}
    			\end{adjustwidth}
    			\caption{Summary statistics of the actual observed counts $y_{1:T}(\bm{s}_{1:m})$ and the kept $10^4$ posterior samples of the parameters $c,\kappa, U_{1:T}(\bm{s}_{1:m})$ for the two hyperparameter specifications (\texttt{hypara3} and \texttt{hypara4}) when the true parameter values are $(\kappa,c) = (0.35, 5)$ and the spatial locations are directed from the third simulation group.}
    			\label{directedkappa0.35c5}
    		}
    	\end{table}
    	
    	\begin{table}[H]
    		{
    			\begin{adjustwidth}{-0cm}{-0cm}
    				\begin{center}
    					\scalebox{0.95}{\begin{tabular}{|*{8}{c|}}
    							\hline
    							\multicolumn{2}{|c|}{} & Minimum & 1st Quantile & Median & Mean & 3rd Quantile & Maximum \\
    							\hline
    							\multicolumn{2}{|c|}{Actual $y_t(\bm{s}_i)$'s} & 0.00  &  4.00  & 11.00  & 15.49  & 22.00 & 143.00  \\
    							\hline
    							\multirow{2}{*}{$U_t(\bm{s}_i)$'s} & \texttt{hypara3} & 0.000 &  4.813 & 11.194 & 15.485 & 21.729 & 178.855       \\
    							\cline{2-8}
    							& \texttt{hypara4} &  0.000  & 4.811 & 11.193 & 15.486 
    							& 21.730 & 180.091 \\
    							\hline
    							\multirow{2}{*}{$c$} & \texttt{hypara3} & 9.513 & 10.231 & 10.410 & 10.414 & 10.589 & 11.351   \\
    							\cline{2-8}
    							& \texttt{hypara4} & 9.613 & 10.273 & 10.452 & 10.459 & 10.637 & 11.510\\
    							\hline
    							\multirow{2}{*}{$\kappa$} & \texttt{hypara3} & 0.2764 & 0.3293 & 0.3408 & 0.3408 & 0.3529 & 0.4064 \\
    							\cline{2-8}
    							& \texttt{hypara4} & 0.2733 & 0.3256 & 0.3384 & 0.3379 & 0.3508 & 0.3991 \\
    							\hline
    					\end{tabular}}
    				\end{center}
    			\end{adjustwidth}
    			\caption{Summary statistics of the actual observed counts $y_{1:T}(\bm{s}_{1:m})$ and the kept $10^4$ posterior samples of the parameters $c,\kappa, U_{1:T}(\bm{s}_{1:m})$ for the two hyperparameter specifications (\texttt{hypara3} and \texttt{hypara4}) when the true parameter values are $(\kappa,c) = (0.35, 10)$ and the spatial locations are directed from the third simulation group.}
    			\label{directedkappa0.35c10}
    		}
    	\end{table}
    	
    	\begin{table}[h]
    		{
    			\begin{adjustwidth}{-0cm}{-0cm}
    				\begin{center}
    					\scalebox{0.95}{\begin{tabular}{|*{8}{c|}}
    							\hline
    							\multicolumn{2}{|c|}{} & Minimum & 1st Quantile & Median & Mean & 3rd Quantile & Maximum \\
    							\hline
    							\multicolumn{2}{|c|}{Actual $y_t(\bm{s}_i)$'s} & 0.00  &  7.00 &  17.00 &  23.27  & 32.00 & 209.00  \\
    							\hline
    							\multirow{2}{*}{$U_t(\bm{s}_i)$'s} & \texttt{hypara3} & 0.000 &  7.196 & 16.840 & 23.275 & 32.375 & 259.715        \\
    							\cline{2-8}
    							& \texttt{hypara4} & 0.000 &  7.195 & 16.839 & 23.274 & 32.371 & 256.028\\
    							\hline
    							\multirow{2}{*}{$c$} & \texttt{hypara3} & 14.47 &  15.48 &  15.74 &  15.75 &  16.02 &  17.90   \\
    							\cline{2-8}
    							& \texttt{hypara4} & 14.05 &  15.51 &  15.77 &  15.78 &  16.04 &  17.52 \\
    							\hline
    							\multirow{2}{*}{$\kappa$} & \texttt{hypara3} & 0.2547 & 0.3253 & 0.3375 & 0.3371 & 0.3493 & 0.3912  \\
    							\cline{2-8}
    							& \texttt{hypara4} & 0.2573 & 0.3237 & 0.3362 & 0.3357 & 0.3482 & 0.4076 \\
    							\hline
    					\end{tabular}}
    				\end{center}
    			\end{adjustwidth}
    			\caption{Summary statistics of the actual observed counts $y_{1:T}(\bm{s}_{1:m})$ and the kept $10^4$ posterior samples of the parameters $c,\kappa, U_{1:T}(\bm{s}_{1:m})$ for the two hyperparameter specifications (\texttt{hypara3} and \texttt{hypara4}) when the true parameter values are $(\kappa,c) = (0.35, 15)$ and the spatial locations are directed from the third simulation group.}
    			\label{directedkappa0.35c15}
    		}
    	\end{table}
    	
    	\begin{table}[h]
    		{
    			\begin{adjustwidth}{-0cm}{-0cm}
    				\begin{center}
    					\scalebox{0.95}{\begin{tabular}{|*{8}{c|}}
    							\hline
    							\multicolumn{2}{|c|}{} & Minimum & 1st Quantile & Median & Mean & 3rd Quantile & Maximum \\
    							\hline
    							\multicolumn{2}{|c|}{Actual $y_t(\bm{s}_i)$'s} & 0.0   &  9.0 &   22.5  &  31.1  &  43.0  & 246.0 \\
    							\hline
    							\multirow{2}{*}{$U_t(\bm{s}_i)$'s} & \texttt{hypara3} & 0.000 &  9.726 & 22.584 & 31.098 & 43.238 & 312.579        \\
    							\cline{2-8}
    							& \texttt{hypara4} &  0.000 &  9.726 & 22.583 & 31.098 & 43.239 & 297.549\\
    							\hline
    							\multirow{2}{*}{$c$} & \texttt{hypara3} & 18.81  & 20.48 &  20.83 &  20.85 &  21.20  & 23.11   \\
    							\cline{2-8}
    							& \texttt{hypara4} & 19.08  & 20.59 &  20.96  & 20.97 &  21.34  & 22.92\\
    							\hline
    							\multirow{2}{*}{$\kappa$} & \texttt{hypara3} & 0.2686 & 0.3313 & 0.3438 & 0.3431 & 0.3555 & 0.4096 \\
    							\cline{2-8}
    							& \texttt{hypara4} & 0.2740 & 0.3264  & 0.3392 & 0.3389 & 0.3516 & 0.4011 \\
    							\hline
    					\end{tabular}}
    				\end{center}
    			\end{adjustwidth}
    			\caption{Summary statistics of the actual observed counts $y_{1:T}(\bm{s}_{1:m})$ and the kept $10^4$ posterior samples of the parameters $c,\kappa, U_{1:T}(\bm{s}_{1:m})$ for the two hyperparameter specifications (\texttt{hypara3} and \texttt{hypara4}) when the true parameter values are $(\kappa,c) = (0.35, 20)$ and the spatial locations are directed from the third simulation group.}
    			\label{directedkappa0.35c20}
    		}
    	\end{table}
    	
    	\begin{table}[h]
    		{
    			\begin{adjustwidth}{-0cm}{-0cm}
    				\begin{center}
    					\scalebox{0.95}{\begin{tabular}{|*{8}{c|}}
    							\hline
    							\multicolumn{2}{|c|}{} & Minimum & 1st Quantile & Median & Mean & 3rd Quantile & Maximum \\
    							\hline
    							\multicolumn{2}{|c|}{Actual $y_t(\bm{s}_i)$'s} & 0.00  & 24.00 &  56.00 &  77.69 & 108.25 & 643.00\\
    							\hline
    							\multirow{2}{*}{$U_t(\bm{s}_i)$'s} & \texttt{hypara3} & 0.00  & 24.27 &  56.22 &  77.69 & 108.40 & 734.59        \\
    							\cline{2-8}
    							& \texttt{hypara4} & 0.00 &  24.27  & 56.22 &  77.69 & 108.39 & 757.83  \\
    							\hline
    							\multirow{2}{*}{$c$} & \texttt{hypara3} & 47.11 &  51.31 &  52.16 &  52.20  & 53.03  & 58.14   \\
    							\cline{2-8}
    							& \texttt{hypara4} & 47.59  & 51.33  & 52.21 &  52.25  & 53.14  & 57.18\\
    							\hline
    							\multirow{2}{*}{$\kappa$} & \texttt{hypara3} & 0.2698 & 0.3305 & 0.3419 & 0.3416 & 0.3531 & 0.4052 \\
    							\cline{2-8}
    							& \texttt{hypara4} & 0.2731 & 0.3288 & 0.3410 & 0.3407 & 0.3529 & 0.4023 \\
    							\hline
    					\end{tabular}}
    				\end{center}
    			\end{adjustwidth}
    			\caption{Summary statistics of the actual observed counts $y_{1:T}(\bm{s}_{1:m})$ and the kept $10^4$ posterior samples of the parameters $c,\kappa, U_{1:T}(\bm{s}_{1:m})$ for the two hyperparameter specifications (\texttt{hypara3} and \texttt{hypara4}) when the true parameter values are $(\kappa,c) = (0.35, 50)$ and the spatial locations are directed from the third simulation group.}
    			\label{directedkappa0.35c50}
    		}
    	\end{table}
    	
    	\begin{table}[h]
    		{
    			\begin{adjustwidth}{-0cm}{-0cm}
    				\begin{center}
    					\scalebox{0.95}{\begin{tabular}{|*{8}{c|}}
    							\hline
    							\multicolumn{2}{|c|}{} & Minimum & 1st Quantile & Median & Mean & 3rd Quantile & Maximum \\
    							\hline
    							\multicolumn{2}{|c|}{Actual $y_t(\bm{s}_i)$'s} & 0    &  48 &    112  &   155   &  215  &  1290 \\
    							\hline
    							\multirow{2}{*}{$U_t(\bm{s}_i)$'s} & \texttt{hypara3} & 0.00 &  48.71 & 112.17 & 154.99 & 215.58 & 1434.44         \\
    							\cline{2-8}
    							& \texttt{hypara4} & 0.00 &  48.71 & 112.17 & 154.98 & 215.58 &1421.74  \\
    							\hline
    							\multirow{2}{*}{$c$} & \texttt{hypara3} & 95.59 & 101.59 & 103.24 & 103.31 & 105.02 & 112.11   \\
    							\cline{2-8}
    							& \texttt{hypara4} & 95.25 & 101.47 & 103.14 & 103.26 & 104.97 & 112.23\\
    							\hline
    							\multirow{2}{*}{$\kappa$} & \texttt{hypara3} & 0.2892 & 0.3352 & 0.3471 & 0.3468 & 0.3585 & 0.4063 \\
    							\cline{2-8}
    							& \texttt{hypara4} & 0.2826 & 0.3352 & 0.3476 & 0.3468 & 0.3585 & 0.4077 \\
    							\hline
    					\end{tabular}}
    				\end{center}
    			\end{adjustwidth}
    			\caption{Summary statistics of the actual observed counts $y_{1:T}(\bm{s}_{1:m})$ and the kept $10^4$ posterior samples of the parameters $c,\kappa, U_{1:T}(\bm{s}_{1:m})$ for the two hyperparameter specifications (\texttt{hypara3} and \texttt{hypara4}) when the true parameter values are $(\kappa,c) = (0.35, 100)$ and the spatial locations are directed from the third simulation group.}
    			\label{directedkappa0.35c100}
    		}
    	\end{table}
    	
    	\begin{table}[h]
    		{
    			\begin{adjustwidth}{-0cm}{-0cm}
    				\begin{center}
    					\scalebox{0.95}{\begin{tabular}{|*{8}{c|}}
    							\hline
    							\multicolumn{2}{|c|}{} & Minimum & 1st Quantile & Median & Mean & 3rd Quantile & Maximum \\
    							\hline
    							\multicolumn{2}{|c|}{Actual $y_t(\bm{s}_i)$'s} & 0.0 &  239.0 &  562.0  & 775.9 & 1078.2 & 6582.0 \\
    							\hline
    							\multirow{2}{*}{$U_t(\bm{s}_i)$'s} & \texttt{hypara3} & 0.0 &  240.5 &  561.0 &  775.9 & 1079.1 & 6936.0         \\
    							\cline{2-8}
    							& \texttt{hypara4} & 0.0 &  240.5 &  561.0 &  775.9 & 1079.1 & 6850.5  \\
    							\hline
    							\multirow{2}{*}{$c$} & \texttt{hypara3} & 480.3  & 512.1 &  520.2  & 520.4 &  528.4 &  565.2   \\
    							\cline{2-8}
    							& \texttt{hypara4} & 478.0  & 511.7 &  520.5 &  520.6 &  529.2  & 572.9 \\
    							\hline
    							\multirow{2}{*}{$\kappa$} & \texttt{hypara3} & 0.2761 & 0.3317 & 0.3425 & 0.3426 & 0.3538 & 0.4001 \\
    							\cline{2-8}
    							& \texttt{hypara4} & 0.2785 & 0.3309 & 0.3427 & 0.3425 & 0.3543 & 0.3981 \\
    							\hline
    					\end{tabular}}
    				\end{center}
    			\end{adjustwidth}
    			\caption{Summary statistics of the actual observed counts $y_{1:T}(\bm{s}_{1:m})$ and the kept $10^4$ posterior samples of the parameters $c,\kappa, U_{1:T}(\bm{s}_{1:m})$ for the two hyperparameter specifications (\texttt{hypara3} and \texttt{hypara4}) when the true parameter values are $(\kappa,c) = (0.35, 500)$ and the spatial locations are directed from the third simulation group.}
    			\label{directedkappa0.35c500}
    		}
    	\end{table}
    	
    	\begin{table}[h]
    		{
    			\begin{adjustwidth}{-0cm}{-0cm}
    				\begin{center}
    					\scalebox{0.95}{\begin{tabular}{|*{8}{c|}}
    							\hline
    							\multicolumn{2}{|c|}{} & Minimum & 1st Quantile & Median & Mean & 3rd Quantile & Maximum \\
    							\hline
    							\multicolumn{2}{|c|}{Actual $y_t(\bm{s}_i)$'s} & 0   &  483  &  1120  &  1552  &  2155  & 13105  \\
    							\hline
    							\multirow{2}{*}{$U_t(\bm{s}_i)$'s} & \texttt{hypara3} & 0.0 &  482.4 & 1120.6 & 1551.5 & 2153.7 & 13526.8        \\
    							\cline{2-8}
    							& \texttt{hypara4} & 0.0  & 482.4 & 1120.6 & 1551.5 & 2153.6 & 13522.2 \\
    							\hline
    							\multirow{2}{*}{$c$} & \texttt{hypara3} & 950.2 & 1024.3 & 1041.2 & 1042.3 & 1060.0 & 1136.8   \\
    							\cline{2-8}
    							& \texttt{hypara4} & 951.4 & 1025.5 & 1043.2 & 1043.9 & 1061.9 & 1147.9\\
    							\hline
    							\multirow{2}{*}{$\kappa$} & \texttt{hypara3} & 0.2836 & 0.3299 & 0.3423 & 0.3417 & 0.3538 & 0.4044 \\
    							\cline{2-8}
    							& \texttt{hypara4} & 0.2788 & 0.3283 & 0.3407 & 0.3404 & 0.3527 & 0.4006\\
    							\hline
    					\end{tabular}}
    				\end{center}
    			\end{adjustwidth}
    			\caption{Summary statistics of the actual observed counts $y_{1:T}(\bm{s}_{1:m})$ and the kept $10^4$ posterior samples of the parameters $c,\kappa, U_{1:T}(\bm{s}_{1:m})$ for the two hyperparameter specifications (\texttt{hypara3} and \texttt{hypara4}) when the true parameter values are $(\kappa,c) = (0.35, 1000)$ and the spatial locations are directed from the third simulation group.}
    			\label{directedkappa0.35c1000}
    		}
    	\end{table}
    	
    	\begin{table}[h]
    		{
    			\begin{adjustwidth}{-0cm}{-0cm}
    				\begin{center}
    					\scalebox{0.95}{\begin{tabular}{|*{8}{c|}}
    							\hline
    							\multicolumn{2}{|c|}{} & Minimum & 1st Quantile & Median & Mean & 3rd Quantile & Maximum \\
    							\hline
    							\multicolumn{2}{|c|}{Actual $y_t(\bm{s}_i)$'s} & 0  &   964 &   2236 &   3103 &   4302 &  26189   \\
    							\hline
    							\multirow{2}{*}{$U_t(\bm{s}_i)$'s} & \texttt{hypara3} & 0.0 &  964.3 & 2237.2 & 3103.5 & 4305.6 & 26774.4         \\
    							\cline{2-8}
    							& \texttt{hypara4} & 0.0  & 964.4 & 2237.2 & 3103.4 & 4305.5 & 26744.6  \\
    							\hline
    							\multirow{2}{*}{$c$} & \texttt{hypara3} & 1887  &  2054 &   2088 &   2089  &  2124  &  2310    \\
    							\cline{2-8}
    							& \texttt{hypara4} & 1904  &  2050  &  2084  &  2086   & 2120  &  2309 \\
    							\hline
    							\multirow{2}{*}{$\kappa$} & \texttt{hypara3} & 0.2706 & 0.3283 & 0.3402 & 0.3399 & 0.3521 & 0.4062  \\
    							\cline{2-8}
    							& \texttt{hypara4} & 0.2638 & 0.3296 & 0.3414 & 0.3412 & 0.3527 & 0.4028 \\
    							\hline
    					\end{tabular}}
    				\end{center}
    			\end{adjustwidth}
    			\caption{Summary statistics of the actual observed counts $y_{1:T}(\bm{s}_{1:m})$ and the kept $10^4$ posterior samples of the parameters $c,\kappa, U_{1:T}(\bm{s}_{1:m})$ for hyperparameter specifications \texttt{hypara3} and \texttt{hypara4} with true parameter values $(\kappa,c) = (0.35, 2000)$ and directed spatial locations from the third simulation group.}
    			\label{directedkappa0.35c2000}
    		}
    	\end{table}
    	
    	\begin{table}[h]
    		{
    			\begin{adjustwidth}{-0cm}{-0cm}
    				\begin{center}
    					\scalebox{0.95}{\begin{tabular}{|*{8}{c|}}
    							\hline
    							\multicolumn{2}{|c|}{} & Minimum & 1st Quantile & Median & Mean & 3rd Quantile & Maximum \\
    							\hline
    							\multicolumn{2}{|c|}{Actual $y_t(\bm{s}_i)$'s} &  0   & 2417 &    5602  &  7758 &  10756  & 65568  \\
    							\hline
    							\multirow{2}{*}{$U_t(\bm{s}_i)$'s} & \texttt{hypara3} & 0  &  2412 &   5604  &  7758  & 10757 &  66557          \\
    							\cline{2-8}
    							& \texttt{hypara4} & 0  &  2412  &  5604  &  7758  & 10757 &  66657 \\
    							\hline
    							\multirow{2}{*}{$c$} & \texttt{hypara3} & 4776   & 5125   & 5213  &  5215 &   5301  &  5720  \\
    							\cline{2-8}
    							& \texttt{hypara4} &  4710  &  5130 &   5214  &  5218  &  5302  &  5710  \\
    							\hline
    							\multirow{2}{*}{$\kappa$} & \texttt{hypara3} & 0.2712 & 0.3294 & 0.3412 & 0.3410 & 0.3536 & 0.4028  \\
    							\cline{2-8}
    							& \texttt{hypara4} & 0.2734 & 0.3294 & 0.3408 & 0.3406 & 0.3524 & 0.3999  \\
    							\hline
    					\end{tabular}}
    				\end{center}
    			\end{adjustwidth}
    			\caption{Summary statistics of the actual observed counts $y_{1:T}(\bm{s}_{1:m})$ and the kept $10^4$ posterior samples of the parameters $c,\kappa, U_{1:T}(\bm{s}_{1:m})$ for hyperparameter specifications \texttt{hypara3} and \texttt{hypara4} with true parameter values $(\kappa,c) = (0.35, 5000)$ and directed spatial locations from the third simulation group.}
    			\label{directedkappa0.35c5000}
    		}
    	\end{table}
    	\clearpage

    	\begin{table}[H]
    		{
    			\begin{adjustwidth}{-0cm}{-0cm}
    				\begin{center}
    					\scalebox{0.92}{\begin{tabular}{|*{7}{c|}}
    							\hline
    							& Minimum & 1st Quantile & Median & Mean & 3rd Quantile & Maximum \\
    							\hline
    							$y_t(\bm{s}_i)$'s &  0.0  &  32.0  &  73.0  &  97.6 &  138.0 &  877.0 \\
    							\hline
    							$U_t(\bm{s}_i)$'s &  0.00  & 32.14  & 73.47 &  97.60 & 137.74 & 975.15 \\
    							\hline
    							$\big|y_t(\bm{s}_i)-U_t(\bm{s}_i)\big|$'s & $1.910\times 10^{-5}$ & 0.4108 & 0.7546 & 0.8700 & 0.9926 & 13.30 \\
    							\hline
    							$c$ &48.80  & 49.76 &  50.00 &  50.00 &  50.25  & 51.45 \\
    							\hline
    							$\kappa$  & 0.3799 & 0.3964 & 0.3996 & 0.3996 & 0.4027 & 0.4139\\
    							\hline
    							$\rho$ & 0.08676 & 0.09662 & 0.09869 & 0.09880 & 0.10089 &0.11156 \\
    							\hline
    					\end{tabular}}
    				\end{center}
    			\end{adjustwidth}
    			\caption{Data simulated from our model with $T=50$ and $m=1600$: summary statistics of the actual observed counts $y_{1:T}(\bm{s}_{1:m})$, absolute error $\big|y_t(\bm{s}_i)-\hat{y}_t(\bm{s}_i)\big|$ for all $(t,i)$, and the kept $5000$ posterior samples of the parameters $c,\kappa, \rho,U_{1:T}(\bm{s}_{1:m})$ for hyperparameter specification \texttt{hypara1}. The true parameter values are $(\kappa,\rho,c) = (0.4, 0.1, 50)$ and the spatial locations are undirected.}
    			\label{m1600T50poisGammaParaEst}
    		}
    	\end{table}
    	
    	\begin{table}[H]
    		{
    			\begin{adjustwidth}{-0cm}{-0cm}
    				\begin{center}
    					\scalebox{0.92}{\begin{tabular}{|*{7}{c|}}
    							\hline
    							& Minimum & 1st Quantile & Median & Mean & 3rd Quantile & Maximum \\
    							\hline
    							$y_t(\bm{s}_i)$'s &   0.00  & 32.00  & 73.00 &  97.95  &138.00 & 827.00 \\
    							\hline
    							$U_t(\bm{s}_i)$'s & 0.00  & 32.39 &  73.24  & 97.95 & 137.65 & 920.08 \\
    							\hline
    							$\big|y_t(\bm{s}_i)-U_t(\bm{s}_i)\big|$'s & 0.000032 & 0.409718 & 0.752823 & 0.871786 & 0.990895 & 11.581919 \\
    							\hline
    							$c$ & 48.59 &  49.97 &  50.30  & 50.30 &  50.61  & 52.00 \\
    							\hline
    							$\kappa$  & 0.3669 & 0.3857 & 0.3900 & 0.3900 & 0.3943 & 0.4112 \\
    							\hline
    							$\rho$ & 0.08974 & 0.10334 & 0.10647 & 0.10644 & 0.10961 & 0.12243  \\
    							\hline
    					\end{tabular}}
    				\end{center}
    			\end{adjustwidth}
    			\caption{Data simulated from our model with $T=50$ and $m=900$: summary statistics of the actual observed counts $y_{1:T}(\bm{s}_{1:m})$, absolute error $\big|y_t(\bm{s}_i)-\hat{y}_t(\bm{s}_i)\big|$ for all $(t,i)$, and the kept $5000$ posterior samples of the parameters $c,\kappa, \rho,U_{1:T}(\bm{s}_{1:m})$ for hyperparameter specification \texttt{hypara1}. The true parameter values are $(\kappa,\rho,c) = (0.4, 0.1, 50)$ and the spatial locations are undirected.}
    			\label{m900T50poisGammaParaEst}
    		}
    	\end{table}
    	\clearpage
        \begin{table}[htb]
        	{
        		\begin{adjustwidth}{-0cm}{-0cm}
        			\begin{center}
        				\scalebox{0.83}{\begin{tabular}{|*{8}{c|}} 
        						\hline
        						\backslashbox{\textbf{Model}}{\textbf{Measure}} & \texttt{runtime} & $\overline{\text{ESS}}$ & MAE & DIC & p.d. & WAIC & p.w. \\
        						\hline
        						\texttt{ST.CARlinear} & 198.8 & 1744.6 & 0.8895834 & 130032.4521 & 726.3246 & 130041.1649 & 721.4034 \\ 
        						\hline
        						our model & 3960 & 1796 & 0.6724177 & 132245.9 & 10784.88 & 146306.7 & 12800.52\\
        						\hline     
        						\texttt{INLA.ST1} & 0.06 & N.A. & 0.8152582 & 121443.7 & 644.4155 & 121442.9 & 634.4881\\
        						\hline     
        						\texttt{INLA.STint} & 0.17 & N.A. & 0.8029276 & 121437.9 & 1296.044 & 121447.7 & 1268.957\\
        						\hline
        						\hline
        						\texttt{ST.CARsepspatial} & 413.4 & 1763.44 & 0.8791573 & 130912.3594  & 584.5215 & 130905.9726 & 591.2164 \\   
        						\hline
        						our model & 2811.6 & $\approx 5000$ & 0.4222 & 137704.1 & 20851.72 & 142312.5 & 13164.35\\
        						\hline     
        						\texttt{INLA.ST1} & 0.05 & N.A. & 0.8629758 & 127409.9 & 332.8639 & 127414.2 & 334.6194\\
        						\hline    
        						\texttt{INLA.STint} & 0.16 & N.A. & 0.8458735 & 127388.5 & 1196.434 & 127411.9 & 1188.314\\
        						\hline
        						\hline
        						\texttt{ST.CARanova} & 340 & 1591.8 & 0.7733028 &  116189.0105 & 392.9050 & 116188.5211 & 388.8913  \\
        						\hline
        						our model & 3888 & 1690.4 & 0.5987155 & 117326.7 & 9040.101 & 129505.7 & 10961.21\\
        						\hline     
        						\texttt{INLA.ST1} & 0.05 & N.A. & 0.7334602 & 110367.8 & 340.7747 & 110369.1 & 339.5231\\
        						\hline 
        						\texttt{INLA.STint} & 0.15 & N.A. & 0.7221786 & 110356.6 & 997.812 & 110371.4 & 991.0257\\
        						\hline
        						\hline
        						\texttt{ST.CARar} & 438.5 & 160.81 & 0.7643704 & 118770.380 & 3364.181 & 118963.927 & 3286.991 \\    
        						\hline
        						our model & 3600 & 2059 & 0.587558 & 118839.3 & 10885.02 & 134954.3 & 13914.49\\
        						\hline     
        						\texttt{INLA.ST1} & 0.05 & N.A. & 0.9683906 & 134942 & 524.072 & 135072.5 & 646.7803\\
        						\hline  
        						\texttt{INLA.STint} & 0.21 & N.A. & 0.7966171 & 133244.5 &  8108.409 & 133422.4 & 7083.456\\
        						\hline
        						\hline
        						\texttt{ST.CARadaptive} & 8394.6 & 553.73 & 2.499 & 257256.44 & 15379.43 & 257808.92 & 12680.56\\     
        						\hline
        						\texttt{ST.CARlocalised} & 2278.8 & 1052.194 & 2.086 & 256728.10  & 16717.58 & 259627.62 & 15125.71\\   
        						\hline
        						our model & 3708 & 4161 & 1.568243 & 268834.9 & 32390.91 & 273852.9 & 19207.85\\ 
        						\hline   
        						\texttt{INLA.STint} & 0.24 & N.A. & 1.578 & 276148.6 & 30308.05 & 272842.8 & 19945.11\\
        						\hline
        				\end{tabular}}
        			\end{center}
        		\end{adjustwidth}
        		\caption{Data simulated from five \texttt{CARBayesST} settings (\texttt{ST.CARadaptive} and \texttt{ST.CARlocalised} share the same simulated data set) with $T=50$ and $m=900$: overall computation time (\texttt{runtime}) in seconds, mean response ESS ($\overline{\text{ESS}}$) at the 5000 kept post-burn-in MCMC iterations, and diagnostics metrics from the corresponding \texttt{CARBayesST} models, our model, and \texttt{INLA.STint}. \texttt{INLA.ST1} performs poorly on the last simulated data set, and that row is thus excluded from the table.}
        		\label{m900T50CARBayesSTtimeDiags}
        	}
        \end{table}
        
        \begin{table}[htb]
        	{
        		\begin{adjustwidth}{-0cm}{-0cm}
        			\begin{center}
        				\scalebox{0.82}{\begin{tabular}{|*{8}{c|}} 
        						\hline
        						\backslashbox{\textbf{Model}}{\textbf{Measure}} & \texttt{runtime} & $\overline{\text{ESS}}$ & MAE & DIC & p.d. & WAIC & p.w. \\
        						\hline
        						our model & 3105 & $\approx 5000$ & 0.871786 & 354027.7 &  43160.52 & 356793.3 & 23600.63\\
        						\hline                     
        						\texttt{ST.CARar} & 454.3 & 1836.4 & 0.726032 & 357016.49 &     43727.37  &   346145.25  &    23725.93 \\
        						\hline                      
        						\texttt{ST.CARsepspatial} & 438.8 & 1837.5 &  0.721146 & 357247.67   &   43862.85  &   346378.54   &   23818.88 \\
        						\hline
        						\texttt{ST.CARadaptive} & 8455.8 & 1832.1 & 0.72530 & 357026.62   &   43734.17  &   346157.38  &    23730.59\\
        						\hline
        						\texttt{ST.CARlocalised} & 2376.7 & 1821.0 & 1.386163 & 330143.28  &    14034.64  &   351824.76   &   25731.95\\
        						\hline
        						\texttt{INLA.STint} & 0.30 & N.A. & 0.7423270 & 358961 & 44208.66 & 352875 & 26861.97\\
        						\hline 
        				\end{tabular}}
        			\end{center}
        		\end{adjustwidth}
        		\caption{Data simulated from our model with $T=50$ and $m=900$: overall computation time (\texttt{runtime}) in seconds, mean response ESS ($\overline{\text{ESS}}$) at the 5000 kept post-burn-in MCMC iterations, and diagnostics metrics from our model and adequate \texttt{CARBayesST}, \texttt{INLA} models. \texttt{ST.CARlinear}, \texttt{ST.CARanova}, and \texttt{INLA.ST1} perform very poorly and are thus not included in the table.}
        		\label{m900T50poisGammaTimeDiags}
        	}
        \end{table}
        \begin{table}[htb]
        	{
        		\begin{adjustwidth}{-0cm}{-0cm}
        			\begin{center}
        				\scalebox{0.8}{\begin{tabular}{|*{7}{c|}} 
        						\hline
        						\backslashbox{\textbf{Model}}{\textbf{Measure}} & \texttt{runtime} & $\overline{\text{ESS}}$ & DIC & p.d. & WAIC & p.w. \\
        						\hline
        						\texttt{ST.CARlinear} & 354.1 & 1654.5 & 180652.004 &     1087.128  &  180662.097   &   1077.577 \\ 
        						\hline
        						our model & 5832 & 1834.3 & 182552.5 & 15187.24 & 203970.8 & 18801.85\\ 
        						\hline     
        						\texttt{INLA.ST1} & 0.11 & N.A. & 180806.4 & 981.5026 & 180809.1 & 970.7971\\
        						\hline     
        						\texttt{INLA.STint} & 0.34 & N.A. & 180785.6 & 1951.996 & 180809.3 & 1922.594\\
        						\hline
        						\hline
        						\texttt{ST.CARanova} & 565.7 & 1342.47 & 185764.3132 &     523.3728 &  185769.0505    &  524.0533\\
        						\hline
        						our model & 5724 & 1990.3 & 188700.5 & 17032.12 & 212612 & 20975.82\\ 
        						\hline     
        						\texttt{INLA.ST1} & 0.08 & N.A. & 185803.7 & 477.9899 & 185806.7 & 477.7393\\
        						\hline 
        						\texttt{INLA.STint} & 0.26 & N.A. & 185784.7 & 1498.866 & 185807.9 & 1491.158\\
        						\hline
        						\hline
        						\texttt{ST.CARar} & 705.5 & 153.72 & 194680.446    &  5827.903 &   194951.276  &    5594.397 \\    
        						\hline
        						our model & 5904 & 1929.8 & 195367.1 & 17007.77 & 220530.1 & 21572.76\\ 
        						\hline     
        						\texttt{INLA.ST1} & 0.08 & N.A. & 204053.4 & 838.319 & 204242.3 & 1015.105\\
        						\hline  
        						\texttt{INLA.STint} & 0.34 & N.A. & 201731.7 & 12100.4 & 202004.7 & 10637.6\\
        						\hline
        				\end{tabular}}
        			\end{center}
        		\end{adjustwidth}
        		\caption{Data simulated from \texttt{ST.CARlinear}, \texttt{ST.CARanova}, and \texttt{ST.CARar} with $T+h=48+2=50$ consecutive time periods and $m+r=1500+21=1521=39^2$ spatial locations: overall computation time (\texttt{runtime}) in seconds, mean response ESS ($\overline{\text{ESS}}$) at the 5000 kept post-burn-in MCMC iterations, and diagnostics metrics from the corresponding \texttt{CARBayesST} models, our model, \texttt{INLA.ST1}, and \texttt{INLA.STint}. }
        		\label{K1521N50CARBayesSTtimeDiags}
        	}
        \end{table}
        \clearpage
    	\begin{figure}[h!]
    		\centering
    		\includegraphics[width=0.9\textwidth]{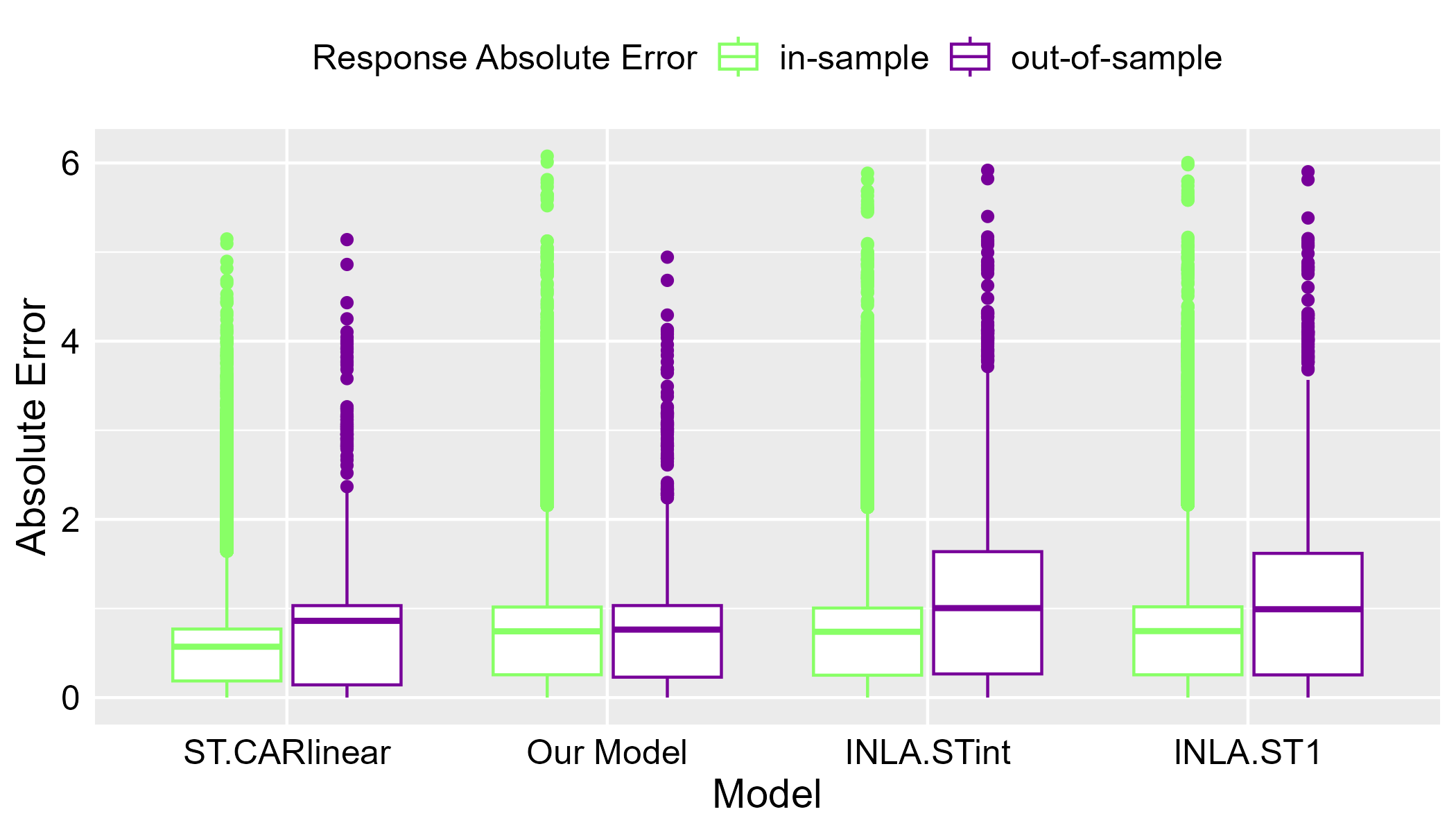}     
    		\caption{Boxplot for data simulated from \texttt{ST.CARlinear} with $50$ consecutive time periods and $m+r=1521=39^2$ spatial locations: distribution of the absolute errors $\big|y_t(\bm{s}_i)-\hat{y}_t(\bm{s}_i)\big|$, where $\hat{y}_t(\bm{s}_i)$ is the posterior mean for each $(t,i)$.}
    		\label{plot:absErrorBoxplotSTCARlinear} 
    	\end{figure}
    	\begin{figure}[h!]
    		\centering
    		\includegraphics[width=0.9\textwidth]{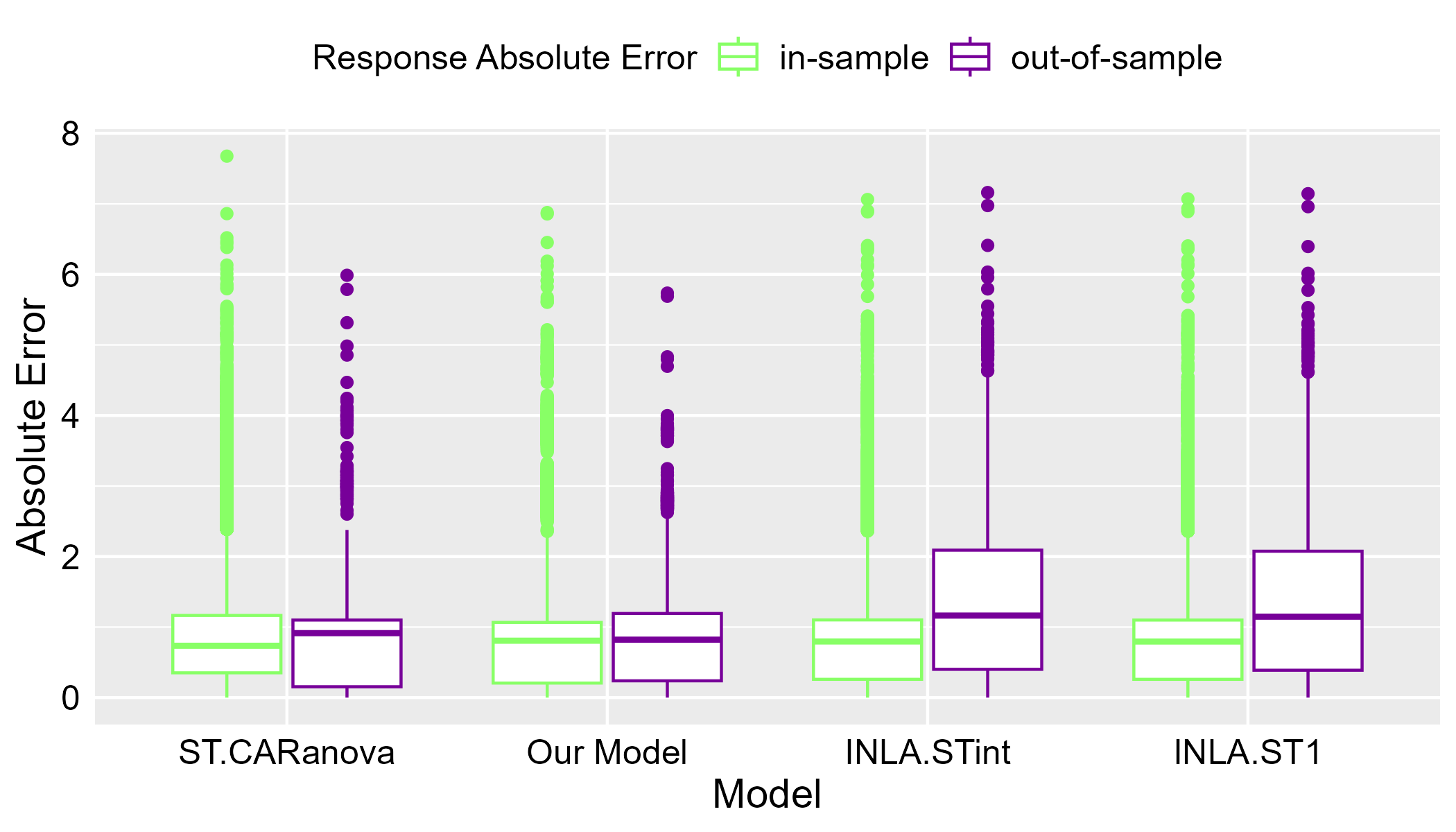}     
    		\caption{Boxplot for data simulated from \texttt{ST.CARanova} with $50$ consecutive time periods and $m+r=1521=39^2$ spatial locations: distribution of the absolute errors $\big|y_t(\bm{s}_i)-\hat{y}_t(\bm{s}_i)\big|$, where $\hat{y}_t(\bm{s}_i)$ is the posterior mean for each $(t,i)$.}
    		\label{plot:absErrorBoxplotSTCARanova} 
    	\end{figure}   
    	
    	\clearpage
    	
    	\section{World Weekly COVID-19 Data Results Complementary  to \Cref{sec:covid}}
    	\begin{table}[htb]
    		{
    			\begin{adjustwidth}{-0cm}{-0cm}
    				\begin{center}
    					\scalebox{0.78}{\begin{tabular}{|*{8}{c|}}
    							\hline
    							\multicolumn{2}{|c|}{} & \textbf{Minimum} & \textbf{1st Quantile} & \textbf{Median} & \textbf{Mean} & \textbf{3rd Quantile} & \textbf{Maximum} \\
    							\hline
    							\multicolumn{2}{|c|}{Actual $y_t(\bm{s}_i)$'s} & 0 &    504  &  3633  & 39190  & 16957 & 5650933 \\
    							\hline
    							$\hat{y}_t(\bm{s}_i)=$ & \texttt{autoregWITS} & 1 &    506  &  3633  & 39190 &  16954 & 5650904\\
    							\cline{2-8}
    							$\bar{U}_t(\bm{s}_i)\times$ & \texttt{autoregCor} & 1 &    506  &  3632 &  39190 &  16956 & 5650902  \\
    							\cline{2-8}
    							\multirow{2}{*}{$e^{\bm{x}_t(\bm{s}_i)\bm{\beta}}$'s} & \texttt{noselfWITS} & 1   &  505  &  3631  & 39190 &  16955 & 5650942 \\
    							\cline{2-8}
    							& \texttt{noselfCor} &  1  &   505  &  3631 &  39190  & 16951 & 5650865 \\
    							\hline
    							\multirow{4}{*}{$c$} & \texttt{autoregWITS} & 0.993 &  1.022  & 1.030  & 1.030  & 1.038 &  1.079  \\
    							\cline{2-8}
    							& \texttt{autoregCor} & 0.9899 & 1.0192 & 1.0271 & 1.0273 & 1.0349 & 1.0710 \\
    							\cline{2-8}
    							& \texttt{noselfWITS} &  4.257 &  4.417  & 4.449  & 4.450 &  4.483 &  4.608 \\
    							\cline{2-8}
    							& \texttt{noselfCor} & 2.714 &  2.803 &  2.824  & 2.825 &  2.846 &  2.950 \\
    							\hline
    							\multirow{4}{*}{$\kappa$} & \texttt{autoregWITS} &  0.000 & $7.257\cdot 10^{-6}$ & $3.122\cdot 10^{-5}$ & $6.216\cdot 10^{-5}$ & $8.260\cdot 10^{-5}$ & $8.317\cdot 10^{-4}$\\
    							\cline{2-8}
    							& \texttt{autoregCor} & $1.000 \cdot 10^{-9}$ & $2.510 \cdot 10^{-3}$ & $3.472\cdot 10^{-3}$ & $3.497\cdot 10^{-3}$ & $4.527\cdot 10^{-3}$ & $8.207\cdot 10^{-3}$ \\
    							\cline{2-8}
    							& \texttt{noselfWITS} & 0.3843 & 0.4033 & 0.4082 & 0.4082 & 0.4130 & 0.4347 \\
    							\cline{2-8}
    							& \texttt{noselfCor} & 0.5009 & 0.5157 & 0.5195 & 0.5195 & 0.5232 & 0.5376\\
    							\hline
    							\multirow{2}{*}{$\rho$} & \texttt{autoregWITS} & 0.8555 & 0.8679 & 0.8708 & 0.8708 & 0.8737 & 0.8865\\
    							\cline{2-8}
    							& \texttt{autoregCor} & 0.8487 & 0.8638 & 0.8669 & 0.8669 & 0.8699 & 0.8847  \\
    							\hline
    					\end{tabular}}
    				\end{center}
    			\end{adjustwidth}
    			\caption{World weekly COVID-19 new cases: summary statistics of the actual observed counts $y_{1:T}(\bm{s}_{1:m})$, the predicted responses $\hat{y}_t(\bm{s}_i)=\bar{U}_t(\bm{s}_i)\times\exp\{\bm{x}_t(\bm{s}_i)\bm{\beta}\}$ for all $(t,i)$, where $\bm{\beta}$ is fixed as described earlier and $\bar{U}_t(\bm{s}_i)$ is the mean of the kept posterior samples of $U_t(\bm{s}_i)$ for each $(t,i)$, and the kept $5000$ posterior samples of $c,\kappa$ for the four settings--\texttt{autoregWITS}, \texttt{autoregCor}, \texttt{noselfWITS}, and \texttt{noselfCor}--and of $\rho$ for the first two settings.}
    			\label{covidCasesPara}
    		}
    	\end{table} 
    	
    	\begin{figure}[ph]
    		\centering
    		\includegraphics[width=\textwidth]{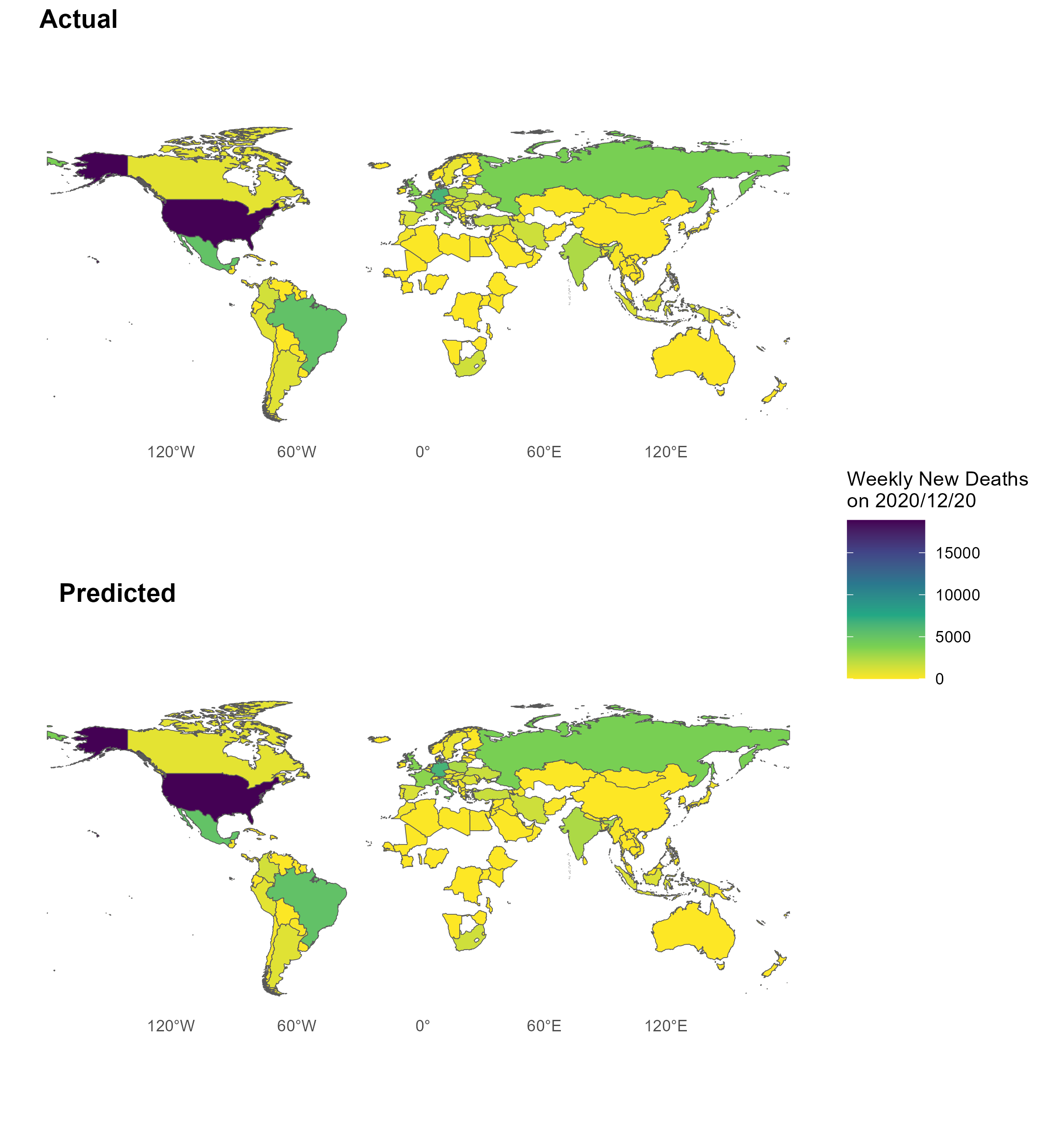}     
    		\caption{Actual and predicted (from the \texttt{noselfWITS} setting) new COVID-19 deaths in the $m=123$ training countries/territories in the week reported on December 20, 2020. The shapefile is obtained from \href{https://public.opendatasoft.com/explore/dataset/world-administrative-boundaries/export/}{Opendatasoft}.}
    		\label{covidDeaths}
    	\end{figure}
    	
    	\begin{table}[ph]
    		{
    			\begin{adjustwidth}{-0cm}{-0cm}
    				\begin{center}
    					\scalebox{0.76}{\begin{tabular}{|*{8}{c|}}
    							\hline
    							\multicolumn{2}{|c|}{} & \textbf{Minimum} & \textbf{1st Quantile} & \textbf{Median} & \textbf{Mean} & \textbf{3rd Quantile} & \textbf{Maximum} \\
    							\hline
    							\multicolumn{2}{|c|}{Actual $y_t(\bm{s}_i)$'s} & 0.0   &  5.0  &  41.0 &  447.6 &  213.0 & 28982.0  \\
    							\hline
    							$\hat{y}_t(\bm{s}_i)=$ & \texttt{autoregWITS} & 0.166  &   4.882 &   41.798 &  447.580  & 213.936 & 28982.182 \\
    							\cline{2-8}
    							$\bar{U}_t(\bm{s}_i)\times$ & \texttt{autoregCor} & 0.164  &   4.889   & 41.723 &  447.585 &  214.045 & 28982.409  \\
    							\cline{2-8}
    							\multirow{2}{*}{$e^{\bm{x}_t(\bm{s}_i)\bm{\beta}}$'s} & \texttt{noselfWITS} & 0.663  &   5.119  &  41.866 &  447.578 &  213.347 & 28979.238  \\
    							\cline{2-8}
    							& \texttt{noselfCor} & 0.541   &  5.034 &   41.828 &  447.581 &  213.384 &28985.861 \\
    							\hline
    							\multirow{4}{*}{$c$} & \texttt{autoregWITS} & 0.4660 & 0.4782 & 0.4807 & 0.4807 & 0.4832 & 0.4950 \\
    							\cline{2-8}
    							& \texttt{autoregCor} & 0.4674 & 0.4780 & 0.4806 & 0.4806 & 0.4831 & 0.4944  \\
    							\cline{2-8}
    							& \texttt{noselfWITS} & 3.433  & 3.549 &  3.578  & 3.578  & 3.605 &  3.719\\
    							\cline{2-8}
    							& \texttt{noselfCor} & 2.968 &  3.071  & 3.096 &  3.095 &  3.119  & 3.262 \\
    							\hline
    							\multirow{4}{*}{$\kappa$} & \texttt{autoregWITS} & 0.000 & $5.475\cdot 10^{-6}$ & $2.171\cdot 10^{-5}$ & $4.362 \cdot 10^{-5}$ & $5.821\cdot 10^{-5}$ & $5.054\cdot 10^{-4}$  \\
    							\cline{2-8}
    							& \texttt{autoregCor} & 0.000 & $6.068\times 10^{-6}$ & $2.501\times 10^{-5}$ & $4.755\times 10^{-5}$ & $6.494\times 10^{-5}$ & $4.973\times 10^{-4}$\\
    							\cline{2-8}
    							& \texttt{noselfWITS} & 0.4529 & 0.4770 & 0.4828 & 0.4827 & 0.4884 & 0.5164 \\
    							\cline{2-8}
    							& \texttt{noselfCor} & 0.4392 & 0.4591 & 0.4633 & 0.4634 & 0.4677 & 0.4889 \\
    							\hline
    							\multirow{2}{*}{$\rho$} & \texttt{autoregWITS} & 0.9162 & 0.9251 & 0.9274 & 0.9274 & 0.9296 & 0.9400  \\
    							\cline{2-8}
    							& \texttt{autoregCor} & 0.9163 & 0.9251 & 0.9274 & 0.9273 & 0.9295 & 0.9385 \\
    							\hline
    					\end{tabular}}
    				\end{center}
    			\end{adjustwidth}
    			\caption{World weekly COVID-19 new deaths: summary statistics of the actual observed counts $y_{1:T}(\bm{s}_{1:m})$, the predicted responses $\hat{y}_t(\bm{s}_i)=\bar{U}_t(\bm{s}_i)\times\exp\{\bm{x}_t(\bm{s}_i)\bm{\beta}\}$ for all $(t,i)$, where $\bm{\beta}$ is fixed as described earlier and $\bar{U}_t(\bm{s}_i)$ is the mean of the kept posterior samples of $U_t(\bm{s}_i)$ for each $(t,i)$, and the kept $5000$ posterior samples of $c,\kappa$ for the four settings--\texttt{autoregWITS}, \texttt{autoregCor}, \texttt{noselfWITS}, and \texttt{noselfCor}--and of $\rho$ for the first two settings.}
    			\label{covidDeathsPara}
    		}
    	\end{table} 
    	
    	\begin{table}[ph]
    		{
    			\begin{adjustwidth}{-0cm}{-0cm}
    				\begin{center}
    					\scalebox{0.91}{\begin{tabular}{|*{7}{c|}}
    							\hline
    							\backslashbox{\textbf{Method}}{\textbf{Metric}} & MAE & MAPE & DIC & p.d. & WAIC & p.w. \\
    							\hline
    							\texttt{autoregWITS} & 1.663053 & 8.244734 & 90033.79 & 9549.816 & 95707.1 & 7880.353\\
    							\hline
    							\texttt{autoregCor} & 1.663616 & 8.238353 & 90036.46 & 9551.499 & 94828.54 & 7452.335\\
    							\hline
    							\texttt{noselfWITS} & 1.208447 & 9.702637 & 91735.95 & 10393.63 & 94901.5 & 7101.675\\
    							\hline
    							\texttt{noselfCor} & 1.188701 & 9.391659 & 91613.89 & 10380.57 & 95029.55 & 7175.584\\
    							\hline
    							\texttt{INLA.STint} & 0.6212796 & 9.18734 & 93428.45 & 12055.36 & 93801.81 & 8511.567\\
    							\hline
    							\texttt{ST.CARarWbinary} & 1.559777 & 12.13594 & 89079.418   &   8702.329 &    87505.799 &     5265.140\\
    							\hline
    							\texttt{ST.CARarWcorr} & 1.535569 & 12.07959 & 88601.097   &   8549.985  &   86870.359  &    5050.861\\
    							\hline
    							\texttt{ST.CARadaptiveWbinary} & 1.912252 & 12.48362 &  88208.038  &    8419.877  &   86989.594   &   5263.732 \\
    							\hline
    							\texttt{ST.CARadaptiveWcorr} & 1.474775 & 11.83246 &  89026.926   &   8886.574     &87508.425    &  5387.956  \\
    							\hline
    							\texttt{ST.CARlocalisedWbinary} & 0.8899884 & 4.244823 & 93041.34   &   11800.16 &    110737.68    &  17358.49 \\
    							\hline
    							\texttt{ST.CARlocalisedWcorr} & 0.8595571 & 3.696582 & 94486.38   &   12652.32  &   122573.85  &    23243.61\\
    							\hline
    					\end{tabular}}
    				\end{center}
    			\end{adjustwidth}
    			\caption{Weekly COVID-19 new deaths: diagnostics metrics from 11 adequate model settings.}
    			\label{covidDeathsDiags}
    		}
    	\end{table}
    	
    	\begin{figure}[ph]
    		\centering
    		\includegraphics[width=\textwidth]{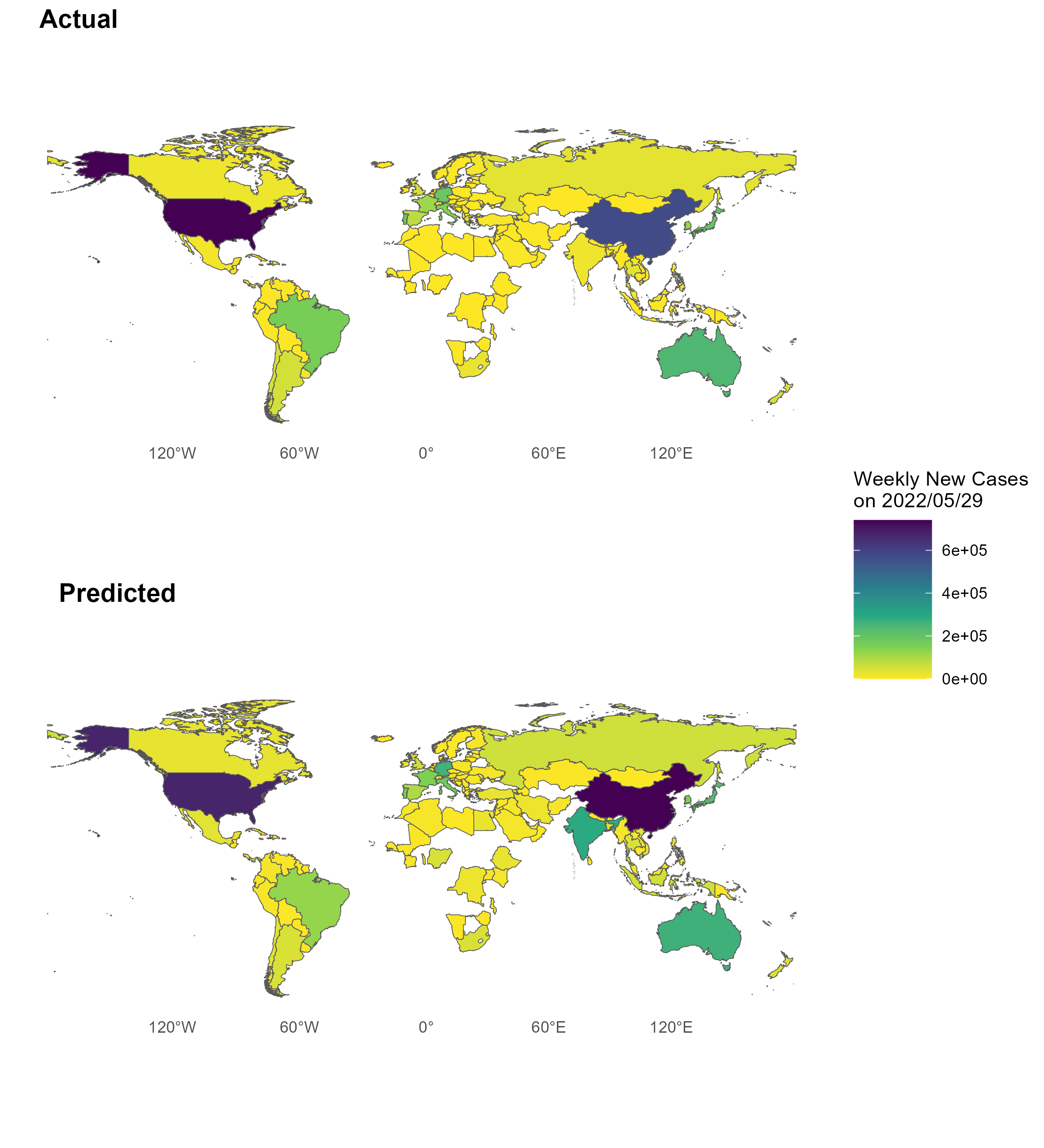}     
    		\caption{Actual and predicted (from the \texttt{autoregWITS} setting) new COVID-19 cases in the $m+r=123+4=127$ training and testing countries/territories in the week reported on May 29, 2022. The shapefile is obtained from \href{https://public.opendatasoft.com/explore/dataset/world-administrative-boundaries/export/}{Opendatasoft}.}
    		\label{covidCasesPredm127NewT1}
    	\end{figure}
    	
    	\begin{figure}[ph]
    		\centering
    		\includegraphics[width=\textwidth]{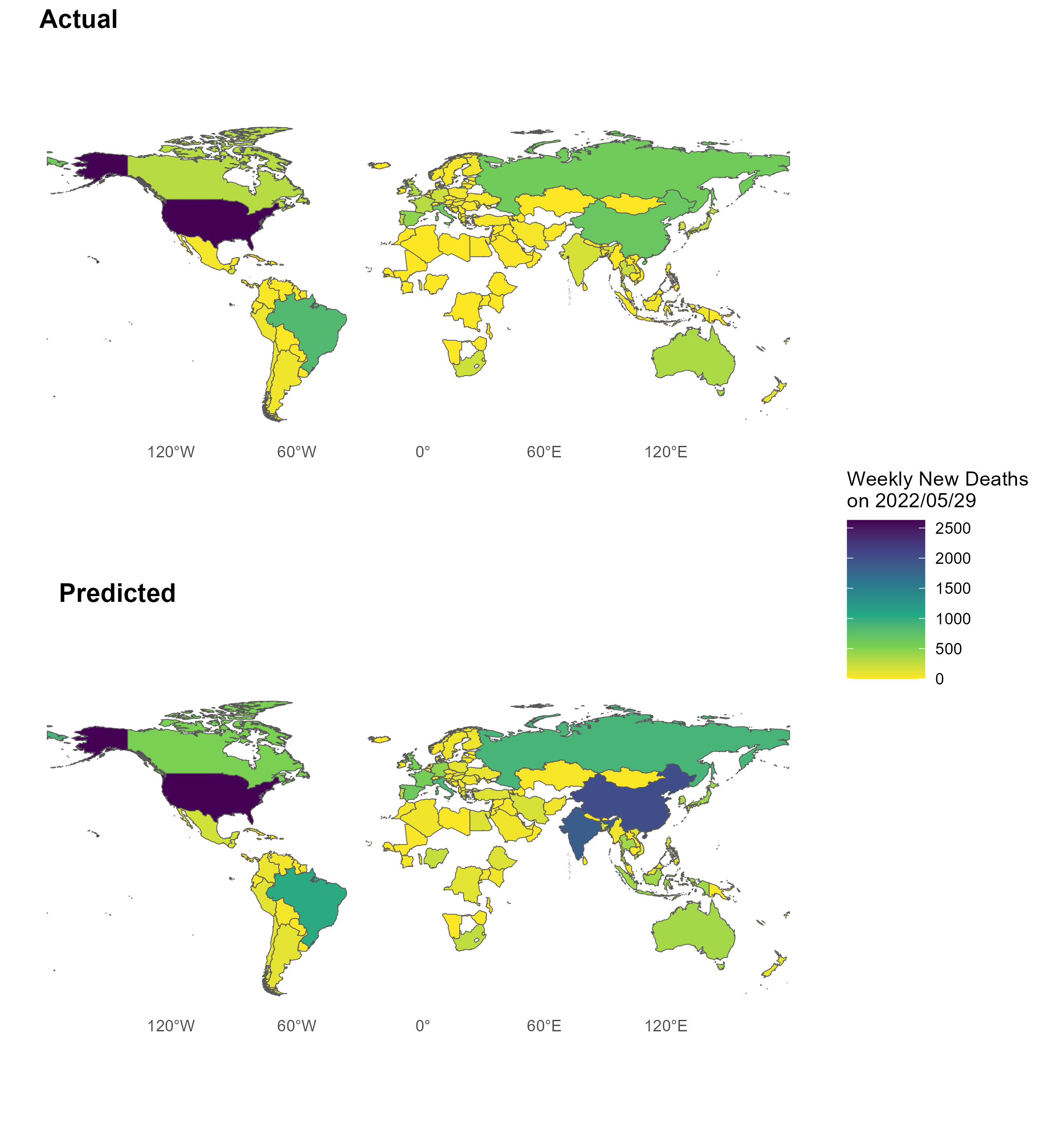}     
    		\caption{Actual and predicted (from the \texttt{autoregWITS} setting) new COVID-19 deaths in the $m+r=123+4=127$ training and testing countries/territories in the week reported on May 29, 2022. The shapefile is obtained from \href{https://public.opendatasoft.com/explore/dataset/world-administrative-boundaries/export/}{Opendatasoft}.}
    		\label{covidDeathsPredm127NewT1}
    	\end{figure}
	\clearpage
	\bibliographystyle{chicago}
	\bibliography{paper2}
\end{document}